\documentclass[reprint,aip,onecolumn,nofootinbib]{revtex4-2}
\usepackage{amssymb}
\usepackage{amsfonts}
\usepackage{amsmath}
\usepackage{graphicx}
\usepackage[ignoreall]{geometry}
\usepackage[colorlinks=true]{hyperref}
\usepackage{enumitem}
\usepackage[all]{xy}
\usepackage{comment}
\usepackage[dvipsnames]{xcolor}
\usepackage{tikz}

\usepackage{dcolumn}% Align table columns on decimal point
\usepackage{bm}% bold math
\usepackage[utf8]{inputenc}
\usepackage[T1]{fontenc}
\usepackage{mathptmx}
\usepackage{etoolbox}

\newtheorem{theorem}{Theorem}

\newtheorem{definition}[theorem]{Definition}

\newtheorem{lemma}[theorem]{Lemma}

\newtheorem{proposition}[theorem]{Proposition}
\newtheorem{remark}[theorem]{Remark}

\newenvironment{proof}[1][Proof]{\noindent\textbf{#1.} }{\ \rule{0.5em}{0.5em}}
\makeatletter
\def\@email#1#2{%
	\endgroup
	\patchcmd{\titleblock@produce}
	{\frontmatter@RRAPformat}
	{\frontmatter@RRAPformat{\produce@RRAP{*#1\href{mailto:#2}{#2}}}\frontmatter@RRAPformat}
	{}{}
}%
\makeatother

\begin{document}

\title[Mathematical foundations for field theories on Finsler spacetimes]{Mathematical foundations for field theories on Finsler spacetimes}

\author{Manuel Hohmann}
\email{manuel.hohmann@ut.ee}
\affiliation{Laboratory of Theoretical Physics, Institute of Physics, University of Tartu, W. Ostwaldi 1, 50411 Tartu, Estonia}
\affiliation{Lepage Research Institute, 17. novembra 1, 08116 Prešov, Slovakia}
\author{Christian Pfeifer}
\email{christian.pfeifer@zarm.uni-bremen.de}
\affiliation{ZARM, University of Bremen, 28359 Bremen, Germany}
\author{Nicoleta Voicu}
\email{nico.voicu@unitbv.ro}
\affiliation{Faculty of Mathematics and Computer Science, Transilvania University, Iuliu Maniu Str. 50, 500091 Brasov, Romania}
\affiliation{Lepage Research Institute, 17. novembra 1, 08116 Prešov, Slovakia}

\begin{abstract}
The paper introduces a general mathematical framework for action based field theories on Finsler spacetimes. As most often fields on Finsler spacetime (e.g., the Finsler fundamental function or the resulting metric tensor) have a homogeneous dependence on the tangent directions of spacetime, we construct the appropriate configuration bundles whose sections are such homogeneous fields; on these configuration bundles, the tools of coordinate free calculus of variations can be consistently applied to obtain field equations. Moreover, we prove that general covariance of natural Finsler field Lagrangians leads to an averaged energy-momentum conservation law which, in the particular case of Lorentzian spacetimes, is equivalent to the usual, pointwise energy-momentum covariant conservation law.
\end{abstract}

\maketitle


\tableofcontents

\section{Introduction}
(Pseudo-)Finsler geometry is the most general geometry admitting a parametrization-invariant arc length of curves. It generalizes Riemannian geometry by using as its fundamental, geometry-defining object, a general line element - which does not necessarily arise as the square root of any quadratic expression in the velocity components, but is just a homogeneous expression of degree one in these. Historically, already Riemmann himself introduced this concept in his habilitation lecture \cite{Riemann1,Riemann2}, however only Finsler investigated it more deeply \cite{Finsler}. Nowadays Finsler geometry is an established field in mathematics \cite{Bao,Bucataru}.

In physics, pseudo-Riemannian geometry is used to describe one of the four fundamental interactions, gravity. In general relativity, gravity is encoded in the Lorentzian geometry of the four-dimensional spacetime manifold, which is determined by the matter content of spacetime via the Einstein equations \cite{EinsteinGR}.
The idea to use geometry based on non-quadratic line elements to describe physical interactions goes far back, at least to Randers \cite{Randers}, who used, in addition to a metric, a $1$-form to search for a unified geometric description of gravity and electromagnetism. Since then, numerous applications of Finsler geometry in physics emerged \cite{Asanov,Pfeifer:2019wus}, for example in the geometric description of fields in media \cite{Cerveny,Klimes,MARKVORSEN2016208,Yajima2009,Perlick,Rubilar:2007qm}, to study non-local Lorentz invariant extensions of fundamental physics \cite{TAVAKOL198523,Tavakol1986,Schreck:2015seb,Kostelecky:2010hs,Kostelecky:2003fs,Kostelecky:2011qz,Bogoslovsky1994,Bogoslovsky,Raetzel:2010je,Amelino-Camelia:2014rga,Lobo:2020qoa}, and to find extensions and modifications of general relativity for an improved description of gravity \cite{Gibbons:2007iu,Lammerzahl:2018lhw,Rutz,Pfeifer-Wohlfarthgravity,kinetic-gas,Minguzzi:2014fxa} that might explain dark matter or dark energy geometrically \cite{Kouretsis:2008ha,Mavromatos:2010nk,Papagiannopoulos:2017whb,Li:2015sja,Hohmann:2016pyt,Saridakis:2021lqd}.

From the mathematical point of view, a major difficulty in the formulation of pseudo-Finsler geometry as generalization of peudo-Riemannian geometry is the existence, in each tangent space, of vectors along which the geometry defining function is either non-smooth or leading to a degenerate metric tensor. One of the first attempts to construct mathematically well defined Lorentz-Finsler spacetimes goes back to Beem \cite{Beem}. It turned out that Beem's definition was to restrictive to cover all cases one is interested in physics and numerous extensions and refinements have been discussed \cite{Asanov,Pfeifer:2011tk,Minguzzi:2014aua,Lammerzahl:2018lhw,Javaloyes:2018lex,Hohmann:2018rpp,Bernal2020}.

On the basis of the improved modern Finsler spacetime definitions, it has recently been suggested that the gravitational field of kinetic gases can be described by Finsler spacetime geometry \cite{kinetic-gas}.

Motivated by the coupling of the kinetic gas to Finsler spacetime geometry and by the recent example of a Finslerian vacuum action, \cite{Hohmann:2018rpp}, we introduce a general framework for mathematically consistent action-based field theories over Finsler spacetimes. The main difficulty for such a construction is that, most often, fields on a Finsler spacetime (e.g., the fundamental function $L$, or the resulting metric tensor) have a non-trivial \textit{homogeneous} dependence on the tangent directions of spacetime. For such homogeneous geometric objects, the naive formulation of the calculus of variations - on configuration manifolds sitting over the tangent bundle - is not well defined, since the variation itself possesses the same homogeneity as the original field; thus, imposing the variation to vanish on the boundary of the arbitrary compact integration domain would typically force the variation to identically vanish, also inside this domain. Therefore, in order to correctly apply the apparatus of the calculus of variations, one must first carefully choose the fiber bundles serving as configuration manifolds for the theory, in such a way as to naturally and consistently accommodate homogeneity. \vspace{12pt}

The main goal of this article is to provide a general construction of configuration bundles, define action integrals for homogeneous fields, in the just explained sense, over Finsler spacetimes and to apply the coordinate free calculus of variations to obtain field equations for the fields in a mathematically rigorous way. Also, we prove that invariance of the corresponding field Lagrangians under lifted spacetime diffeomorphisms leads to an averaged energy-momentum conservation law, which generalizes the (pointwise) energy-momentum conservation law in general relativity. \vspace{12pt}

To achieve this, we proceed as follows:

Section \ref{sec:Fins} reviews the notion of Finsler spacetimes, \cite{cosmoFinsler} and the geometry of Finsler spacetimes, on which our later construction is based. We also give a brief discussion on the different definitions, which can be found in the literature. Most importantly, we discuss the homogeneity properties of the appearing geometric objects.

Section \ref{sec:PTM+} presents the positively projectivized tangent bundle $PTM^+$ (also called in the literature on positive definite Finsler spaces, the \textit{projective sphere bundle}, \cite{Bao}), which will serve as base manifold for action integrals for field theories on Finsler spacetimes. We first discuss in detail the general concept of $PTM^+$ without any geometric fields on the manifold in consideration, and then, how Finsler geometry can be understood on $PTM^+$. Albeit this is still a preparatory topic, it is treated in quite some detail, since a systematic analysis of $PTM^{+}$ and of the various structures it gives rise to, both on general manifolds and on pseudo-Finsler spaces, seems to be missing in the literature.

Having set the stage, we use $PTM^+$ as base manifold for general physical fields having a homogeneous dependence on the direction; these fields are modeled as sections into configuration bundles over $PTM^+$ in Section \ref{sec:Fields}. We introduce the corresponding configuration bundles and fibered automorphisms thereof, that serve in deforming sections and thus give rise to variations.

Eventually, Section \ref{sec:FinslerFieldTheories} combines all the previous concepts to write down the general form of well defined action integrals for homogeneous fields on Finsler spacetime. Once this is done, their field equations are then obtained by the standard techniques of calculus of variations - discussed here in a coordinate-free form.

In Section \ref{sec:EMDistri}, we derive the response of Lagrangians to compactly supported diffeomorphisms on the spacetime manifold, which leads to the novel notion of an \textit{energy-momentum distribution tensor}. It satisfies an averaged covariant conservation law and can be integrated to an energy-momentum tensor \textit{density} on spacetime. Only in very special cases, in particular in the case of a Lorentzian spacetime geometry, this energy-momentum tensor density can be "un-densitized" to yield an energy-momentum tensor on the base manifold.

The necessary notions of geometric calculus of variations (jet bundles over fibered manifolds, fibered automorphisms, the first variation formula in terms of differential forms and their Lie derivatives) are briefly presented in Appendix \ref{app:A}.

\section{Pseudo-Finsler spaces and Finsler spacetime manifolds}

\label{sec:Fins}

We begin this article by presenting a precise definition of Finsler
spacetimes, \cite{cosmoFinsler} on which we will base the presentation and discussion of this
article. We will comment on its relation to other definitions of Finsler
spacetimes given in the literature \cite{Javaloyes:2018lex,Bernal2020,Lammerzahl:2012kw,Hasse:2019zqi}, highlight the importance of the details
in the definition which ensure the existence of a well defined causal
structure and discuss some classes of examples. Moreover, we briefly review the
geometric notions on Finsler spacetimes such as connections and curvature.

The notions presented in this section set the stage for the construction of
action based field theories on Finsler spacetimes.

\subsection{The notion of Finsler spacetime}

\label{sec:defFins}

Let $M$ be a connected, orientable smooth manifold and $TM$, its tangent
bundle with projection $\pi _{TM}:TM\rightarrow M$. We will denote by $x^{i}$ the coordinates in a local chart on $M$ and by
$(x^{i},\dot{x}^{i})$, the naturally induced local coordinates of points $(x,\dot{x}) \in TM$. Whenever there is no risk of confusion, we will omit the indices, i.e., write $(x,\dot{x})$ instead of $(x^{i},\dot{x}^{i})$. Commas $_{,i}$ will mean partial
differentiation with respect to the base coordinates $x^{i}$ and dots $%
_{\cdot i}$ partial differentiation with respect to the fiber coordinates $%
\dot{x}^{i}$. Also, by $\overset{\circ }{TM}=TM\backslash \{0\}$, we will mean the tangent
bundle of $M$ without its zero section.

A \emph{conic subbundle} of $TM$ is a non-empty open submanifold $\mathcal{Q}%
\subset TM\backslash \{0\}$, with the following properties:

\begin{itemize}
\item $\pi_{TM}(\mathcal{Q})=M$;

\item \textit{conic property:} if $(x,\dot{x})\in \mathcal{Q}$, then, for
any $\lambda >0:$ $(x,\lambda \dot{x})\in \mathcal{Q}$.
\end{itemize}

A \textit{pseudo-Finsler space} is, \cite{Bejancu}, a triple $(M,\mathcal{A}%
,L),$ where $M$ is a smooth manifold, $\mathcal{A}\subset ~\overset{\circ }{%
TM}$ is a conic subbundle and $L:\mathcal{A}\rightarrow \mathbb{R}$ is a
smooth function obeying the following conditions:

\begin{enumerate}
\item positive 2-homogeneity:\ $L(x,\alpha \dot{x})=\alpha ^{2}L(x,\dot{x}),$
$\forall \alpha >0,$ $\forall (x,\dot{x})\in \mathcal{A}.$

\item at any $\left( x,\dot{x}\right) \in \mathcal{A}$ and in one (and then,
in any) local chart around $(x,\dot{x}),$ the Hessian:
\begin{equation*}
g_{ij}=\dfrac{1}{2}\dfrac{\partial ^{2}L}{\partial \dot{x}^{i}\partial \dot{x%
}^{j}}
\end{equation*}
is nondegenerate.
\end{enumerate}

The conic subbundle $\mathcal{A},$ where $L$ is defined, smooth and with
nondegenerate Hessian, is called the set of \textit{admissible vectors}. In
the following, we will consider as $\mathcal{A},$ the maximal set with these
properties - hence, we will write simply $(M,L)$ instead of $(M,\mathcal{A}%
,L).$

Another important conic subbundle in a pseudo-Finsler space is the set of
\textit{non-null admissible vectors:}
\begin{equation}
\mathcal{A}_{0}:=\mathcal{A~}\backslash ~L^{-1}\{0\}.  \label{A_0}
\end{equation}%
This is the set where we can divide by $L$ in order to adjust the
homogeneity degree of geometric objects in $\dot{x}$.

\begin{definition}[Finsler spacetimes]\label{def} A Finsler spacetime is a 4-dimensional, connected
pseudo-Finsler space obeying the extra condition:

\begin{enumerate}
\item[3.] There exists a conic subbundle $\mathcal{T}\subset
\mathcal{A}$ with connected fibers $\mathcal{T}_x, x \in M$, such that, on $\mathcal{T}:$ $L>0,$ $g$ has Lorentzian
signature $(+,-,-,-)$ and $L$ can be continuously extended as $0$ to the
boundary $\partial \mathcal{T}.$
\end{enumerate}
\end{definition}
The role of condition $3.$ is to ensure the existence of a proper causal
structure for $(M,L)$.

In the following, though we will not specify this explicitly, we will always consider that $L$ is continuously prolonged as 0 on $\partial \mathcal{T}$; in particular, $L(0)=0$.

The above definition is a minor modification of the one introduced in \cite{cosmoFinsler} to ensure that the cone of timelike vectors is salient at each point. It recovers, in the particular case $\mathcal{A}:=\mathcal{T}$, the definition of improper Finsler spacetimes in \cite{Bernal2020}\footnote{Physically, this ensures the existence of a global time orientation.}. The existence and uniqueness of geodesics with given initial conditions $(x,\dot{x})\in \mathcal{\bar{T}}$, which was explicitly required in an older version of this definition in \cite{Hohmann:2018rpp}, follows from the axioms $1-3$ above, see \cite{Bernal2020}, and thus the definition presented here also covers the Finsler spacetimes discussed in \cite{Lammerzahl:2012kw,Hasse:2019zqi}.

In principle it would be possible to include directions in $\mathcal{T}$
that are not in $\mathcal{A},$ but just in $\overline{\mathcal{A}}$, see e.g.,
\cite{Caponio-Masiello}, \cite{Caponio-Stancarone} yet, for our purposes, it
will be more convenient to assume that $\mathcal{T}\subset \mathcal{A}$, in
order to avoid unnecessary complications in variational procedures involving $\mathcal{T}$ or subsets thereof.

\noindent \textbf{Timelike vectors and the observer space.}

For the application of Finsler spacetimes in physics, besides the sets of
admissible (respectively, non-null admissible) directions $\mathcal{A}$ and $%
\mathcal{A}_{0}$, the following subsets of $TM$ play an important role:

\begin{enumerate}
\item The conic subbundle $\mathcal{T},$ called the set of \textit{future
pointing timelike vectors}.

\item The \textit{observer space}, or set of unit future-pointing timelike directions
\begin{equation}  \label{def observer space}
\mathcal{O}:=\{(x,\dot{x})\in \mathcal{T~}|~L(x,\dot{x})=1\}\,,
\end{equation}
which satisfies the inclusion $\mathcal{O\subset T\subset A}_{0}\subset
\mathcal{A}$. Moreover, due to the homogeneity of $L$, we have at any $x\in
M $:
\begin{equation}  \label{T_O_rel}
\mathcal{T}_{x}=\left( 0,\infty \right) \cdot \mathcal{O}_{x}.
\end{equation}

\item The set $\mathcal{N}:=L^{-1}\{0\}$ has the meaning of set of \textit{%
null} or \textit{lightlike} vectors. By continuously extending $L$ as zero to the boundary $\partial \mathcal{T}$, as specified above, we always have the
inclusion
\begin{equation}
\partial \mathcal{T\subset N}.  \label{eq:nullbdry}
\end{equation}%
It is important to notice that the null cone $\mathcal{N}$ might not be
contained in $\mathcal{A}$, but just in $\overline{\mathcal{A}}$.
\end{enumerate}

%\FRAME{itbpFU}{3.0225in}{2.3324in}{0in}{\Qcb{\textit{admissible vs timelike
%vectors}}}{}{Figure}{\special{language "Scientific Word";type
%"GRAPHIC";maintain-aspect-ratio TRUE;display "USEDEF";valid_file "T";width
%3.0225in;height 2.3324in;depth 0in;original-width 5.4578in;original-height
%4.2013in;cropleft "0";croptop "1";cropright "1";cropbottom "0";tempfilename
%'QSDC7T01.bmp';tempfile-properties "XPR";}}\FRAME{itbpFU}{2.3125in}{2.2381in%
%}{0in}{\Qcb{\textit{observer space}}}{}{Figure}{\special{language
%"Scientific Word";type "GRAPHIC";maintain-aspect-ratio TRUE;display
%"USEDEF";valid_file "T";width 2.3125in;height 2.2381in;depth
%0in;original-width 3.1358in;original-height 3.0338in;cropleft "0";croptop
%"1";cropright "1";cropbottom "0";tempfilename
%'QSDC8F02.bmp';tempfile-properties "XPR";}}

%\begin{figure}[h!]
%	\centering
%	\includegraphics[width=0.5\textwidth]{DoubleConesV2Cut}
%	\caption{The future pointing unit timelike vectors, the observer space and the null directions of a double null-cone Finsler spacetime in the tangent space at $x\in M$.}
%	\label{fig:1}
%\end{figure}

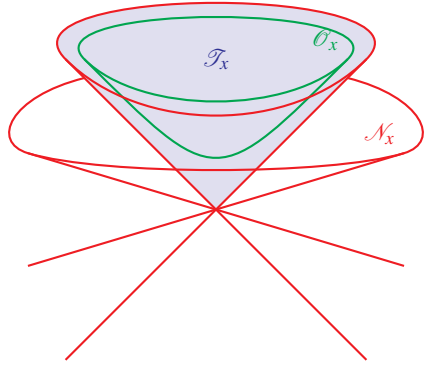
\begin{figure}[h!]
	\centering
\begin{tikzpicture}[thick,scale=0.5]
\fill[Blue!10!White] (-4,4) -- (0,0) -- (4,4) .. controls +(2,2) and +(-2,2) .. (-4,4);
\draw[Red] (-5,-1.5) -- (5,1.5) node [above left] {$\mathcal{N}_x$};
\draw[Red] (-5,1.5) -- (5,-1.5);
\draw[Red] (-4,-4) -- (4,4);
\draw[Red] (-4,4) -- (4,-4);
\draw[Green] (-3.5,4) .. controls +(3.5,-3.5) and +(-3.5,-3.5) .. (3.5,4) node [above left] {$\mathcal{O}_x$};
\draw[Green] (-3.5,4) .. controls +(1.5,-1.5) and +(-1.5,-1.5) .. (3.5,4) .. controls +(1.5,1.5) and +(-1.5,1.5) .. (-3.5,4);
\draw[Red] (-4,4) .. controls +(2,-2) and +(-2,-2) .. (4,4) .. controls +(2,2) and +(-2,2) .. (-4,4);
\draw[Red] (-3.5,3.5) .. controls +(-2,-0.5) and +(-1,0.3) .. (-5,1.5) .. controls +(2,-0.6) and +(-2,-0.6) .. (5,1.5) .. controls +(1,0.3) and +(2,-0.5) .. (3.5,3.5);
\draw[Blue] (0,4) node {$\mathcal{T}_x$};
\end{tikzpicture}
\caption{Future pointing timelike vectors \(\mathcal{T}_x\), observer space \(\mathcal{O}_x\) and null directions \(\mathcal{N}_x\) of a double null-cone Finsler spacetime, in the tangent space at $x\in M$.}
\label{fig:1}
\end{figure}

The null boundary condition \eqref{eq:nullbdry} ensures that, for every $%
x\in M,$ the future timelike cone $\mathcal{T}_{x}$ is actually an entire
connected component of $L^{-1}((0,\infty ))\cap T_{x}M$. This leads to the
following consequences.

\begin{proposition}
At any $x\in M,$ the observer space $\mathcal{O}_{x}$ is a connected component of the indicatrix $I_{x}=L^{-1}(1)\cap T_{x}M$.
\end{proposition}
This statement follows immediately from the connectedness and maximality of $\mathcal{T}_{x}$.

As a consequence of the maximal connectedness of $\mathcal{O}_{x}$, a result
by Beem, \cite{Beem} ensures that $\mathcal{O}_{x}$ is a strictly convex
hypersurface of $T_{x}M$ and moreover, the set
\begin{equation*}
\mathcal{C}_x:=\mathcal{T}_{x}\cap L^{-1}([1,\infty ))=[1,\infty )\cdot
\mathcal{O}_{x}
\end{equation*}%
is also convex. Based on this, we can state:

\begin{proposition}
In a Finsler spacetime as defined above, all future timelike cones $\mathcal{T}_{x},$ $x\in
M, $ are convex.
\end{proposition}

\begin{proof}
Fix $x\in M$ and consider some arbitrary $u,v\in \mathcal{T}_{x},$ $\alpha
\in \lbrack 0,1].$ In order to show that $w:=(1-\alpha )u+\alpha v$ lies in $%
\mathcal{T}_{x},$ we rescale it by $\beta \geq \max
(L(u)^{-1/2},L(v)^{-1/2});$ this way, the endpoints $\beta u,\beta v$ lie in
$\mathcal{C}_x$ and, by the convexity of $\mathcal{C}_x,$ we find that $%
\beta w\in \mathcal{C}_x\subset \mathcal{T}_{x}.$ The statement then
follows from the conicity of $\mathcal{T}_{x}.$
\end{proof}

The above result generalizes a similar statement in \cite{Hohmann:2018rpp}, which was proven under the more restrictive condition that $L$ is defined and continuous on the entire $TM$.
The convexity of the cones $\mathcal{T}_{x}$ ensures a well defined causal
structure on a Finsler spacetime and that timelike geodesics locally extremize the
length between timelike separated points.

The \textit{Finslerian pseudo-norm}, which defines the canonical geometric
length measure for curves on a Finsler spacetime, is defined by
\begin{align}
F:= \sqrt{\left\vert L\right\vert}\,,
\end{align}
which implies
\begin{equation}
L=\epsilon F^{2},~\ \ \epsilon =sign(L).  \label{relation L-F}
\end{equation}

\bigskip

\noindent \textbf{Examples of Finsler spacetimes}

The above definition allows for Finsler spacetimes of:

\begin{enumerate}
\item Lorentzian type\textit{. }If $a:M\rightarrow T_{2}^{0}M,$ $x\mapsto
a_{x}=a_{ij}(x)dx^{i}\otimes dx^{j}$ is a Lorentzian metric on $M$, then, we
can set $\mathcal{A}=\overset{\circ }{TM}$, as\ $L:TM\rightarrow \mathbb{R},$
$L(x,\dot{x})=a_{x}(\dot{x},\dot{x})$ is smooth on $TM.$ Accordingly, $F(x,
\dot{x})=\sqrt{\left\vert a_{ij}(x)\dot{x}^{i}\dot{x}^{j}\right\vert }.$

\item Randers type, \cite{Randers}, given by $F(x,\dot{x})=\sqrt{|a_{ij}(x)\dot{x}^{i}\dot{x}^{j}|}%
+b_{i}(x) \dot{x}^{i}$, where $a$ is as above and $%
b=b_{i}dx^{i}$ is a differential 1-form on $M.$ We proved in \cite%
{Hohmann:2018rpp} that, if $a^{ij}b_{i}b_{j}\in \left( 0,1\right) $ then $F$
provides a Finsler spacetime structure on $M.$ In the context of physics,
these geometries are employed to study the motion of an electrically charged
particle in an electromagnetic field, the propagation of light in static
spacetimes \cite{Werner:2012rc}, Lorentz violating field theories from the standard model extension \cite{Schreck:2015seb,Kostelecky:2012ac,Silva:2020tqr} and Finsler gravitational waves \cite{Heefer:2020hra}. Recently also spinors have been constructed on Randers geometries in terms of Clifford bundles \cite{torrome2021note}.

\item Bogoslovsky/Kropina type $F=(|a_{ij}(x)\dot{x}^{i}\dot{x}^{j}|)^{\frac{1-q%
}{2}}(b_{k}(x)\dot{x}^{k})^{q}$ \cite{Bogoslovsky,Kropina}, where $q\in
\mathbb{R}$; the conditions upon the 1-form $b$, such that $F$ defines a
spacetime structure depend on the value of $q;$ a detailed discussion is
made in \cite{Hohmann:2018rpp}. In physics, approaches to quantum field
theories and modifications of general relativity that are only invariant
under a subgroup of the Lorentz group, based on this geometry, have been
investigated under the name very special and very general relativity \cite%
{Cohen:2006ky,Gibbons:2007iu,Fuster:2015tua,Fuster:2018djw,Elbistan:2020mca}.

\item Polynomial $m$-th root type $F=|G_{a_{1}\cdots a_{m}}(x)\dot{x}%
^{a_{1}}\ldots \dot{x}^{a_{m}}|^{\frac{1}{m}}$, which appear in physics for
example in the description of propagation in birefringent media, in the
context of premetric electrodynamics and in the minimal standard model
extension \cite{Perlick,Rubilar:2007qm,Schreck:2015seb,Gurlebeck:2018nme}.

\item Anisotropic conformal transformations of pseudo-Riemannian geometry $F = e^{\sigma(x,\dot x)}\sqrt{|a_{ij}(x)\dot{x}^{i}\dot{x}^{j}|}$ have been studied in the context of an extension of the Ehlers-Pirani-Schild axiomatic to Finsler geometry \cite{Tavakol1986,Tavakol2009,TAVAKOL198523}, as well as examples for Finsler spacetimes according to Beem's definition \cite{Minguzzi:2014aua}.

\item General first order perturbations of  $F = \sqrt{|a_{ij}(x)\dot{x}^{i}\dot{x}^{j}|} + \epsilon h(x,\dot x)$ pseudo-Riemannian geometry, which are often used in the study of the physical  phenomenology of Planck scale modified dispersion relations \cite{Amelino-Camelia:2014rga,Girelli:2006fw,Letizia:2016lew,Lobo:2020qoa,Raetzel:2010je}.

\item Finsler spacetime metrics $L(x,\dot{x})=\omega_{x}^{2}(\dot{x}) - \hat{F}^{2}(x,\dot{x})$ obtained, \cite{Javaloyes:2018lex}, by "signature-reversing" from positive definite Finsler ones $\hat{F}:TM \rightarrow \mathbb{R}$ using a nonzero 1-form $\omega \in \Omega_{1}(M)$.
Their (2-homogeneous) Finsler function $L$ is well defined and smooth on the entire slit tangent bundle - which is quite a rare feature among pseudo-Finsler metrics. The  set $\mathcal{T}= \{(x,\dot{x}):  \omega_{x}(\dot{x}) \geq \hat{F}(x,\dot{x}) \}$ satisfies the requests for the set of future-pointing timelike vectors.

\end{enumerate}

\subsection{Geometric objects on Finsler spacetimes}

Typical Finslerian objects in a Finsler spacetime $(M,L)$ (more generally,
in a pseudo-Finsler space) are obtained similarly to the corresponding
objects in positive definite Finsler spaces $(M,F)$, see, e.g., \cite%
{Bao,Chern-Chen-Lam,Bucataru}, just taking care that we have to restrict
them to $\mathcal{A}$ or, if necessary, to $\mathcal{A}_{0}$.

Apart from the Finslerian pseudo-norm $F$, the fundamental building blocks
of the geometry of Finsler spacetimes are obtained from partial derivatives
of the Finsler Lagrangian function $L$. Below, we briefly present the coordinate expressions of the typical Finslerian geometric objects to be used in the following.

\noindent \textbf{Hilbert form, Finsler metric tensor and Cartan tensor}

On a Finsler spacetime $(M,L)$ the \textit{Hilbert form} $\omega: \mathcal{A}_{0} \rightarrow T^{0}_{1}M$, the
\textit{Finslerian metric tensor} $g: \mathcal{A} \rightarrow T^{0}_{2}M $ and the \textit{Cartan tensor} $C: \mathcal{A} \rightarrow T^{0}_{3}M $ are
expressed, in every manifold induced local coordinate chart, as
\begin{align}
\omega _{(x,\dot{x})}& := F_{\cdot i}dx^{i}\,, & F_{\cdot i} &=
\epsilon \dfrac{g_{ij}\dot{x}^{j}}{F}\,,  \label{eq:HilbertForm} \\
g_{(x,\dot{x})}& :=g_{ij}(x,\dot{x})dx^{i}\otimes dx^{j}\,, & g_{ij}& :=%
\frac{1}{2}L_{\cdot i\cdot j}  \label{eq:FinsMet} \\
C_{(x,\dot{x})}& :=C_{ijk}(x,\dot{x})dx^{i}\otimes dx^{j}\otimes dx^{k}\,, &
C_{ijk}& :=\frac{1}{4}L_{\cdot i\cdot j\cdot k}\,.  \label{eq:CartanT}
\end{align}%
We note that the Hilbert form $\omega$ is only defined on $\mathcal{A}_{0}$ (as it involves  derivatives of $F=\sqrt{|L|}$, which are not defined at points where $L=0$), while the Finsler metric and the Cartan tensor are defined on $\mathcal{A}$.

A curve $c:[a,b]\rightarrow M$ is called admissible if all its tangent
vectors are in $\mathcal{A}.$ The arc length of a regular admissible curve $%
c:t\in \lbrack a,b]\mapsto c(t)$ on $M$ is calculated as%
\begin{equation}
l(c)=\underset{a}{\overset{b}{\int }}F(c(t),\dot{c}(t))dt,
\label{eq:Flength}
\end{equation}%
where $\dot{c}(t)=\frac{dc}{dt}(t).$
If, moreover, $\dot{c}(t)$ is nowhere lightlike, i.e., $(c(t),\dot{c}(t))\in \mathcal{%
A}_{0}$ for all $t$, then $l(c)$ can also be expressed in terms of the Hilbert form as:
\begin{equation}
l(c)=\underset{a}{\overset{b}{\int }}C^{\ast }\omega =\underset{{Im}C}{\int } \omega ,  \label{arc_length_Hilbert_form}
\end{equation}%
where the symbol $C:[a,b]\rightarrow TM,$ $t \mapsto (c(t),\dot{c}(t))$ denotes here the canonical lift of $c$
to $TM$.

\begin{proposition}[Finsler geodesics], see, e.g., \cite{Bucataru}:
Critical points of the length functional \eqref{eq:Flength} are called Finsler geodesics. In arc
length parametrization, they are determined by the Finsler geodesic equation
\begin{equation}
\ddot{x}^{i}(s)+2G^{i}(x(s),\dot{x}(s))=0,  \label{geodesic_eqn}
\end{equation}%
where $\dot{x}(s)=\frac{dx}{ds}(s)$; the geodesic coefficients are well defined at all points $\left( x,\dot{x}%
\right) \in \mathcal{A}$ and given by
\begin{equation}
2G^{i}(x,\dot{x})=\dfrac{1}{2}g^{ih}(L_{\cdot h,j}\dot{x}^{j}-L_{,h})\,.
\label{geodesic_spray}
\end{equation}
\end{proposition}

A \textit{nonlinear connection} will be understood as a
connection on the fibered manifold $\mathcal{\ A}$ in the sense of \cite[pp.
30-32]{Giachetta}, i.e.\,, as a splitting of the tangent bundle $T\mathcal{A}
$ of $\mathcal{A}$\,,
\begin{equation*}
T\mathcal{A}=H\mathcal{A}\oplus V\mathcal{A}.
\end{equation*}%
The \textit{vertical subbundle} $V\mathcal{A}=\ker d(\pi_{TM}|_\mathcal{A})$
and the \textit{horizontal subbundle} $H\mathcal{A}$ are vector subbundles
of the tangent bundle $(T\mathcal{A},\pi_{T\mathcal{A}},\mathcal{A})$.
The local adapted basis will be denoted by $(\delta _{i},\dot{\partial}%
_{i}), $ where $\delta _{i}:=\tfrac{\partial }{\partial x^{i}}-G_{~i}^{j}%
\tfrac{\partial }{\partial \dot{x}^{j}},$ $\dot{\partial}_{i}=\tfrac{%
\partial }{\partial \dot{x}^{i}}$ and its dual basis, by $(dx^{i},\delta
\dot{x}^{i}=d\dot{x}^{i}+G_{~j}^{i}dx^{j})$. Here $G^i{}_j$ are general local connection coefficients.

We denote by $\mathfrak{h}$ and $\mathfrak{v}$ the horizontal and,
accordingly, the vertical projector determined by the nonlinear connection; that is, for any vector $%
X\in T\mathcal{A},$ locally written as $X=X^{i}\delta _{i}+\dot{X}^{i}\dot{%
\partial}_{i}$ \ we will have: $\mathfrak{h}X=X^{i}\delta _{i}$ and $%
\mathfrak{v}X=\dot{X}^{i}\dot{\partial}_{i}$.

\noindent \textbf{Cartan (canonical) nonlinear connection.}
The Cartan nonlinear connection $N$ on a Finsler spacetime $(M,L)$
is defined by the local connection coefficients
\begin{equation}
G_{~j}^{i}=G_{~\cdot j}^{i}\,,  \label{def:nlin}
\end{equation}%
Arc-length parametrized geodesics of the Finsler spacetime $(M,L)$ are
autoparallel curves of the canonical nonlinear connection.

\noindent \textbf{Nonlinear curvature tensor and Finsler Ricci scalar}

The curvature tensor of the canonical nonlinear connection on a
Finsler spacetime $(M,L)$ is a tensor on $TM$, which has the coordinate expression:
\begin{equation}
R^{i}{}_{jk}dx^{j}\wedge dx^{k}\otimes \dot{\partial}_{i}=dx^{j}\wedge dx^{k}\otimes [\delta
_{j},\delta _{k}]=(\delta _{k}G^{i}{}_{j}-\delta
_{j}G^{i}{}_{k})dx^{j}\wedge dx^{k}\otimes \dot{\partial}_{i}\,.
\label{N_curvature_comps}
\end{equation}%
The \textit{Finsler-Ricci scalar} $R_{0}$ makes sense on $\mathcal{A}_{0}$ and is given
by
\begin{equation*}
R_{0}=\dfrac{1}{L}R^{i}{}_{ik}\dot{x}^{k}\,.
\end{equation*}
The way introduced it above is equal to minus the one employed in \cite{Bao} (and denoted by $Ric$).

Besides the canonical nonlinear connection, it is possible to additionally define several linear connections on $\mathcal{A}$, which preserve the distributions
generated by the nonlinear connection. In this article we will pick, for simplicity, one of
these linear connections as a mathematical tool to ensure that all objects
we are dealing with are well defined tensors. Our particular choice of the linear connection is unessential, since it is just an auxiliary tool. The whole construction is independent of the typical Finslerian linear connections that one may use.

\noindent \textbf{Chern-Rund linear connection}

The Chern-Rund linear covariant derivative on a Finsler spacetime $(M,L)$,
defined on $\mathcal{A}\subset TM$, is locally given by the relations
\begin{equation}
\mathrm{D}_{\delta _{k}}\delta _{j}=\Gamma ^{i}{}_{jk}\delta _{i},\quad
\mathrm{D}_{\delta _{k}}\dot{\partial}_{j}=\Gamma ^{i}{}_{jk}\dot{\partial}%
_{i},\quad \mathrm{D}_{\dot{\partial}_{k}}\delta _{j}=\mathrm{D}_{\dot{%
\partial}_{k}}\dot{\partial}_{j}=0\,,  \label{Chern-Rund connection}
\end{equation}%
where $\Gamma ^{i}{}_{jk}:=\frac{1}{2}g^{ih}(\delta _{k}g_{hj}+\delta
_{j}g_{hk}-\delta _{h}g_{jk})$. We denote by $_{|i}$ $\mathrm{D}$-covariant
differentiation with respect to $\delta _{i}$.

The Chern-Rund linear covariant derivative allows us to introduce the
\textit{dynamical covariant derivative} in a very simple way, namely, as $%
\nabla :\Gamma (T\mathcal{A})\rightarrow \Gamma (T\mathcal{A})$ with $\nabla
=\dot{x}^{i}\mathrm{D}_{\delta _{i}}$. An important remark is that,
since $\dot{x}^{i}\Gamma _{~ji}^{k}=G_{~j}^{k},$ the dynamical covariant
derivative only depends on the canonical nonlinear connection $N,$ see \cite%
{Bucataru} (it can actually be introduced independently of $D$ or of any
other additional structure).

The dynamical covariant derivative can be used to define a measure of the change of the Cartan tensor along horizontal curves, called the
Landsberg tensor, see \cite{Bao}.

\noindent \textbf{Landsberg tensor}

The Landsberg tensor $P=P_{~jk}^{i}dx^{j}\otimes dx^{k}\otimes \delta _{i}$
is a tensor on $TM$, defined, in any local chart, by:
\begin{equation}
P_{~jk}^{i}=g^{mi}\nabla C_{mjk}=G_{~\cdot j \cdot k}^{i}-\Gamma _{~jk}^{i}\,.
\label{Landsberg tensor}
\end{equation}%
Its trace is denoted by $P_{i}=P^{j}{}_{ij}$.

The following identities will be useful when we consider action integrals
and calculus of variations on Finsler spacetimes:
\begin{align}
& \delta _{i}L=L_{|i}=0,\quad g_{ij|k}=0,~\ \ \ \dot{x}_{|j}^{i}=0,\  \\
& \nabla L=0,\quad \quad \nabla g_{ij}=0,~\ \ \ \ \nabla \dot{x}^{i}=0, \\
& P_{~jk}^{i}\dot{x}^{k}=0,~\ \ P_{i}\dot{x}^{i}=0.
\end{align}%
They can all be proven by using the homogeneity properties of the tensors
involved and the definition of the canonical nonlinear connection in terms
of the Finsler Lagrangian.

\noindent \textbf{(Semi)-Riemannian geometry as Finsler geometry}

Choosing $L=g_{ij}(x)\dot{x}^{i}\dot{x}^{j}$, the geometry of a Finsler
spacetime $(M,L)$ becomes essentially the geometry of the pseudo-Riemannian
spacetime manifold $(M,g)$. In this case, $G_{~j}^{i}=\gamma _{~jk}^{i}(x)%
\dot{x}^{k}$ and $R_{~jk}^{i}=r_{j~kl}^{~i}\dot{x}^{l}$ (where we have
denoted by small letters the geometric objects specific to Riemannian
geometry). The relation between the Finsler-Ricci scalar $R_{0}$ and the
usual Riemannian one $r=g^{ij}r_{i~jk}^{~k}$ is: $g^{ij}(LR_{0})_{\cdot
i\cdot j}=-2r.$

\subsection{Homogeneous geometric objects on $TM$} \label{sec:homog_geom_obj}

Homogeneity is a key concept in pseudo-Finslerian geometry, as the positive
homogeneity of $L$ in $\dot{x}$ entails the positive homogeneity of all
typical Finslerian geometric objects. We will briefly present here some
results on homogeneous geometric objects defined on conic subbundles $%
\mathcal{Q}$ $\subset \overset{\circ }{TM}.$ The results are straightforward
extensions of the results in \cite{Bucataru} and \cite{Szilasi:2014csf},
referring to objects defined on the whole slit tangent bundle.

\begin{definition}[Fiber homotheties]\label{def:homot}
By fiber homotheties on $\overset{\circ }{TM},$ we understand the mappings $%
\chi _{\alpha }:\overset{\circ }{TM}\rightarrow \overset{\circ }{TM},$ $\chi
_{\alpha }(x,\dot{x})=\left( x,\alpha \dot{x}\right) ,$ where $\alpha >0$.
\end{definition}

Fiber homotheties form a 1-parameter group of diffeomorphisms of $\overset{%
\circ }{TM},$ isomorphic to $(\mathbb{R}_{+}^{\ast },\cdot )$ and are generated
by the \textit{Liouville vector field}%
\begin{equation*}
\mathbb{C}=\dot{x}^{i}\dot{\partial}_{i}. \label{Liouville vector}
\end{equation*}%
We denote the corresponding group action by $\chi ,$ i.e.:%
\begin{equation}
\chi :\overset{\circ }{TM}\times \mathbb{R}_{+}^{\ast }\rightarrow \overset{%
\circ }{TM},~\ \chi \left( (x,\dot{x}),\alpha \right) =\chi _{\alpha }(x,%
\dot{x}).  \label{chi}
\end{equation}

\begin{definition}[Homogeneous tensor field]
Let $T$ be a tensor field over the conic subbundle $\mathcal{Q\subset }%
\overset{\circ }{TM}$. $T$ is called \textit{positively homogeneous} of
degree $k\in \mathbb{R}$ if and only if, for all $\alpha >0,$ its pullback
along the restriction $\chi _{\alpha }:\mathcal{Q}\rightarrow \mathcal{Q}$
satisfies
\begin{equation}
\chi _{\alpha }^{\ast }T=\alpha ^{k}T.
\end{equation}.%
\end{definition}

\begin{theorem}
\label{thm:homogeneity} A tensor field $T$ over $\mathcal{Q}$ is positively
homogeneous of degree $k\in \mathbb{R}$ if and only if
\begin{equation} \label{homog_Lie_deriv}
 \mathfrak{L}_{\mathbb{C}}T = kT\,.
 \end{equation}
\end{theorem}

\begin{proof}
See~\cite[Lemma 4.2.9]{Szilasi:2014csf} for the proof in the special case of scalar
	functions and~\cite[Lemma 4.2.14]{Szilasi:2014csf} for vector fields. In order to prove it in the general case, we momentarily reinterpret the multiplicative 1-parameter group $\{\chi_{\alpha}\}$ as the additive group $\mathbb{R}$, by setting $t:=\log(\alpha) \in \mathbb{R}$ and $\phi_{t}(x,\dot{x})= (x,e^{t}\dot{x})= \chi_{\alpha}(x,\dot{x})$, for all $(x,\dot{x}) \in TM$.
Assume, first, that $T$ is $k$-homogeneous, which means: $\phi _{t}^{\ast }T= e^{kt}T$. Then,

\begin{equation*}
\mathfrak{L}_{\mathbb{C}}T=\left. \dfrac{d}{dt }(\phi _{t}^{\ast
}T)\right\vert _{t =0 }=\left. \dfrac{d}{dt }(e^{kt}T)\right\vert _{t=0}=kT.
\end{equation*}%
Conversely, assume (\ref{homog_Lie_deriv}) holds. Differentiating the identity $\phi_{t}^{\ast}\phi_{\varepsilon}^{\ast}T = \phi_{t+ \varepsilon}^{\ast}T$ with respect to $\varepsilon$ at $\varepsilon = 0$, one finds:
\begin{equation*}
\phi _{t} ^{\ast }\mathfrak{L}_{\mathbb{C}}T =%
\dfrac{d}{dt} (\phi _{t} ^{\ast } T), ~\ \forall t.
\end{equation*}
Using (\ref{homog_Lie_deriv}), this leads to the differential equation $\dfrac{d}{dt} (\phi _{t} ^{\ast } T) = k\phi _{t} ^{\ast } T$ in the unknown $f(t)=\phi _{t} ^{\ast } T$. Integrating this equation with the initial condition $f(0) = \phi_{0}^{ \ast}T = T $, we find $\phi _{t} ^{\ast }T = e^{kt}T$, which, reverting to the old notation, is precisely $\chi _{\alpha }^{\ast }T=\alpha ^{k}T$.
\end{proof}

In particular, positive 0-homogeneity in $\dot{x},$ i.e., invariance under the fiber rescalings $\chi _{\alpha },$ $\alpha >0,$ can be treated as invariance under the flow  of $\mathbb{C}$.

\textbf{Note.} In \cite{Bucataru}, $k$-homogeneity of vector fields is defined differently (it is, in our terms ($k-1$)-homogeneity). We prefer, yet, this definition, which allows a unified treatment of tensor fields of any rank.

In the following, we will simply refer to positive homogeneity in $\dot{x}$ as \textit{homogeneity.} Some canonical examples of homogenous structures on the tangent bundle are:
\begin{enumerate}
\item The Liouville vector field $\mathbb{C}$ is homogeneous of degree $0$,
since $\mathcal{L}_{\mathbb{C}}\mathbb{C}=[\mathbb{C},\mathbb{C}]=0\,.$

\item The vertical local basis vectors $\dot{\partial}_{i}$ are homogeneous
of degree -1, as $[\mathbb{C},\dot{\partial}_{i}]=-\dot{\partial}_{i}.$

\item The\textit{\ natural tangent structure of }$TM,$\textit{\ }%
\begin{equation}
J=dx^{i}\otimes \dot{\partial}_{i}  \label{tangent structure}
\end{equation}%
is a globally defined, $\left( -1\right) $-homogeneous tensor of type (1,1).
Homogeneity follows from:%
\begin{equation*}
\mathfrak{L}_{\mathbb{C}}J=\mathfrak{L}_{\mathbb{C}}(dx^{i})\otimes \dot{%
\partial}_{i}+dx^{i}\otimes \lbrack \mathbb{C},\dot{\partial}%
_{i}]=0-dx^{i}\otimes \dot{\partial}_{i} = -J,
\end{equation*}%
where we have used 2.\ and $\mathfrak{L}_{\mathbb{C}}(dx^{i})=d\mathbf{i}_{%
\mathbb{C}}dx^{i}+\mathbf{i}_{\mathbb{C}}ddx^{i}=0$.
\end{enumerate}

\begin{definition}[Homogeneous nonlinear connection]
A nonlinear connection $T\mathcal{Q}=H\mathcal{Q}\oplus V\mathcal{Q}$ on the
conic subbundle $\mathcal{Q\subset }\overset{\circ }{TM},$ is called \textit{%
homogeneous,} if fiber homotheties preserve the horizontal subbundle,\ i.e., $(\chi _{\alpha })_{\ast }X\in H\mathcal{Q}$ for all $%
\alpha >0$ and all $X\in H\mathcal{Q}$.

\end{definition}
As it has been shown in \cite[Prop. 2.10.1]{Bucataru}, or in \cite[Cor.
7.5.10]{Szilasi2011}, that a nonlinear connection on $TM$ is homogeneous if and
only if the almost product structure $\mathfrak{P}=\mathfrak{h}-\mathfrak{v}$
is $0$-homogeneous; the result holds without modifications on $\mathcal{Q}%
\subset TM$.

In coordinates, homogeneity of a connection is characterized by the fact
that its coefficients $G_{~j}^{i}=G_{~j}^{i}\left( x,\dot{x}\right) $ are $1$%
-homogeneous functions in $\dot x$. An example of a homogeneous nonlinear connection is the Cartan nonlinear connection, \eqref{def:nlin}, of a Finsler spacetime.

Almost all Finsler geometric objects discussed above are \textit{anisotropic tensor fields,} which thus deserve
a special mentioning here. These can be mapped into specific tensor fields on the tangent bundle, called \textit{distinguished tensor fields}, or \textit{d-tensor fields}; for the latter, homogeneity can be discussed in a natural manner.

\begin{definition}\label{def:aniso}, \cite{Javaloyes_IJGMMP_2019}:
An anisotropic tensor field on the conic subbundle $\mathcal{Q\subset }\overset%
{\circ }{TM}$ is a section of the pullback bundle $\pi _{TM|\mathcal{Q}}^{\ast }(%
\mathcal{T}_{q}^{p}M),$ i.e.\ , a smooth mapping:

\begin{equation*}
T:\mathcal{Q}\rightarrow \mathcal{T}_{q}^{p}M,~\left( x,\dot{x}\right)
\mapsto T_{\left( x,\dot{x}\right) };
\end{equation*}%
i.e., for any $\left( x,\dot{x}\right) \in \mathcal{Q},$ $T_{\left( x,\dot{x}%
\right) }$ is a tensor on $M,$ based at $x=\pi _{TM}(x,\dot{x})$.
\end{definition}

Consequently, an anisotropic tensor field will be locally expressed as: $T_{\left( x,%
\dot{x}\right) }=T_{j_{1}...j_{q}}^{i_{1}...i_{p}}\left( x,\dot{x}\right)
(\partial _{i_{1}}\otimes ...\otimes \partial _{i_{p}}\otimes
dx^{j_{1}}\otimes ...\otimes dx^{j_{q}})|_{x}.$

In the presence of a nonlinear connection $N$ on $TM$ the following definition makes sense.

\begin{definition}, \cite{Bucataru}:
A d-tensor field on a conic subbundle $\mathcal{Q}\subset \overset{\circ}{TM}$ (regarded as a manifold) is a tensor field ${T\in }\mathcal{T}_{q}^{p}(\mathcal{Q})$, obeying the condition:
\begin{equation*}
T\left( \omega _{1},...,\omega _{p},V_{1},...,V_{q}\right) =T\left(
\varepsilon _{1}\omega _{1},...\varepsilon _{p}\omega _{p},\varepsilon
_{p+1}V_{1},...,\varepsilon _{p+q}V_{q}\right) ,
\end{equation*}%
for an arbitrarily \textit{fixed} choice of the projectors $\varepsilon
_{1},..,\varepsilon _{p+q}\in \left\{ \mathfrak{h,v}\right\} .$

\end{definition}

For instance, if $V$ is an arbitrary vector field on $\mathcal{Q},$ its
horizontal and vertical components $\mathfrak{h}V$ and $\mathfrak{v}V$,
taken separately, are d-tensor fields (of type $\left( 1,0\right) $), as
each of them acts on a single specified component $\mathfrak{h}\omega $ or $%
\mathfrak{v}\omega $ of a 1-form $\omega \in \Omega _{1}(\mathcal{Q})$,
whereas their sum is typically, not a d-tensor field.

With respect to the horizontal/vertical adapted local bases of $T\mathcal{Q}$ and $T^{\ast }%
\mathcal{Q},$ a d-tensor field $T$ will be expressed as a linear combination of
tensor products of $\delta _{i},\dot{\partial}_{i},dx^{i}$ and $\delta \dot{x%
}^{i}$, i.e., ${T}_{\left( x,\dot{x}\right)
}=T_{j_{1}...j_{q}}^{i_{1}...i_{p}}\left( x,\dot{x}\right) (\delta
_{i_{1}}\otimes ...\otimes \dot{\partial}_{i_{p}}\otimes dx^{j_{1}}\otimes
...\otimes \delta \dot{x}^{j_{q}})|_{(x,\dot{x})}.$

In the presence of a \textit{homogeneous} nonlinear connection, the adapted
basis elements $\delta _{i}$ are 0-homogeneous (and, as we have seen above, $%
\dot{\partial}_{i}$ are $\left( -1\right) $-homogeneous), hence the degree
of homogeneity (if any) of a d-tensor field $T$ can be established in local
coordinates, by simply evaluating the $\dot{x}$-homogeneity degree of the
coefficients $T_{j_{1}...j_{q}}^{i_{1}...i_{p}}.$ \bigskip

\textbf{Note.} Anisotropic tensor fields can be mapped into (multiple) d-tensor fields on $\mathcal{Q} \subset TM$ via horizontal or vertical lifts determined by the nonlinear connection. Yet, when doing this, one must take into account that using horizontal lifts $\partial_{i} \mapsto \delta_{i}$, one obtains a d-tensor of different degree of homogeneity, compared to the one obtained via a vertical lift $\partial_{i} \mapsto \dot{\partial}_{i}$, due to the $-1$-homogeneity of $\dot{\partial}_{i}$.

\bigskip

Examples of canonical homogeneous d-tensors are the Liouville vector field $%
\mathbb{C}$ and the tangent structure~$J$. Further examples of homogeneous
d-tensors arise once we consider a pseudo-Finsler structure on $M$. For
instance, on a Finsler spacetime, all the tensor fields encountered in the
previous section are homogeneous d-tensor fields on $\mathcal{A},$
of some degree $m$:

\begin{itemize}
\item the Finslerian metric tensor $g=g_{ij}dx^{i}\otimes dx^{j}$ ($k=0$);

\item the curvature $R=R_{~jk}^{i}dx^{j}\otimes dx^{k}\otimes \dot{\partial}%
_{i}$ of the canonical linear connection $\left( k=0\right) $;

\item the Landsberg tensor $P=P_{~jk}^{i}dx^{j}\otimes dx^{k}\otimes \delta_i$ $\left( k=0\right) .$
\end{itemize}

Other d-tensor fields, such as the Hilbert form $\omega =F_{\cdot i}dx^{i},$ or the \emph{Reeb vector field}
\begin{equation*}
\ell=l^{i}\delta _{i},
\end{equation*}%
are only defined on $\mathcal{A}_{0}=\mathcal{A}\backslash L^{-1}\left(
0\right) ,$ since the functions $ l^{i}=\dot x^i F^{-1}$ are only
defined on $\mathcal{A}_{0}.$ Both $\omega $ and $\ell$ are $0$-homogeneous
in $\dot x$ and will play a crucial role in the following, as we will see in
Section \ref{sssec:FinsPTM+}.

An important feature of both the Chern connection $D$ (and more generally, of any
of the typical Finslerian connections in the literature) and of the
dynamical covariant derivative $\nabla $ on $T\mathcal{A},$ is that they
preserve the distributions generated by the canonical nonlinear connection $N$ and hence, they map d-tensors into d-tensors, \cite{Bucataru}. Moreover, the degree of homogeneity of d-tensors is preserved by $D$-covariant differentiation with respect to 0-homogeneous vector fields (and is increased by 1 by dynamical covariant differentiation).

\section{The positively projectivized tangent bundle $PTM^{+}$\label{Section:PTM+}}

\label{sec:PTM+}

The positively projectivized tangent bundle $PTM^{+}$ is essential for a
mathematically well defined calculus of variations on Finsler spacetimes. It
also gives a nice way to understand positively homogeneous geometric objects
on $TM$, such as the Finsler function, or d-tensors, as sections of bundles
sitting over $PTM^{+}$, which we will discuss in detail in Section \ref%
{sec:Fields}.

We will first introduce $PTM^{+}$ over general manifolds before we formulate
the geometry of Finsler spacetimes on $PTM^{+}$. This reformulation is
important to construct well defined integrals of homogeneous functions. We will show that integration on domains in $PTM^{+}$ is actually
equivalent to integration over subsets of the observer space $\mathcal{O}$, with the
advantage that $PTM^{+}$ is explicitly independent of the
Finsler Lagrangian $L$, whereas the observer space (and therefore, all its subsets, which one may use as integration domains) are defined in terms of $L$.

\subsection{Definition and structure over general manifolds}

We first give the definition of the positively projectivized tangent bundle,
before we analyze its structure and point out how objects on $PTM^{+}$ are
related to $0$-homogeneous objects on $TM$. In the context of Finsler spacetimes, the positively projectivized tangent bundle was briefly discussed in \cite{Hohmann:2018rpp}. In the literature on positive definite Finsler geometry, $PTM^+$ is typically called the \textit{projective sphere bundle}.

Actually, in positive definite Finsler geometry, this bundle is interchangeably used with the \textit{indicatrix bundle}, as the two bundles are globally diffeomorphic. But, in Lorentzian Finsler geometry, as we will see below, this diffeomorphism does no longer exist, hence, in order to avoid any confusion, we preferred to make a clear distinction by the used terminology.

\subsubsection{Definition and structure}

\begin{definition}[The positive, or oriented, projective tangent bundle]
Let $M$ be a connected, orientable smooth manifold of dimension $n$. The
positive projective tangent bundle is defined as the quotient
space
\begin{equation}
PTM^{+}:=\overset{\circ }{TM}_{/\sim }  \label{PTM+}
\end{equation}%
where $\sim $ is the equivalence relation on $\overset{\circ }{TM}$ given
by:
\begin{equation}
(x,\dot{x})~\sim ~\left( x,\dot x'\right) \Leftrightarrow ~\dot x'=\alpha \dot{x}~~\text{%
for~~some~\ }\alpha >0\,.  \label{equiv_rel_PTM+}
\end{equation}
\end{definition}

In other words we identify the half-line $\left\{ (x,\alpha \dot{x}%
)~|~\alpha >0\right\} $ as a single point. We denote by
\begin{equation*}
\pi ^{+}:\overset{\circ }{TM}\rightarrow PTM^{+},(x,\dot{x})\mapsto \lbrack
(x,\dot{x})]
\end{equation*}%
the canonical projection.

The usual projectivized tangent bundle $PTM$ is obtained from $PTM^{+}$ by
deleting the distinction between positive and negative scaling factors, in
other words:%
\begin{equation*}
PTM=PTM^{+}{}_{/\mathbb{Z}_{2}}.
\end{equation*}
Conversely, by attaching orientations to the lines representing points of $%
PTM,$ one gets $PTM^{+}.$ In other words, $PTM^{+}$ is the  \textit{canonical oriented double covering} (also called orientation covering in (\textit{\cite{Lee}, } p. 394)) of the $(2n-1)$-dimensional manifold $PTM,$ in particular, it is always orientable. The above discussion can be summarized as follows.

\begin{proposition}
If $M$ is a connected smooth $n$-dimensional manifold, then $PTM^{+}$ is a
smooth, orientable manifold of dimension $2n-1$.
\end{proposition}

The orientability of $PTM^{+}$ is essential when considering integrals on $%
PTM^{+}.$

\bigskip

The smooth structure on $PTM^{+}$ is constructed as follows. Start with an atlas on $TM$, induced by an atlas on $M$ and denote by $(x^{i},\dot{x}^{i})$ the corresponding coordinate functions; then, for each local chart domain $U \in TM$ and each $i=0,\dots ,n-1$, define the open sets: $U_{i}=\left\{ (x,\dot{x})\in TU~|~\dot{x}^{i}>0\right\} ,$ $U_{i'}=\left\{ (x,\dot{x})\in U~|~\dot{x}^{i}<0\right\} .$ Then, for each $U^{+} \in \{ \pi^{+}(U_{i}), \pi^{+}(U_{i'}) \} $ and each $[(x,\dot{x})]\in U^{+}$, we
define the diffeomorphisms $\phi^{+}:=(x^{i},u^{\alpha })$ as:%
\begin{equation}
(x^{i},u^{\alpha })=(x^{0},...,x^{n-1},\dfrac{\dot{x}^{0}}{\dot{x}^{i}},...,%
\dfrac{\dot{x}^{i-1}}{\dot{x}^{i}},\dfrac{\dot{x}^{i+1}}{\dot{x}^{i}},...,%
\dfrac{\dot{x}^{n-1}}{\dot{x}^{i}}).  \label{standard coords PTM+}
\end{equation}%
The result is a differentiable atlas $\{ (U^{+},\phi^{+} )\} $ on $PTM^{+}.$

\bigskip

Using these charts, a quick direct computation shows that the projection $%
\pi ^{+}:\overset{\circ }{TM}\rightarrow PTM^{+},$ $(x^{i},\dot{x}%
^{i})\mapsto (x^{i},u^{\alpha })$ is a submersion. Since, obviously, $\pi
^{+}$ is surjective, it follows that $(\overset{\circ }{TM},\pi
^{+},PTM^{+})$ is a fibered manifold; actually, it posseses an even richer
structure, as has already been pointed out in \cite{Hohmann:2018rpp}. Let us
briefly recall this result:

\begin{proposition}[The principal bundle $(\protect\overset{\circ }{TM},%
\protect\pi ^{+},PTM^{+},\mathbb{R^*_+})$]\label{prop:ptm+bundle}
The slit tangent bundle $\overset{\circ }{TM}$ is a principal bundle over $%
PTM^{+},$ with fiber $(\mathbb{R}_{+}^{\ast },\cdot ).$
\end{proposition}

\begin{proof}
Consider $\chi $, as defined in (\ref{chi}), as the right action of the Lie
group $(\mathbb{R}_{+}^{\ast },\cdot )$ on $\overset{\circ }{TM}$. This
action preserves the fibers $(\pi ^{+})^{-1}(\left[ x,\dot{x}\right]
)=\left\{ \left( x,\alpha \dot{x}\right) ~|~\alpha \in \mathbb{R}_{+}^{\ast
}\right\} $ of $\pi ^{+},$ i.e., the half-lines with direction $(x,\dot{x}).$
Moreover, each of the fibers of $\pi ^{+}$ is obviously homeomorphic to $\mathbb{R}%
_{+}^{\ast }$.
\end{proof}

\bigskip

\noindent Actually, $PTM^{+}$ is nothing but the space of orbits of the Lie group action \eqref{chi}.

\bigskip
The Liouville vector field $\mathbb{C}$ is tangent to the fibers $(\pi
^{+})^{-1}[(x,\dot{x})]$ (i.e., it is $\pi ^{+}$\textit{-vertical}), which,
taking into account that these fibers are $1$-dimensional, means that $%
\mathbb{C}$ actually generates the tangent spaces to these fibers.

\bigskip

In its turn, $PTM^{+}$ is a fibered manifold over $M.$ More precisely, we
have the following result.

\begin{proposition}[Structure of the bundle $(PTM^{+},\protect\pi _{M},M,\mathbb{S}^{n-1})$]
The triple $(PTM^{+},\pi _{M},M)$, where $\pi _{M}:PTM^{+}\rightarrow M,$ $%
[(x,\dot{x})]\mapsto x$, is a fibered manifold with fibers diffeomorphic to
Euclidean spheres.
\end{proposition}

\begin{proof}
The projection $\pi _{M}$ is obviously a surjective submersion, meaning that
$(PTM^{+},\protect\pi _{M},M)$ is, indeed, a fibered manifold. Its fibers $\pi
_{M}^{-1}\left( x\right) =\left\{ [(x,\dot{x})]~|~\dot{x}\in
T_{x}M\right\} $ are orientation coverings of the projective tangent
spaces $PT_{x}M\simeq P\mathbb{R}^{n};$ but, the orientation covering of the
projective space $P\mathbb{R}^{n}$ is nothing but the round sphere $\mathbb{S%
}^{n-1}.$
\end{proof}

Moreover, $PTM^{+}$ is a \textit{natural bundle} over $M$, meaning that it
is obtained from $M$ via a covariant functor (see the Appendix for more
details on natural bundles); the natural lift to $PTM^{+}$ of a diffeomorphism $f:M \rightarrow M$ is given by $[(x, \dot{x})] \mapsto [(f(x), df_{x}(\dot{x}))]$, which is well defined by virtue of the linearity of $df_{x}$. 
On natural bundles, one can speak about general covariance of
Lagrangians, which is essential in ensuring the existence of a well defined
notion of energy-momentum tensor.

\textbf{Note:\ } As already stated above, the bundle $PTM^{+}$ is better known in the literature on positive definite Finsler spaces under the
name of \textit{projective sphere bundle} over $M$, see, e.g., \cite{Bao}; though the name is very well
justified by the above Proposition, we preferred to avoid this
terminology, in order to avoid any confusions with the
\textit{indicatrix bundle} $L=1$. This distinction is necessary
since, for positive definite Finsler structures, the fibers $%
I_{x}=L^{-1}\left( 1\right) $ of the indicatrix bundle are diffeomorphic to Euclidean spheres (i.e., diffeomorphic to the fibers of $PTM^{+}$), while in Lorentzian signature, this is no longer the case; actually, in the latter case, we have already seen that the fibers of the indicatrix bundle are generally disconnected, containing the observer spaces $\mathcal{O}_{x}$ as connected components. Moreover, it is essential for our later considerations to stress that the
construction of $PTM^{+}$ is completely independent of any pseudo-Finslerian
(or pseudo-Riemannian) structure whatsoever.

\subsubsection{From $PTM^{+}$ to $TM$ and back}

\label{sssec:PTM+TM}

Local computations on $PTM^{+}$ are much simplified if one uses \textit{%
local homogeneous coordinates} instead of the usual local coordinates $%
(x^{i},u^{\alpha })$ defining the manifold structure, in the same fashion as
on $PTM$, see \cite{Chern-Chen-Lam}.

For a given equivalence class $[(x,\dot{x})]$, local homogeneous  coordinates\footnote{Here, the word "local" is just meant to stress that these coordinates do not make sense globally, but only over a coordinate neighborhood. (Local) homogeneous coordinates are, obviously, not local coordinates as typically defined in differential geometry, since their number is equal to $2 dim(M)$, whereas the dimension of $PTM^{+}$ is $2 dim(M)-1$ .} \label{footnote_homog_coords} are defined as the coordinates $(x^{i},\dot{x}^{i})$ in the corresponding chart on $TM$ of an arbitrarily chosen representative of the class $[(x,\dot{x})];$ i.e., homogeneous coordinates are only unique up to multiplication by a positive scalar of the $\dot{x}$-coordinates.

In these coordinates, local computations on $PTM^{+}$ will become identical
to those on $TM,$ just with due care that the involved expressions in $%
\left( x^{i},\dot{x}^{i}\right) $ - which formally correspond to geometric
objects on $TM$ - should really define objects on $PTM^{+}$. A necessary
(but not always sufficient) condition is that these formally defined
geometric objects on $TM$ should be positively 0-homogeneous\ in $\dot{x}$,
i.e., invariant under the flow of $\mathbb{C}.$ Here we list the most
frequently encountered examples:

\begin{itemize}
\item \textbf{Functions.} A function $f:\overset{\circ }{TM}\rightarrow
\mathbb{R},$ $f=f(x,\dot{x})$ can be identified with a function $f^{+}$ on $%
PTM^{+}$ if and only if it is positively 0-homogeneous in $\dot{x};$ more
precisely, $f^{+}:PTM^{+}\rightarrow \mathbb{R}$ is defined by:%
\begin{equation*}
f^{+}[(x,\dot{x})]=f\left( x,\dot{x}\right) ,
\end{equation*}%
i.e., $f:=f^{+}\circ \pi ^{+}$; in homogeneous coordinates, $f^{+}$ and $f$
have identical coordinate expressions. The function $f^{+}$ is
differentiable at $[(x,\dot{x})]$ if and only if $f$ is differentiable at
one representative~$(x,\dot{x});$

\item \textbf{Vector fields.} For a vector field $X=X^{i}\partial _{i}+%
\tilde{X}^{i}\dot{\partial}_{i}\in \mathcal{X}(\overset{\circ }{TM})$, the
projection%
\begin{equation*}
X^{+}:=(\pi ^{+})_{\ast }X
\end{equation*}%
is a well defined vector field on $PTM^{+}$ if and only if $X$ is positively
0-homogeneous in $\dot{x}$, i.e., $\mathfrak{L}_{\mathbb{C}}X=0.$

This is justified as follows. Having in view that $\pi ^{+}$ is surjective,
the necessary and sufficient condition for $X^{+}$ to be a well defined
vector field on $PTM^{+}$ is that the mapping $[(x,\dot{x})]\mapsto X_{~%
\left[ (x,\dot{x})\right] }^{+}=(\pi ^{+})_{\ast }X_{\left( x,\dot{x}\right)
}$ is independent on the choice of $\left( x,\dot{x}\right) $ in the class $%
\left[ \left( x,\dot{x}\right) \right] ;$ but this means precisely
0-homogeneity of~$X$.

In coordinates, this boils down to the fact that the functions $X^{i}$ are
positively 0-homogeneous, while $\tilde{X}^{i}$ are
1-homogeneous in $\dot{x}.$

An interesting remark is that the correspondence $X\mapsto X^{+}$ is
surjective, but not injective, as all vector fields of the form $X + f\mathbb{C}$ on $TM$, where $f$ is a 0-homogeneous function, descend onto the same vector field $X^{+}\in \mathcal{X}(PTM^{+})$.

\item \textbf{Differential forms. }For differential forms $\rho $ on $%
\overset{\circ }{TM}$, $0$-homogeneity is necessary, but not sufficient in
order to be identified with differential forms on $PTM^{+}.$ The following
result (derived in \cite{Hohmann:2018rpp}) is just a coordinate-free
restatement of a similar result on $PTM,$ see \cite{Chern-Chen-Lam}:

\begin{proposition}
Let $\rho \in \Omega (TM)$ be defined on a conic subbundle of $TM.$ Then,
there exists a unique differential form $\rho ^{+}\in \Omega (PTM^{+})$ such
that $\rho =(\pi ^{+})^{\ast }\rho ^{+}$ if and only if the following
conditions are fulfilled:

\begin{enumerate}
\item $\rho $ is 0-homogeneous in $\dot{x},$ i.e.,
\begin{equation}
\mathfrak{L}_{\mathbb{C}}\rho =0;  \label{Lie deriv rho}
\end{equation}

\item $\rho $ is $\pi ^{+}$-horizontal, i.e.,
\begin{equation}
\mathbf{i}_{\mathbb{C}}\rho =0.  \label{contraction_rho_C}
\end{equation}
\end{enumerate}
\end{proposition}
\end{itemize}

\begin{remark}\cite{Hohmann:2018rpp}
	\begin{enumerate}
		\item The projection $\pi ^{+}$ is locally represented in homogeneous
		coordinates as the identity: $\pi ^{+}:(x^{i},\dot{x}^{i})\mapsto (x^{i},%
		\dot{x}^{i})$. The latter relation tells us that the
			coordinate expressions of geometric objects on $TM$ (e.g., $f,X$, $\rho $)\
	that can be identified with geometric objects $f^{+},X^{+},\rho ^{+}$ etc. on
		 $PTM^{+},$ will be identical to the expressions of the
			latter in homogeneous coordinates.

		\item Exterior differentiation of forms $\rho ^{+}\in \Omega (PTM^{+})$ can
		be carried out identically to exterior differentiation of the corresponding
		form $\rho \in \Omega (\overset{\circ }{TM}),$ since:%
		\begin{equation*}
			d\rho =d\left( (\pi ^{+})^{\ast }\rho ^{+}\right) =\left( \pi ^{+}\right)
			^{\ast }d\rho ^{+}.
		\end{equation*}%
		In particular, differentiation of functions on $PTM^{+}$ is carried out identically to the one on $TM$.
	\end{enumerate}
\end{remark}

\subsection{Over Finsler spacetimes $(M,L)$}

After having introduced $PTM^{+}$ in the previous section, we now
demonstrate that the geometry of a Finsler spacetime can be understood in
terms of geometric objects on $PTM^{+}$. This eventually enables us to write
down the desired action integrals for field theories in a mathematically
precise way.

\subsubsection{Finsler Geometry on $PTM^+$ and volume forms}

\label{sssec:FinsPTM+} On a Finsler spacetime, we defined the conic subbundles $\mathcal{A}, \mathcal{A}_{0},\mathcal{T},\mathcal{N}\subset
\overset{\circ }{TM}$, and the observer space $\mathcal{O}$, see Section \ref%
{sec:defFins}. We will denote by a plus sign, e.g., $\mathcal{T}^{+}=\pi
^{+}(\mathcal{T}),$ $\mathcal{A}^{+}=\pi ^{+}\mathcal{(A})$ etc., their
images through $\pi ^{+}:\overset{\circ }{TM}\rightarrow PTM^{+}$; also, we
will always use local homogeneous coordinates on $PTM^{+}$.

\noindent \textbf{Canonical nonlinear connection.}

The canonical nonlinear connection $N$ on $\mathcal{A}$, see equation~\eqref%
{def:nlin} can be transplanted to $\mathcal{A}^{+}$ in a natural way, as follows.
Start with an arbitrary vector $X^{+}\in T\mathcal{A}^{+}$. As we have seen above, it always corresponds to a positively
0-homogeneous vector $X\in T\mathcal{A}$ (which is unique up to a multiple of $\mathbb{C}$). Then, $X$ is decomposed into its $N$-horizontal and vertical
components $\mathfrak{h}X=X^{i}\delta _{i}$ and $\mathfrak{v}X=\dot{X}^{i}%
\dot{\partial}_{i};$ both components are positively 0-homogeneous, due to the homogeneity of $N$, hence they descend back onto vectors $\mathfrak{h}X^{+},\mathfrak{v}X^{+}$ on $T\mathcal{A%
}^{+}$. Moreover,  $\mathfrak{h}X^{+},\mathfrak{v}X^{+}$ are uniquely defined by $X^{+}$, as the possible multiple of $\mathbb{C}$ appearing in the procedure will be projected back to $PTM^{+}$ into the zero vector.
This naturally gives rise to a splitting%
\begin{equation*}
X^{+}=\mathfrak{h}X^{+}+\mathfrak{v}X^{+},
\end{equation*}%
i.e., to a connection $N^{+}$ on $\mathcal{A}^{+}:=\pi ^{+}(\mathcal{A}):$%
\begin{equation}
T\mathcal{A}^{+}=H\mathcal{A}^{+}\oplus V\mathcal{A}^{+}.
\label{connection A+}
\end{equation}%
The vectors $\mathfrak{h}X^{+}=(\pi ^{+})_{\ast }(\mathfrak{h}X)$ and $%
\mathfrak{v}X^{+}:=(\pi ^{+})_{\ast }(\mathfrak{v}X)$ are expressed in
homogeneous coordinates as:%
\begin{equation*}
\mathfrak{h}X^{+}=X^{i}\delta _{i},~\ \ \mathfrak{v}X^{+}=\dot{X}^{i}\dot{%
\partial}_{i}.
\end{equation*}%
Similarly, the Chern-Rund connection $D$ gives rise to a linear connection $%
D^{+} $ on $\mathcal{A}^{+},$ having the same local expression of covariant
derivatives as $D$.\\

\noindent \textbf{Contact structure and volume form for the set of non-null admissible
directions.}

In the following, we will identify a
canonical volume form on the set of admissible non-null directions $\mathcal{A}%
_{0}^{+}=\pi ^{+}(\mathcal{A}_{0})$. The Hilbert form $\omega =F_{\cdot i}dx^{i}$,
defined on $\mathcal{A}_{0}$ obeys the conditions:
\begin{equation*}
\mathbf{i}_{\mathbb{C}}\omega =0,~\ \ \mathfrak{L}_{\mathbb{C}}\omega =d%
\mathbf{i}_{\mathbb{C}}\omega +\mathbf{i}_{\mathbb{C}}d\omega =0,
\end{equation*}%
which means that it can be identified with a differential form $\omega ^{+}$
on $\mathcal{A}_{0}^{+}\subset PTM^{+},$ such that $(\pi ^{+})^{\ast }\omega
^{+}=\omega ;$ in homogeneous coordinates:
\begin{equation}
\omega ^{+}=F_{\cdot i}dx^{i}  \label{Hilbert form}
\end{equation}%
and%
\begin{equation}
d\omega ^{+}=\frac{1}{F}(\epsilon g_{ij} - F_{\cdot i} F_{\cdot j} )\delta \dot{x}^{j}\wedge dx^{i}.
\end{equation}

A direct calculation, see, \cite{Chern-Chen-Lam}, shows that, for $\dim M=4,$%
\begin{equation}
\omega ^{+}\wedge ~d\omega ^{+}\wedge d\omega ^{+}\wedge d\omega ^{+}~=3!%
\dfrac{\det g}{L^{2}}\mathbf{i}_{\mathbb{C}}(d^{4}x\,\wedge d^{4}\dot{x})=3!%
\dfrac{\det g}{L^{2}}\mathrm{Vol}_{0},  \label{omega contact structure}
\end{equation}%
where
\begin{equation}\label{eq:Vol0}
\mathrm{Vol}_{0}=\mathbf{i}_{\mathbb{C}}(d^{4}x\,\wedge d^{4}\dot{x})\,,
\end{equation}
is always nonzero. In other words, the Hilbert form $\omega ^{+}$
defines a \textit{contact structure} on $PTM^{+}$.

\bigskip

In contact geometry, the \textit{Reeb vector field} $\ell^{+}\in \mathcal{X}(%
\mathcal{A}_{0}^{+})$ corresponding to the contact structure $\omega ^{+}$
is uniquely defined by the conditions%
\begin{equation}
\mathbf{i}_{\ell^{+}}(\omega ^{+})=1,~\ \mathbf{i}_{\ell^{+}}d\omega ^{+}=0.
\label{Reeb vector field conds}
\end{equation}%
In our case, this gives:

\begin{proposition}
The Reeb vector field $\ell^{+}$ corresponding to the contact structure $%
\omega ^{+}$ on $\mathcal{A}_{0}^{+}$ is expressed in local homogeneous
coordinates as:
\begin{equation*}
\ell^{+}=l^{i}\delta _{i},~\ \ l^{i}= \dfrac{\dot{x}^{i}}{\sqrt{F}}.
\end{equation*}
\end{proposition}

\begin{proof}
We have $\mathbf{i}_{\ell^{+}}\omega ^{+}=1$ and
\begin{align}
\mathbf{i}_{\ell^{+}}d\omega ^{+}
= F^{-1}(\epsilon g_{ij}- F_{\cdot i} F_{\cdot j}) \mathbf{i}_{\ell^{+}}(\delta \dot{x}^{i}\wedge dx^{j})
= F^{-1}(\epsilon g_{ij}- F_{\cdot i} F_{\cdot j}) l^{j}\delta \dot{x}^{i}=0,
\end{align}
where for the last equallity we used that $F_{\cdot j} l^j = 1$ and $\epsilon g_{ij} l^j =F_{\cdot i}$.
\end{proof}

The importance of the Reeb vector field is given by the following result.

\begin{proposition}
Let $c:[a,b]\rightarrow M,$ $s\mapsto x(s)$ be a non-lightlike admissible curve parametrized by arc
length and $C:[a,b]\rightarrow \mathcal{A}_{0}^{+},$ $s\mapsto \lbrack (x(s),\dot{x}%
(s))],$ its canonical lift. Then, $C$ is an integral curve of $\ell ^{+}$ if
and only if $c$ is an arc-length parametrized geodesic of $(M,L).$
\end{proposition}

\begin{proof}
In homogeneous coordinates, $\dot{C}=\dot x^i(s)\delta _{i}+\delta_s \dot{x}%
^{i}(s)\dot{\partial}_{i}$, where $\delta_s \dot{x}^{i}(s)=\ddot
x^i(s)+2G^{i}(x^{j}(s),\dot x^j(s));$ that is, $C$ is an integral curve of $%
\ell^{+}$ is and only if:%
\begin{equation*}
\dot x^i(s)=l^{i},~\delta_s \dot x^i(s) = 0.
\end{equation*}%

The first condition above is trivially satisfied by any curve parametrized by arc length, since $L(x,\dot x(s))=\mathrm{sign}(L)$; taking into account the properties
of the canonical nonlinear connection, the second condition is equivalent to
the fact that $c$ is a arc-length parametrized geodesic of $(M,L)$, see \eqref{geodesic_eqn}.
\end{proof}

The contact structure $\omega ^{+}$, now enables us to identify a
canonical volume form on $\mathcal{A}_{0}^{+}$.

\begin{definition}[Canonical volume form]
Let $(M,L)$ be a Finsler spacetime, $\mathcal{A}_{0}^{+} \subset PTM^{+}$, the set of its admissible, non-null directions
and $\omega ^{+},$ the Hilbert form on $\mathcal{A}_{0}^{+}$. Then
\begin{equation}
d\Sigma ^{+}:=\dfrac{\mathrm{sign}(\det g) }{3!}\omega ^{+}\wedge (d\omega ^{+})^{3}=%
\dfrac{\left\vert \det g\right\vert }{L^{2}}\mathrm{Vol}_{0},
\label{canonical volume form}
\end{equation}%
where $\mathrm{Vol}_{0}$ is as in \eqref{eq:Vol0}, is called the canonical volume form on $\mathcal{A}_{0}^{+}$.
\end{definition}

Note that, on $\mathcal{A}_{0}^{+}$, $g$ is nondegenerate, so, $d\Sigma ^{+}$ is well defined.

With respect to this volume form, the divergence of horizontal and vertical
vector fields, $X=X^{i}\delta _{i}$ and $Y=Y^{i}\dot{\partial}_{i}$, on $%
\mathcal{A}_{0}^{+}$, is \cite{Hohmann:2018rpp}:
\begin{align}
\mathrm{div}(X)& =(X^{i}{}_{|i}-P_{i}X^{i})\,,
\label{divergence_horizontal} \\
\mathrm{div}(Y)& =(Y^{i}{}_{\cdot i}+2C_{i}Y^{i}-\frac{4}{L}Y^{i}\dot{x}%
_{i})\,,  \label{divergence_vertical}
\end{align}%
where $P_{i}$ is the trace of the Landsberg tensor \eqref{Landsberg tensor}
and $C_{i}$ is the trace of the Cartan tensor \eqref{eq:CartanT}. For any $%
f:\mathcal{A}_{0}^{+} \rightarrow \mathbb{R}$, the above equations imply
\begin{equation}
\nabla f=\mathrm{div}(f\ell )=\mathrm{div}(fl^{i}\delta _{i}).\
\label{divergence nabla}
\end{equation}

\subsubsection{Integration on $PTM^{+}$ and integration on observer space}

In positive definite Finsler spaces, the unit sphere bundle~$L^{-1}(1)$ is
globally diffeomorphic to $PTM^{+}$, \cite{Bao}. But, passing to Finsler
spacetimes, this is no longer true; this is easy to see since the fibers $%
I_{x}=L^{-1}(1)\cap T_{x}M$ are non-compact (they are, even in the simplest
case of Lorentzian metrics, hyperboloids), while the fibers of $PTM^{+}$ are
compact. Still, we will be able to establish a correspondence between the
observer space $\mathcal{O}$ and the set of future pointing timelike
directions $\mathcal{T}^{+}:=\pi ^{+}(\mathcal{T}).$ A
preliminary result, proven in \cite{Hohmann:2018rpp}, refers to compact subsets of $\mathcal{T}.$

\begin{proposition}, see \cite{Hohmann:2018rpp}:
\label{diffeo_compact subsets}

\begin{enumerate}
\item For any admissible compact, connected subset $D\subset L^{-1}(1),$ the projection $\pi ^{+}:D\mapsto
\pi ^{+}(D)\subset \pi ^{+}(L^{-1}(1))$ is a diffeomorphism.

\item For any connected, admissible and non-null compact subset $%
D^{+}\subset \pi ^{+}(\mathcal{A}_{0})$ and any differential form $\rho ^{+}
$ on $PTM^{+}$:
\begin{equation}
\underset{D^{+}}{\int }\rho ^{+}=\underset{D}{\int }\rho ,
\label{rel_integrals_SM_PTM}
\end{equation}
where $\rho =\left( \pi ^{+}\right) ^{\ast }\rho ^{+}$ is a differential
form on $\overset{\circ }{TM}$ and $D:=\left( \pi ^{+}\right)
^{-1}(D^{+})\cap L^{-1}(1)$.
\end{enumerate}
\end{proposition}

The above result will be mostly applied to \textit{pieces} $D\subset \mathcal{O}%
\subset L^{-1}(1)$, where, by a piece $D\subset X$, we will understand, \cite{Krupka-book}, a compact $n$-dimensional submanifold of $X$ with boundary. Yet, it can be extended to the whole observer
space, as long as we integrate compactly supported differential forms.

\begin{proposition}
\label{diffeo_O_T}{\ }
In any Finsler spacetime:
\begin{enumerate}
\item The mapping $\pi ^{+}:\mathcal{O}\rightarrow \mathcal{T }^{+}$ is a
diffeomorphism;

\item for any compactly supported 7-form $\rho ^{+}$ on $\mathcal{T}^{+}:$
\begin{equation}
\underset{\mathcal{T}^{+}}{\int }\rho ^{+}=\underset{\mathcal{O}}{\int }\rho
,  \label{int_O}
\end{equation}
where $\rho =\left( \pi ^{+}\right) ^{\ast }\rho ^{+}.$
\end{enumerate}
\end{proposition}

\begin{proof}

\begin{enumerate}
\item \textit{Injectivity}: Assume $\pi ^{+}(x,\dot{x})=\pi ^{+}(x',\dot x')$ for
some $(x,\dot{x}),(x',\dot x')\in \mathcal{O}.$ It follows that $[(x,\dot{x}%
)]=[(x',\dot x')],$ i.e., $x=x'$ and there exists an $\alpha >0$ such that $\dot x'=\alpha
\dot{x}.$ Applying $L$ to both hand sides, we find $L(x,\dot x')=\alpha ^{2}L(x,%
\dot{x});$ but, on $\mathcal{O},$ $L=1,$ which means that $\alpha ^{2}=1.$
Since $\alpha >0,$ it follows that $(x,\dot x')=(x,\dot{x}).$

\textit{Surjectivity:} Consider an arbitrary $[(x,\dot{x})]\in \mathcal{T}
^{+}.$ It means that $(x,\dot{x})\in \mathcal{T};$ as $\mathcal{T}$
is a conic subbundle of $TM,$ the vector $(x,\alpha \dot{x}),$ with $\alpha
:=L\left( x,\dot{x}\right) ^{-1/2},$ also belongs to $\mathcal{T}.$ But $%
L(x,\alpha \dot{x})=1,$ hence $(x,\alpha \dot{x})\in \mathcal{O}$. Since $%
\pi ^{+}(x,\dot{x})=\pi ^{+}(x,\alpha \dot{x})=[(x,\dot{x})],$ it follows
that $[(x,\dot{x})]\in \pi ^{+}(\mathcal{O}).$

\textit{Smoothness}: of $\pi ^{+}$ and of its inverse follow immediately,
working in homogeneous coordinates, in which $\pi ^{+}$ is represented as
the identity.

\item follows immediately from $\rho =\left( \pi ^{+}\right) ^{\ast }\rho
^{+}$ and point $1$.
\end{enumerate}
\end{proof}

In particular, the above result shows that:
\begin{equation}
\mathcal{O}^{+}=\mathcal{T}^{+}.  \label{O-T_rel}
\end{equation}

With this section we have established that integration of differential forms on the observer space of Finsler spacetimes can be understood as integration of differential forms on (subsets of) $PTM^{+}$.

\section{Fibered manifolds and fields over a Finsler spacetime}

\label{sec:Fields}

Having understood how integrals of homogeneous functions on Finsler
spacetimes can be constructed, the next step in constructing action
integrals is to understand fields (and their derivatives) as sections $%
\gamma $ into fibered manifolds $Y$ over $PTM^{+}$. But, with this aim, we
need to understand the structure of such fibered manifolds.

For an improved readability of the article, we give a detailed summary of
jet bundles over fibered manifolds and coordinate free calculus of
variations in Appendix \ref{app:A}.

\subsection{Fibered manifolds over $PTM^{+}$}
Consider a Finsler spacetime $(M,L)$ and denote by $(Y,\Pi ,PTM^{+})$ an
arbitrary fibered manifold of dimension $7+m.$ Then, $Y$ will acquire a
double fibered manifold structure:
\begin{equation}
Y\overset{\Pi }{\longrightarrow }PTM^{+}\overset{\pi _{M}}{\longrightarrow }%
M.  \label{double fibered structure}
\end{equation}%
As a consequence, $Y$ will admit an atlas consisting of fibered charts $%
\left( V,\psi \right) ,$ $\psi =(x^{i},u^{\alpha },z^{\sigma}),$ $%
i=0,...,3, $ $\alpha =0,...,2$, $\sigma =1,...,m$ on $Y,$ that are adapted
to both fibrations, i.e.\thinspace , the two projections will be represented
in these charts as:
\begin{equation*}
\Pi :(x^{i},u^{\alpha },z^{\sigma })\mapsto \left( x^{i},u^{\alpha }\right)
,\quad \pi _{M}:(x^{i},u^{\alpha })\mapsto \left( x^{i}\right) \,.
\end{equation*}%
Further, corresponding to any induced local chart $(\Pi(V),\phi ),$ $%
\phi =(x^{i},u^{\alpha })$ on $PTM^{+},$ we can introduce the homogeneous
coordinates $(x^{i},\dot{x}^{i})$, which we will sometimes denote
collectively as $(x^{A}).$ This way, we obtain on $V = \Pi^{-1}(U^+)$ the coordinate
functions
\begin{equation*}
\tilde{\psi}:=(x^{i},\dot{x}^{i},y^{\sigma })=(x^{A},y^{\sigma })
\end{equation*}%
on $V$, which we will call \textit{fibered homogeneous coordinates}. The corresponding fiber coordinate $y^\sigma$ is typically not unique, its relation to the original coordinates $(x^i,u^\alpha,z^\sigma)$ may depend on the choice of representative $(x,\dot x)$ of $[(x,\dot x)]\in PTM^+$. The precise relation will be discussed in the applications.

In fibered homogeneous coordinates, local sections (physical fields)  $\gamma
:W^{+}\rightarrow Y,\ [(x,\dot{x})]\mapsto \gamma [(x,\dot{x})]$
(where $W^{+} \subset PTM^{+}$ is open) are represented as:%
\begin{equation}
\gamma :(x^{i},\dot{x}^{i})\mapsto (x^{i},\dot{x}^{i},y^{\sigma }(x^{i},\dot{%
x}^{i})).  \label{Finslerian section}
\end{equation}%
The set of all such sections is denoted by $\Gamma (Y)$.

\begin{remark}\label{remarkHomFibCoord}
	
Generically, the physical fields we are considering are $k$-homogeneous with respect to $\dot x$. Hence their coordinate representation in fibered homogeneous coordinates satisfies
$\gamma(x,\alpha \dot x^i) = (x^i,\alpha \dot x^i, y^\sigma(x^i,\alpha\dot x^i)) = (x^i,\alpha \dot x^i,\alpha^k y^\sigma(x^i,\dot x^i)) $. 

Alternatively we could represent them in the original coordinates on $Y$ as $\gamma(x^i,u^\alpha) =  (x^i, u^\alpha,z^\sigma(x^i,u^\alpha))$, where $ z^\sigma(x^i,u^\alpha(x,\dot x)) =: z^\sigma(x^i,\dot x^i) $ is a zero-homogeneous in $\dot x$ when expressed in terms of $\dot x$, i.e.\ does not depend on the representative of $[(x,\dot x)]\in PTM^+$.
\end{remark}

On the jet bundle $J^{r}Y,$ fibered charts $(V,\tilde{\psi})$ induce the
fibered charts \footnote{We note that, since we are using \textit{homogeneous} coordinates over each chart domain, the number of coordinate functions $(y^{\sigma}_{,i},y^{\sigma}_{\dot i})$ is $8m$, not $7m$ as one would expect taking into account the dimension of the fibers of $J^{1}Y \rightarrow Y$.} $(V^{r},\tilde{\psi}^{r})$, with:%
\begin{equation*}
\tilde{\psi}^{r}=(x^{i},\dot{x}^{i},y^{\sigma },y_{~,i}^{\sigma },y_{~\cdot
i}^{\sigma },...,y_{~\cdot i_{1}\cdot i_{2}...\cdot i_{r}}^{\sigma })\,,
\end{equation*}%
where, for $k=1,...,r$ and $\gamma \in \Gamma (Y)$ locally represented
as in (\ref{Finslerian section}),%
\begin{equation*}
y_{~,i_{1}...\cdot i_{k}}^{\sigma }(x^j,\dot x^j)=\dfrac{\partial ^{k}}{%
\partial x^{i_{1}}...\partial \dot{x}^{i_{k}}}(y^{\sigma }(x^{j},\dot{x}%
^{j}))
\end{equation*}%
are all partial $x,\dot{x}$-derivatives up to the total order $k$. The
canonical projections $\Pi ^{r,s}:J^{r}Y\rightarrow J^{s}Y,$ $J_{\left( x,\dot{x%
}\right) }^{r}\gamma \mapsto J_{\left( x,\dot{x}\right) }^{s}\gamma $ (with $%
r>s$), are then represented as:
\begin{equation*}
\Pi ^{r,s}:(x^{i},\dot{x}^{i},y^{\sigma },y_{~,i_{1}}^{\sigma
},...,y_{~\cdot i_{1}\cdot i_{2}...\cdot i_{r}}^{\sigma })\mapsto (x^{i},%
\dot{x}^{i},y^{\sigma },y_{~,i_{1}}^{\sigma },...,y_{~\cdot i_{1}\cdot
i_{2}...\cdot i_{s}}^{\sigma }),
\end{equation*}%
accordingly,
\begin{equation*}
\Pi ^{r}:J^{r}Y\rightarrow PTM^{+},~(x^{i},\dot{x}^{i},y^{\sigma
},y_{~,i_{1}}^{\sigma },...,y_{~\cdot i_{1}\cdot i_{2}...\cdot
i_{r}}^{\sigma })\mapsto (x^{i},\dot{x}^{i}).\,
\end{equation*}%

In the calculus of variations, we will need two classes of differential
forms on $J^{r}Y,$ namely, horizontal forms and contact forms, see Appendix~%
\ref{app:Lagrangians} for more details.

\begin{enumerate}
\item $\Pi ^{r}$-\textit{horizontal forms} $\rho\in \Omega _{k}(J^{r}Y)$ are
defined as forms that vanish whenever contracted with a $\Pi ^{r}$%
-vertical vector field. In the natural local basis $(dx^{i},d\dot{x}%
^{i},dy^{\sigma },...dy_{~\cdot i_{1}...\cdot i_{r}}^{\sigma }),$ they are
expressed as:%
\begin{equation}
\rho =\dfrac{1}{k!}\rho _{i_{1}i_{2}...i_{k}}dx^{i_{1}}\wedge
dx^{i_{2}}\wedge ...\wedge d\dot{x}^{i_{k}},
\end{equation}%
where $\rho _{i_{1}i_{2}...i_{k}}$ are smooth functions of the coordinates
on $J^{r}Y$. In particular, Lagrangians will be characterized as $\Pi ^{r}$%
-horizontal 7-forms $\lambda =\Lambda \mathrm{Vol}_{0}$ on $J^{r}Y,$ where $\mathrm{Vol}_{0}$ is as in \eqref{eq:Vol0}.

Similarly, $\Pi ^{r,s}$-horizontal forms, $0\leq s\leq r$ are
locally generated by wedge products of $dx^{i},d\dot{x}^{i},dy^{\sigma
},dy_{~,i}^{\sigma }...,dy_{~\cdot i_{1}...\cdot i_{s}}^{\sigma }$.

\item \textit{Contact forms} on $J^{r}Y$ are, by definition, forms $\rho \in \Omega
_{k}(J^{r}Y)$ that vanish along prolonged sections, i.e., $J^{r}\gamma
^{\ast }\rho =0,~\forall \gamma \in \Gamma (Y).$ For instance,
\begin{equation}\label{eq:contactForm}
\theta ^{\sigma }=dy^{\sigma }-y_{~,i}^{\sigma }dx^{i}-y_{~\cdot i}^{\sigma
}d\dot{x}^{i},~\ \ \theta _{~,i}^{\sigma }=dy_{~,i}^{\sigma
}-y_{~,i,j}^{\sigma }dx^{j}-y_{~,i\cdot j}^{\sigma }d\dot{x}^{j}~\ \ ~\ \
etc.
\end{equation}%
are contact forms, composing the so-called contact basis $\{dx^{i},d\dot{x}%
^{i},\theta ^{\sigma },\theta _{~,i}^{\sigma },\theta _{~\cdot i}^{\sigma
},...\theta _{~\cdot i_{1}...\cdot i_{r-1}}^{\sigma
},dy_{~,i_{1}...,i_{r}}^{\sigma },...dy_{~\cdot i_{1}...\cdot i_{r}}^{\sigma
}\}$ of $\Omega (J^{r}Y)$.

An important class of contact forms are\textit{\ source forms} (or \textit{%
dynamical forms}), $\rho \in \Omega _{8}(J^{r}Y)$ that can be expressed,
corresponding to any fibered chart, as:%
\begin{equation}
\rho =\rho _{\sigma }\theta ^{\sigma }\wedge \mathrm{Vol}_{0}
\label{source_forms_PTM+}
\end{equation}%
(see the Appendix for a coordinate-free definition); Euler-Lagrange forms of
Lagrangians fall into this class.
\end{enumerate}

\bigskip

Raising to $J^{r+1}Y,$ any differential form $\rho \in \Omega _{k}(J^{r}Y)$
can be uniquely decomposed as:%
\begin{equation*}
\left( \Pi ^{r+1,r}\right) ^{\ast }\rho =h\rho +p\rho ,
\end{equation*}%
where $h\rho $ is horizontal and $p\rho $ is contact. The horizontal
component $h\rho $ is what survives of $\rho $ when pulled back by prolonged
sections $J^{r}\gamma$ (where $\gamma \in \Gamma (Y),$) i.e.,
\begin{equation}
J^{r}\gamma ^{\ast }\rho =J^{r+1}\gamma ^{\ast }(h\rho ).  \label{prop_h}
\end{equation}

The mapping $h:\Omega(J^{r}Y) \rightarrow \Omega(J^{r+1}Y)$ is a morphism of exterior
algebras, called \textit{horizontalization. }On the natural basis 1-forms,
it acts as:
\begin{equation}
hdx^{i}:=dx^{i},~hd\dot{x}^{i}=d\dot{x}^{i},~\ hdy^{\sigma }=y_{~,i}^{\sigma
}dx^{i}+y_{~\cdot i}^{\sigma }d\dot{x}^{i}~\ \ etc.
\end{equation}%
Accordingly, for any smooth function $f$ on $J^{r}Y$, we obtain:
\begin{equation}
hdf=d_{A}fdx^{A}=d_{i}fdx^{i}+\dot{d}_{i}fd\dot{x}^{i},  \label{eq:hdf2}
\end{equation}%
where $d_{i}f$ and $\dot{d}_{i}f$ represent total $x^{i}$- and, accordingly,
total $\dot{x}^{i}$-derivatives (of order $r+1$). Using (\ref{prop_h}) for $%
\rho =df,$ we find:%
\begin{equation}
\partial _{i}(f\circ J^{r}\gamma )=d_{i}f\circ J^{r+1}\gamma ,~\ \ \ \dot{%
\partial}_{i}(f\circ J^{r}\gamma )=\dot{d}_{i}f\circ J^{r+1}\gamma \,.
\label{total derivs Finsler}
\end{equation}%
Alternatively, one may use a nonlinear connection on $\mathcal{A}^{+}\subset
PTM^{+}$ (e.g., the canonical one), to introduce the \textit{total adapted
derivative} operators
\begin{equation}
\boldsymbol{\delta }_{i}:=d_{i}-G_{~i}^{j}\dot{d}_{j},
\label{total adapted derivative}
\end{equation}%
which help constructing manifestly covariant expressions. More precisely,
using these operators, we can write \eqref{eq:hdf2} as
\begin{equation}
hdf=(\boldsymbol{\delta }_{i}f)dx^{i}+(\dot{d}_{i}f)\delta \dot{x}^{i}\,.
\label{hdf_covariant}
\end{equation}%
If $f$ is a coordinate invariant scalar function, then $\mathbf{\delta }%
_{i}f $ and $\dot{d}_{i}f$ are d-tensor components.

\subsection{Fibered automorphisms}

Variations of sections and, accordingly, of actions, are given by
1-parameter groups of fibered automorphisms of $Y$. But, in the case of
Finsler spacetimes, these will also have to take into account the doubly
fibered structure of the configuration bundle $Y$. This is why we introduce:

\begin{definition}[Automorphisms of $Y$]
An \textit{automorphism} of a fibered manifold $(Y,\Pi ,PTM^{+})$ is a
diffeomorphism $\Phi :Y\rightarrow Y$ such that there exists a fibered
automorphism $\phi $ of $(PTM^{+},\pi _{M},M)$ with $\Pi \circ \Phi =\phi
\circ \Pi .$
\end{definition}

In particular, this means that there exists a diffeomorphism $\phi
_{0}:M\rightarrow M$ which makes the following diagram commute:%

\begin{equation}\label{eq:phi-phi0}
	\xymatrix{Y \ar^{\Phi}[r] \ar_{\Pi}[d] & Y \ar^{\Pi}[d]\\
		PTM^+ \ar^{\phi}[r] \ar_{\pi_M}[d] & PTM^+ \ar^{\pi_M}[d]\\
		M \ar^{\phi_0}[r] & M}
\end{equation}

Locally, a fibered automorphism of $Y$ is represented as:%
\begin{equation*}
\tilde{x}^{i}=\tilde{x}^{i}(x^{j}),~\ \overset{\cdot }{\tilde{x}}\overset{}{%
^{i}}=~\overset{\cdot }{\tilde{x}}\overset{}{^{i}}(x^{j},\dot{x}^{j}),~\ \ \
\ \tilde{y}^{\sigma }=\tilde{y}^{\sigma }(x^{i},\dot{x}^{i},y^{\mu }).
\end{equation*}%
An automorphism of $Y\ $is called \textit{strict} if it covers the identity
of $PTM^{+}$, i.e., $\phi =id_{PTM^{+}}.$

Generators of 1-parameter groups $\left\{ \Phi _{\varepsilon }\right\} $ of
automorphisms of $Y$ are vector fields $\Xi \in \mathcal{X}(Y)$ that are
projectable with respect to both projections $\Pi $ and $\pi _{M};$ in
fibered homogeneous coordinates, this is expressed as:%
\begin{equation}
\Xi =\xi ^{i}(x^{j})\partial _{i}+\dot{\xi}^{i}(x^{j},\dot{x}^{j})\dot{%
\partial}_{i}+\Xi ^{\sigma }(x^{j},\dot{x}^{j},y^{\mu })\tfrac{\partial}{\partial y^\sigma}.
\label{eq:XI}
\end{equation}%
In particular, strict automorphisms are generated by $\Pi $-vertical vector
fields $\Xi =\Xi ^{\sigma }(x^{j},\dot{x}^{j},y^{\mu })\partial/\partial y^\sigma$.

Given such a 1-parameter group $\left\{ \Phi _{\varepsilon }\right\} ,$ any
section $\gamma \in \Gamma (Y)$ is deformed into the section%
\begin{equation*}
\gamma _{\varepsilon }:=\Phi _{\varepsilon }\circ \gamma \circ \phi
_{\varepsilon }^{-1}.
\end{equation*}%
In first approximation, if $\gamma $ is locally represented as: $\gamma:\left( x^{i},\dot{x}^{i}\right) \mapsto \left( x^{i},\dot{x}^{i},y^{\sigma}\left( x^{i},\dot{x}^{i}\right) \right) ,$ then:
\begin{equation*}
\gamma _{\varepsilon }:\left( x^{i},\dot{x}^{i}\right) \mapsto \left( x^{i},%
\dot{x}^{i},y^{\sigma }\left( x^{i},\dot{x}^{i}\right) +\varepsilon (\tilde{\Xi}
^{\sigma }\circ J^{1} \gamma)|_{\left( x^{i},\dot{x}^{i}\right) }+\mathcal{O}(\varepsilon ^{2})\right),
\end{equation*}%
where
\begin{equation}
\tilde{\Xi}^{\sigma }:=(\Xi ^{\sigma }-\xi ^{i}y_{~,i}^{\sigma }-\dot\xi
^{i}y_{~\cdot i}^{\sigma }).  \label{Csi_tilde_Finslerian}
\end{equation}
The functions $\tilde{\Xi} \circ J^{1} \gamma $, defined on each local chart in the domain of $\gamma $, are commonly (though
in a somewhat sloppy manner) denoted by $\delta y^{\sigma }.$

The automorphisms $\Phi _{\varepsilon }:Y\rightarrow Y$ are prolonged into
automorphisms $J^{r}\Phi _{\varepsilon }$ of $J^{r}Y$ by the rule:
\begin{equation*}
J^{r}\Phi _{\varepsilon }(J_{(x,\dot{x})}^{r}\gamma ):=J_{\phi (x,\dot{x}%
)}^{r}\gamma _{\varepsilon }.
\end{equation*}
The generator of the 1-parameter group $\left\{ J^{r}\Phi _{\varepsilon
}\right\}$ is called the $r$-th prolongation of the vector field $\Xi $
and denoted by $J^{r}\Xi $ (see the Appendix for the precise coordinate
formula for $J^{1}\Xi $).

\subsection{Homogeneous geometric objects on $TM$ as sections\label{sec:Homog_objects_sections}}

In order to apply the apparatus of calculus of variations with Finslerian
geometric objects (e.g., Finsler function $L,$ metric tensor $g,$
nonlinear/linear connection, homogeneous d-tensors) as dynamical variables,
we will describe these geometric objects as \textit{sections} of fiber
bundles $\left( Y,\Pi ,PTM^{+}\right) $. \textit{A priori}, $k$-homogeneous Finslerian geometric objects are
(locally defined)\ sections  $f:\overset{\circ }{TM}\rightarrow \overset{%
	\circ }{Y}$ into some fiber bundle $\overset{\circ }{Y}$ sitting on $\overset%
{\circ }{TM}$ (e.g, a bundle of tensors, or a bundle of connections over $%
\overset{\circ }{TM}$ etc.). The key idea allowing us to reinterpret them as sections of a bundle $Y$ sitting over $PTM^{+}$, is that $k$-homogeneity can be
interpreted as \textit{equivariance}, with respect to the action of $\left(
\mathbb{R}_{+}^{\ast },\cdot \right) $ on the fiber bundle $\overset{\circ }{%
	Y}$ and, respectively, on the principal bundle $(\overset{\circ }{TM},\pi
^{+},PTM^{+},\mathbb{R}_{+}^{\ast }).$

The construction of the configuration bundle $\left( Y,\Pi ,PTM^{+}\right) $ follows essentially the same
line of reasoning as the one made in \cite[Sec 5.4]{Giachetta}, in the case
of principal connections and relies on factoring out the action of $\mathbb{R}_{+}^{\ast }$, from both the total space and the base of the original bundle $\overset{%
	\circ }{Y}$.

Consider a fiber bundle $\overset{\circ }{Y}~\overset{\overset{\circ }{\Pi }}%
{\rightarrow }~\overset{\circ }{TM},$ with typical fiber $Z.$ Noticing that $%
\overset{\circ }{Y}$ can be identified with the fibered product $\overset{%
	\circ }{TM}\times _{\overset{\circ }{TM}}\overset{\circ }{Y}$ (via the
isomorphism $(\overset{\circ }{\Pi },id_{\overset{\circ }{Y}})$ covering the
identity of $\overset{\circ }{TM}$), it will be convenient to abuse the notation by explicitly mentioning the base point of any element in $\overset{\circ}{Y}$. This means that we will identify elements $y\in \overset{\circ }{Y}$ as triples $\left( x,\dot{x},y\right) ,$ where $\left( x,\dot{x}\right) =~\overset{\circ }{\Pi }(y)$.
We assume that $(\mathbb{R}_{+}^{\ast },\cdot )$ acts on $\overset{\circ }{Y}
$ by fibered automorphisms:%
\begin{equation}\label{Haction1}
	H:\mathbb{R}_{+}^{\ast }\times ~\overset{\circ }{Y}~\rightarrow ~%
	\overset{\circ }{Y},~\ \ H(\alpha ,\cdot )= H_{\alpha
	}\in \textrm{Aut}(\overset{\circ }{Y})
\end{equation}%
as:%
\begin{equation}\label{Haction2}
	H_{\alpha }(x,\dot{x},y)=\left( x, \alpha \dot{x}, \alpha^{k}y \right) ,
\end{equation}%
for some fixed $k\in \mathbb{R}.$

In particular, the above rule means that:

\begin{itemize}
	\item Each automorphism ${H}_{\alpha }\in \textrm{Aut}(\overset{\circ }{Y})$
	covers the homothety $\chi _{\alpha }:\overset{\circ }{TM}\rightarrow
	\overset{\circ }{TM}$, defined in Definition~\ref{def:homot}.

	\item The action is free and proper (properness is proven by verifying that
	the mapping $f:\mathbb{R}_{+}^{\ast }\times \overset{\circ }{Y}~\rightarrow ~%
	\overset{\circ }{Y}\times \overset{\circ }{Y},$ $\left( \alpha ,x,\dot{x}%
	,y\right) \mapsto (x,\alpha \dot{x},\alpha ^{k}y,x,\dot{x},y)$ is proper;
	the latter holds as the projection of a compact subset of a Cartesian
	product onto each factor is compact).
\end{itemize}

\textbf{Note. }In the following, we do not assume a specific form of the
fiber $Z$ of $\overset{\circ }{Y},$ we just assume that, for a given $k,$
rescaling of fiber elements $y$ by the power $\alpha ^{k},$ $\forall \alpha
>0,$ makes sense; e.g., in the case of vector bundles over $\overset{\circ }{%
	TM},$ this makes sense for any $k\in \mathbb{R},$ whereas for bundles whose
fibers do not admit a rescaling of elements, one is forced to choose $k=0$. An important example of bundles $\overset{\circ }{Y}$ are the pullback bundles $\pi^*_{TM|\mathcal{D}}(\mathcal{T}^p_qM)$, whose sections are the anisotropic tensors introduced in Definition \ref{def:aniso}.

\bigskip

This way (see \cite{Lee}, Ch. 21), the space of orbits of the action $H$, i.e., the set:%
\begin{equation} \label{eq:Y}
	Y=\overset{\circ }{Y}_{/\sim },
\end{equation}%
where the equivalence relation $\sim $ is given by:%
\begin{equation*}
	(x,\dot{x},y)~\sim (x^{\prime },\dot{x}^{\prime },y^{\prime
	})~\Leftrightarrow ~\exists \alpha >0:~~(x^{\prime },\dot{x}^{\prime
	},y^{\prime })= H_{\alpha }(x,\dot{x},y),
\end{equation*}%
is a smooth manifold. Moreover, $\overset{\circ }{Y}$ becomes a principal
bundle over $Y,$ with fiber $\mathbb{R}_{+}^{\ast }$ and projection $%
\textrm{proj}_{Y}:\overset{\circ }{Y}~\rightarrow Y,$ $(x,\dot{x},y)\mapsto \lbrack x,%
\dot{x},y].$

\begin{theorem}

	\begin{enumerate}
		\item The manifold $Y=\overset{\circ }{Y}_{/\sim }$ is a fiber bundle over $%
		PTM^{+},$ with typical fiber $Z.$

		\item $k$-homogeneous sections $f:\mathcal{Q}\rightarrow \overset{\circ }{Y},
		$ where $\mathcal{Q}\subset \overset{\circ }{TM}$ \ is a conic subbundle,
		are in a one-to-one correspondence with local sections $\gamma :\mathcal{Q}%
		^{+}\rightarrow Y,$ where $\mathcal{Q}^{+}=\pi ^{+}(\mathcal{Q})\subset
		PTM^{+}.$
	\end{enumerate}
\end{theorem}

\begin{proof}

	\begin{enumerate}
		\item First, let us define the projection:%
		\begin{equation*}
			\Pi :Y\rightarrow PTM^{+},~\ \ [\left( x,\dot{x},y\right) ]\mapsto \lbrack
			(x,\dot{x})].
		\end{equation*}%
		This mapping is independent of the choice of representatives in the class $%
		[\left( x,\dot{x},y\right) ],$ as $\Pi [ \left( x,\alpha \dot{x}%
		,\alpha ^{k}y\right) ]=[(x,\alpha \dot{x})]=[(x,\dot{x})]=\Pi [ \left(
		x,\dot{x},y\right) ]$ and surjective.

		A local trivialization of $Y$ can be obtained using the principal bundle
		structures of both $\overset{\circ }{Y}$ over $Y$ and of $\overset{\circ }{TM%
		}$ over $PTM^{+}.$ More precisely, start with $\overset{\circ }{V}~=~\overset%
		{\circ }{\Pi }\overset{}{^{-1}}(U)\subset ~\overset{\circ }{Y},$ where $U\in
		\{U_{i},U_{i^{\prime }}\}\subset ~\overset{\circ }{TM}$ is a (small enough)
		coordinate neighborhood on which, say, $\dot{x}^{3}$ keeps a constant sign
		(as introduced in Section \ref{Section:PTM+}); then, $\overset{\circ }{V}$ is diffeomorphic
		to $U\times Z.$

		But, on the one hand, using the principal bundle structure of $(\overset{%
			\circ }{TM},\pi ^{+},PTM^{+},\mathbb{R}_{+}^{\ast }),$ the coordinate
		neighborhood $U$ is diffeomorphic to $U^{+}\times \mathbb{R}_{+}^{\ast },$
		where $U^{+}=\pi ^{+}(U)$ and, on the other hand, using the principal bundle
		structure of $\overset{\circ }{Y}$ over $Y$, the coordinate neighborhood $\overset{\circ }{V}$ is in its
		turn, diffeomorphic to $V\times \mathbb{R}_{+}^{\ast },$ where $%
		V:=\textrm{proj}_{Y}(\overset{\circ }{Y}).$

		This way, a trivialization of $\overset{\circ }{Y}$ can be written as
		follows:
		\begin{equation*}
		\xymatrix{
		V \times \mathbb{R}_{+}^{\ast} \ar[r] \ar_{\overset{\circ}{\Pi}}[d] & (U^{+} \times \mathbb{R}_{+}^{\ast})\times Z \ar^{\mathrm{proj}_1}[dl] \\
				U^{+}\times \mathbb{R}_{+}^{\ast} &
				}.
		\end{equation*}%
		A system of $\mathbb{R}_{+}^{\ast }$-adapted fibered coordinate functions on
		$\overset{\circ }{Y}$ is of the form $\left( x^{i},u^{\alpha },\dot{x}%
		^{3},z^{\sigma }\right) ,$ where $\left( x^{i},u^{\alpha }\right) ,$ with $%
		u^{\alpha }=\dfrac{\dot{x}^{\alpha }}{\dot{x}^{3}},$ $\alpha =0,1,2$ are
		coordinate functions on $U^{+}\subset PTM^{+}.$

		A trivialization of $Y$ is then obtained by discarding the $\mathbb{R}%
		_{+}^{\ast }$ factor in the above diagram:%
		\begin{equation*}
		\xymatrix{
				V \ar[r] \ar_{\Pi}[d] & U^{+}\times Z \ar^{\mathrm{proj}_1}[dl] \\
				U^{+} &
				}.
		\end{equation*}%
		It remains to show that $V=\Pi ^{-1}(U^{+}):$

		$\subset :$ If $[x,\dot{x},y]\in V=\mathrm{proj}_{Y}(\overset{\circ }{V}),$ then
		there exists a representative $(x,\dot{x},y)\in \overset{\circ }{V}=\overset{\circ }{\Pi }%
		\overset{}{^{-1}}(U)$ of the class $[x,\dot{x},y].$ But then, $\Pi _{0}(x,%
		\dot{x},y)=(x,\dot{x})\in U,$ which means $[(x,\dot{x})]=\Pi \lbrack x,\dot{x%
		},y]\in U^{+},$ i.e., $[x,\dot{x},y]\in \Pi ^{-1}(U^{+}).$

		$\supset :$ Starting with $[x,\dot{x},y]\in \Pi ^{-1}(U^{+}),$ we find that $%
		[(x,\dot{x})]\in U^{+};$ picking a representative $(x,\dot{x})\in U$ of the
		class $[(x,\dot{x})],$ the corresponding representative $(x,\dot{x},y)$ of
		the class $[(x,\dot{x},y)]$ belongs to $\overset{\circ }{\Pi }\overset{}{%
			^{-1}}(U)=\overset{\circ }{V}.$ That is, the class $[(x,\dot{x},y)]$ is in $%
		\textrm{proj}_{Y}(\overset{\circ }{V})=V.$

		Using the above trivialization, the corresponding fibered coordinates on $%
		V\subset Y$ are then obtained by discarding the $\dot{x}^{3}$ coordinate from the coordinates $(x^{i},u^{\alpha },\dot{x}^{3},z^{\sigma })$ on $\overset{\circ}{Y}$
		i.e.,%
		\begin{equation*}
			\psi =(x^{i},u^{\alpha },z^{\sigma }).
		\end{equation*}

		\item Let $f:\mathcal{Q}\rightarrow \overset{%
			\circ }{Y},$ $\left( x,\dot{x}\right) \mapsto f\left( x,\dot{x}\right) \in
		\overset{\circ }{Y}_{(x,\dot{x})}$ be a $k$-homogeneous section, i.e., $f\left(
		x,\alpha \dot{x}\right) =\alpha ^{k}f\left( x,\dot{x}\right) ,$ $\forall
		\alpha >0$ and define:%
		\begin{equation*}
			\gamma :\mathcal{Q}^{+}\rightarrow Y,~\ \ \gamma \lbrack (x,\dot{x})]=[x,%
			\dot{x},f(x,\dot{x})].
		\end{equation*}%
		The mapping $\gamma $ is independent of the choice of representatives $(x,%
		\dot{x})\in \lbrack (x,\dot{x})]$ by virtue of the $k$-homogeneity of $f.$
		Moreover, $(\Pi \circ \gamma )[(x,\dot{x})]=[(x,\dot{x})]$ for all $(x,\dot{x%
		})\in \mathcal{Q},$ which makes $\gamma $ a well defined local section of $Y.
		$

		The correspondence $f\mapsto \gamma $ is obviously injective. To prove
		surjectivity, pick an arbitrary $\gamma \in \Gamma (Y)$ and define $f(x,\dot x),$
		for every representative $\left( x, \dot x \right) \in \lbrack \left( x,\dot{x}%
		\right) ],$ as the third component $y$ of the representative $(x,\dot x,y)\in
		\gamma \lbrack (x,\dot{x})];$ then, $f(x,\alpha \dot x)=\alpha ^{k}y$ by the
		definition of equivalence classes in $Y,$ which means that $f$ is a $k$%
		-homogeneous section of $\overset{\circ }{Y}.$
	\end{enumerate}
\end{proof}

\bigskip

\textbf{Note: }Coming back to our discussion before Remark~\ref{remarkHomFibCoord}. Since we fixed the group action of $\mathbb{R}^*_+$, defined in \eqref{Haction1}, \eqref{Haction2}, on $Y$, on each fibered chart domain $V=\Pi ^{-1}(U^{+}),$ we can explicitly
introduce \textit{homogeneous fibered coordinates} as the local coordinates
\begin{equation*}
	\left( x^{i},\dot{x}^{i},y^{\sigma }\right) :=(x^{i},~\dot{x}^{3}u^{0}, ~\dot{x}^{3}u^{1}, ~\dot{x}^{3}u^{2}, \dot{x}^{3};~(\dot{x}^{3})^{k}z^{\sigma })
\end{equation*}
of an arbitrarily chosen representative of the class $[x,\dot{x},y]$  where $u^{0} = \frac{\dot x^{0}}{\dot x^{3}}$ etc. These are, obviously unique up to positive rescaling, i.e., $\left( x^{i},\dot{x}^{i},y^{\sigma }\right) $ and $(x^{i},\alpha \dot{x}^{i},\alpha^{k}y^{\sigma })$ will represent the same class.

\bigskip

\begin{enumerate}
	\item \textbf{Finsler functions }$L:\mathcal{A}\rightarrow \mathbb{R}.$ In
	this case, $\overset{\circ }{Y}~=\overset{\circ }{TM}\times \mathbb{R}$ is a trivial line bundle, which means the configuration bundle $Y=\overset{\circ }{Y}_{/\sim }$ is the space of orbits of the Lie group action $H:\mathbb{R}_{+}^{\ast }\times \overset{\circ }{Y}\rightarrow ~\overset{\circ }{Y}$
	given by the fibered automorphisms:%
	\begin{equation*}
		H_{\alpha }:\overset{\circ }{Y}\rightarrow ~\overset{\circ }{Y,}\ H_{\alpha
		}(x,\dot{x},y)=(x,\alpha \dot{x},\alpha ^{2}y),~\ \ \forall \alpha >0.
	\end{equation*}%
	This way, 2-homogeneous Finsler functions are identified with local sections
	$\gamma $
	\begin{equation*}
		L\mapsto \gamma \in \Gamma (Y),~\ \gamma \lbrack (x,\dot{x})]=[x,\dot{x},L(x,%
		\dot{x})].
	\end{equation*}%
	In homogeneous fibered coordinates, the class $[x,\dot{x},L(x,\dot{x})]$ is
	represented as $(x^{i},\dot{x}^{i},L(x,\dot{x})).$

	\item \textbf{The 1-particle distribution function} \textbf{of a kinetic
		gas, }see \cite{kinetic-gas,Hohmann:2020yia}, can be understood as a 0-homogeneous mapping \textbf{\ }$%
	\varphi :\mathcal{Q}\rightarrow \mathbb{R},$ defined on some conic subbundle
	$\mathcal{Q}\subset \overset{\circ }{TM}.$ Again $\overset{\circ }{Y}~=\overset{\circ }{TM}\times \mathbb{R}$ and the configuration
	bundle is $Y=(\overset{\circ }{TM}\times \mathbb{R)}_{/\sim },$ i.e.\ the space of orbits of the  Lie group action $H:\mathbb{R}_{+}^{\ast }\times \overset{\circ }{Y}\rightarrow ~\overset{\circ }{Y}$ is given by
	\begin{equation*}
	H_{\alpha }:\overset{\circ }{Y}\rightarrow ~\overset{\circ }{Y,}\ H_{\alpha}(x,\dot{x},y)=(x,\alpha \dot{x},y),~\ \ \forall \alpha >0.
	\end{equation*}%
	The corresponding section $\gamma :\mathcal{Q}^{+}\rightarrow Y,$ $\gamma
	\lbrack (x,\dot{x})]=[x,\dot{x},\varphi (x,\dot{x})]$ is represented in
	fibered homogeneous coordinates as: $\gamma :\left( x^{i},\dot{x}^{i}\right)
	\mapsto (x^{i},\dot{x}^{i},\varphi (x,\dot{x})).$

	\item $0$\textbf{-homogeneous metric tensors }$g:\mathcal{A}\rightarrow
	T_{2}^{0}(\overset{\circ }{TM})$, $g_{(x,\dot{x})}=g_{ij}(x,\dot{x}%
	)dx^{i}\otimes dx^{j}$ (which are thus treated as tensors of type (0,2) on $%
	\overset{\circ }{TM},$ operating on horizontal vector fields $X\in \Gamma
	(HTM),$ see Section \ref{sec:homog_geom_obj}), are obtained as sections $\gamma $ of the bundle $%
	Y=\overset{\circ }{Y}_{/\sim },$ where $\overset{\circ }{Y}=T_{2}^{0}(%
	\overset{\circ }{TM})$. The  Lie group action $H:\mathbb{R}_{+}^{\ast }\times \overset{\circ }{Y}\rightarrow ~\overset{\circ }{Y}$ is given by
	\begin{equation*}
	H_{\alpha }:\overset{\circ }{Y}\rightarrow ~\overset{\circ }{Y,}\ H_{\alpha}(x,\dot{x},y)=(x,\alpha \dot{x},y),~\ \ \forall \alpha >0.
	\end{equation*}%
	In fibered homogeneous coordinates (which are naturally induced by the
	coordinates $\left( x^{i}\right) $ on $M$), these sections are represented
	as $\gamma :\left( x^{i},\dot{x}^{i}\right) \mapsto (x^{i},\dot{x}^{i},g_{ij}(x,\dot{x})).$
\end{enumerate}

Homogeneous d-tensors of any rank and any homogeneity degree can be treated
similarly.

\section{Finsler field Lagrangians, action, extremals}
\label{sec:FinslerFieldTheories} Finally, we are in the position to
explicitly construct action based field theories on Finsler spacetimes. The
Finsler-related geometric notions have been introduced in Section \ref%
{sec:Fins}. Afterwards, we discussed the proper base manifold, $PTM^{+},$ for action integrals having homogeneous fields as dynamical variables in
Section \ref{sec:PTM+}, and we demonstrated that these homogeneous fields can
be understood as sections into fiber bundles over $PTM^{+}$ in Section \ref%
{sec:Fields}.

\subsection{Actions for fields as sections of $PTM^+$}

\label{ssec:ActFeq} We now display all necessary definitions needed for well
defined action based field theories on Finsler spacetimes.

\begin{definition}[Fields]
\label{def:fields} A homogeneous field on a Finsler spacetime $(M,L)$ is a local
section $\gamma $ of a fibered manifold $(Y,\Pi ,PTM^{+})$ over the positive projective tangent bundle $PTM^{+}$.
\end{definition}

\begin{definition}[Lagrangians]
\label{def:Lag} On a configuration bundle $(Y,\Pi ,PTM^{+})$ over a Finsler spacetime $(M,L)$, a \textit{Finsler field Lagrangian of order $r$} is a $\Pi^{r}$%
-horizontal $7$-form $\lambda^+ \in \Omega _{7}(J^{r}Y)$.
\end{definition}
This definition is a particular instance of the general definition of Lagrangians given in Appendix~\ref{app:natbun}.

In homogeneous fibered coordinates, any
Lagrangian on $Y$ can be expressed as:
\begin{equation}
\lambda^+ =\Lambda d\Sigma ^{+}=\mathcal{L}\mathrm{Vol}_{0}\,.
\label{eq:Lagrangian}
\end{equation}%
where $\Lambda =\Lambda (x^{i},\dot{x}^{i},y^{\sigma },y_{,i}^{\sigma
},...,y_{\cdot i_{1}...\cdot i_{r}}^{\sigma })$ is the Lagrange function and $d\Sigma ^{+}$ is an invariant volume form on an appropriately chosen open subset $\mathcal{Q}^{+} \subset PTM^{+}$; for instance, one can choose the canonical volume form \eqref{canonical volume form} on the set $\mathcal{A}_{0}^{+} \subset PTM^{+}$ of non-null admissible directions over a Finsler spacetime; in this case, we obtain the Lagrange density
\begin{equation}
\mathcal{L}=\Lambda \frac{|\det g|}{L^{2}}.  \label{relation_Lambda_L}
\end{equation}

\bigskip

\textbf{Note.} The pulled back form $J^{r}\gamma ^{\ast }\lambda ^{+}$ (where $\gamma \in \Gamma (Y)$) is a
differential form on $PTM^{+},$
hence, it must be invariant under positive rescaling in $\dot{x}$. In
coordinates, this becomes equivalent to the result below.

\begin{proposition}
In local homogeneous coordinates corresponding to any fibered chart $%
(V^{r},\psi ^{r})$ on $J^{r}Y$, any Finsler field Lagrangian function $\Lambda
:V^{r}\rightarrow \mathbb{R}$ must obey:
\begin{equation}
\dot{x}^{i}\dot{d}_{i}\Lambda =0\,.  \label{homogeneity condition Lagrangian}
\end{equation}
\end{proposition}

\begin{proof}
Pick an arbitrary section of $\Pi ,$ say, $\gamma :U\rightarrow Y,$ where $%
U\subset PTM^{+}$ is a local chart domain. The function $\Lambda \circ
J^{r}\gamma $ is then defined on a subset of $PTM^{+},$ hence, it must be
0-homogeneous in $\dot{x};$ that is,
\begin{equation*}
\dot{x}^{i}\dot{\partial}_{i}(\Lambda \circ J^{r}\gamma )=0.
\end{equation*}%
But, from (\ref{total derivs Finsler}), $\dot{\partial}_{i}(\Lambda \circ
J^{r}\gamma )=(\dot{d}_{i}\Lambda )\circ J^{r+1}\gamma .$ Substituting into
the above relation and taking into account the arbitrariness of $\gamma $,
we get the result.
\end{proof}

\bigskip

The action attached to the Lagrangian \eqref{eq:Lagrangian} and to a piece $%
D^{+}\subset PTM^{+}$ is the function $S_{D^{+}}:\Gamma (Y)\rightarrow
\mathbb{R},$ given by:
\begin{equation*}
S_{D^{+}}(\gamma )=\underset{D^{+}}{\int }J^{r}\gamma ^{\ast }\lambda ^{+}.
\end{equation*}

By Proposition \ref{diffeo_compact subsets}, such action integrals on timelike domains $D^{+}$ can
equivalently be understood as integrals over pieces $D\subset \mathcal{O}$,
i.e.\ as actions formulated on the observer space. The advantage in the
representation of the action as integrals on $PTM^+$, is that the domain of
the integral does not depend on the Finsler Lagrangian.

\bigskip

The preparation from the previous sections, in particular, the formulation of
fields as sections of a configuration bundle $(Y,\Pi ,PTM^{+})$, allows us
now to straightforwardly apply the coordinate-free formulation of the
calculus of variations for Finsler field Lagrangians.

The \textit{variation} of the action under the flow $\left\{ \Phi
_{\varepsilon }\right\} $ of a doubly projectable vector field $\Xi \in
\mathcal{X}(Y)$ is given by the Lie derivative, see Appendix \ref{app:A}:
\begin{equation}
\delta _{\Xi }S_{D^+}(\gamma )=\underset{D^{+}}{\int }J^{r}\gamma ^{\ast }%
\mathfrak{L}_{J^{r}\Xi }\lambda ^{+}.
\end{equation}

A field $\gamma \in \Gamma (Y),$ $[(x,\dot{x})]\mapsto \gamma \lbrack (x,%
\dot{x})]$ on a Finsler spacetime is a \textit{critical section\ }for $S,$
if for any piece $D^{+}\subset PTM^{+}$ and for any $\Pi $-vertical vector
field $\Xi $ such that $\mathrm{supp}(\Xi \circ \gamma )\subset $ $D^{+}$: $\delta
_{\Xi }S_{D}(\gamma )=0.$

For any Lagrangian $\lambda ^{+}\in \Omega _{7}(J_{r}Y),$ there exists (see \cite%
{Krupka-book}, or Appendix \ref{app:A}) a unique source
form $\mathcal{E}_{\lambda ^{+}}\in \Omega _{8}(J^{s}Y)$ with $s\leq 2r,$
called the \textit{Euler-Lagrange form} of $\lambda ^{+},$ such that:%
\begin{equation}
J^{r}\gamma ^{\ast }(\mathfrak{L}_{J^{r}\Xi }\lambda ^{+})=J^{s}\gamma
^{\ast }\mathbf{i}_{J^{s}\Xi }\mathcal{E}_{\lambda ^{+}}-d(J^{s}\gamma
^{\ast }\mathcal{J}^{\Xi }),  \label{first_var_Finslerian}
\end{equation}%
for some $\mathcal{J}^{\Xi }\in \Omega _{6}(J^{s}Y).$ The 6-form $\mathcal{J}%
^{\Xi }$ (which is interpreted as a \textit{Noether current}), is only
unique up to a total derivative; in integral form, the above relation reads:%
\begin{equation}
\underset{D^{+}}{\int }J^{r}\gamma ^{\ast }(\mathfrak{L}_{J^{r}\Xi }\lambda
^{+})=\underset{D^{+}}{\int }J^{s}\gamma ^{\ast }\mathbf{i}_{J^{s}\Xi }%
\mathcal{E}(\lambda ^{+})-\underset{\partial D^{+}}{\int }J^{s}\gamma
^{\ast }\mathcal{J}^{\Xi }\,.
\end{equation}%
In a local contact basis,\eqref{eq:contactForm}, $\mathcal{E}_{\lambda ^{+}}$ is thus given as:%
\begin{equation*}
\mathcal{E}_{\lambda ^{+}}=\mathcal{E}_{\sigma }\theta ^{\sigma }\wedge
\mathrm{Vol}_{0}.
\end{equation*}%
The precise meaning of the requirement that $\mathcal{E}_{\lambda ^{+}}$ is
a source form is that the interior product
\begin{equation*}
\mathbf{i}_{J^{s}\Xi }\mathcal{E}_{\lambda ^{+}}=(\tilde{\Xi}^{\sigma }%
\mathcal{E}_{\sigma })\mathrm{Vol}_{0},
\end{equation*}%
only depends on the functions $\tilde{\Xi}^{\sigma }=\Xi ^{\sigma }-\xi ^{i}y^{\sigma}_{,i}-\dot{\xi}^{i}y^{\sigma}_{\cdot i}$, not on higher order components of $\Xi$.

In order to identify the Euler-Lagrange form, $\Pi $\textit{-vertical}
variation vector fields $\Xi =\Xi ^{\sigma }\partial _{\sigma }$ are
sufficient. More general transformations will, yet, be used when determining
energy-momentum tensors.

The field equations of $\lambda _{+}$ are then given by
\begin{equation*}
\mathcal{E}_{\sigma }\circ J^{s}\gamma =0.
\end{equation*}

\subsection{Finsler gravity sourced by a kinetic gas}

As an example for a field theory on Finsler spacetimes we discuss in the jet bundle language the dynamics
of a Finsler spacetime sourced by a kinetic gas - a theory which is considered as an extension of general relativity \cite%
{Hohmann:2018rpp,Pfeifer-Wohlfarthgravity,kinetic-gas,Hohmann:2020yia}.

We fist discuss the purely geometric (vacuum) field theory, where the
Finsler function $L$ itself is the dynamical field, and then add a matter
Lagrangian as source of these dynamics.

\subsubsection{Finsler gravity Lagrangian}

\label{sssec:FinslerGrav} We have shown above that, for theories using the
2-homogeneous Finsler function $L:\mathcal{A}\rightarrow \mathbb{R}$ as the
dynamical variable, the appropriate configuration bundle is (\ref{eq:Y}), with fiber $\mathbb{R}$; we will re-denote it here as $\left( Y_{g},\Pi _{g},PTM^{+}\right) $ and the homogeneous coordinates
corresponding to a fibered chart on $Y_{g}$ by $(x^{i},\dot{x}^{i},{\hat{L}}%
)$, i.e., comparing to the notations in Section \ref{sec:Fields}, $y^1 = \hat{L}$. The hat is meant to distinguish the last coordinate function on $Y_g$ from mappings $L:\mathcal{A}\rightarrow
\mathbb{R},$ $L=L\left( x,\dot{x}\right) ,$ i.e., from \textit{components of
sections }$\gamma $ of the configuration bundle; more precisely, $L=\hat{L}%
\circ \gamma $.

Briefly, we have:%
\begin{equation*}
Y_{g}:=(\overset{\circ }{TM}\times \mathbb{R})/_{\sim },\quad
\Pi _{g}:[(x,\dot{x},{\hat{L}})]\mapsto \lbrack (x,\dot{x})].
\end{equation*}%
As already said above, we identify $2$-homogeneous functions $L:\mathcal{A}\rightarrow \mathbb{R}$
(where $\mathcal{A}\subset TM$ is a conic subbundle) with sections $\gamma
\in \Gamma (Y_{g}),$
\begin{equation*}
L\mapsto \gamma :\mathcal{A}^{+}\rightarrow Y_{g},  \gamma \lbrack (x,\dot{x}%
)]=[x,\dot{x},L(x,\dot{x})]\,.
\end{equation*}%
Locally, $\gamma $ is described as: $(x^{i},\dot{x}^{i})\mapsto (x^{i},\dot{x}^{i},L(x,\dot{x})).$

On $Y_{g},$ a Lagrangian of order $r$ is expressed in fibered homogeneous
coordinates as $\lambda ^{+}=\Lambda d\Sigma ^{+},$ where $d \Sigma^{+}$ is the canonical volume form \eqref{canonical volume form} and $\Lambda ~=\Lambda
(x^{i},\dot{x}^{i},{\hat{L}},{\hat{L}}_{,i},{\hat{L}}_{\cdot i},...{\hat{L}}%
_{\cdot i_{1}...\cdot i_{r}})$ becomes a 0-homogeneous function of $\dot{x}$ whenever evaluated along sections $\gamma \in \Gamma (Y)$. The local contact basis of $\Omega (J^{r}Y)$
is then denoted by $\{dx^{i},d\dot{x}^{i},\theta ,\theta _{,i},\theta
_{\cdot i},...,\theta _{\cdot i_{1}...\cdot i_{r-1}},d{\hat{L}}%
_{,i_{1}...,i_{r}},...d{\hat{L}}_{\cdot i_{1}....\cdot i_{r}}\}$, where the
(unique) first order contact form is:
\begin{equation}
\theta =d{\hat{L}}-{\hat{L}}_{,i}dx^{i}-{\hat{L}}_{\cdot i}d\dot{x}^{i}.
\label{contact form Y_g}
\end{equation}%
On $Y_{g},$ it is convenient to use \textit{formal} adapted derivatives:%
\begin{equation*}
\boldsymbol{\delta }_{i}:=d_{i}-G_{~i}^{j}\dot{d}_{j},
\end{equation*}%
where the word "formal" means that $G_{~i}^{j}\in \mathcal{F}(Y_{g})$ are considered as functions on a chart of the jet bundle $J^{r}Y$
 - constructed by the usual formula from the coordinate functions ${\hat{L}};{\hat{L%
}}_{,i};...;{\hat{L}}_{\cdot i\cdot j};$ i.e., only when evaluated along
\textit{sections }$\gamma ,$ they become the usual canonical nonlinear
connection coefficients, defined on charts of $TM$. In particular, we get:%
\begin{equation}
\boldsymbol{\delta }_{i}{\hat{L}=0.}  \label{adapted deriv_L}
\end{equation}%
Using the latter relation and (\ref{hdf_covariant}), the contact form $%
\theta $ can be written in a manifestly covariant form:%
\begin{equation}
\theta =d{\hat{L}}-{\hat{L}}_{\cdot i}\delta \dot{x}^{i}.
\label{contact form Y_g covariant}
\end{equation}%
Source forms on $J^{r}Y_{g}$ are locally expressed as:%
\begin{equation*}
\rho =f~\theta \wedge d\Sigma ^{+},
\end{equation*}%
where $f=f(x^{i},\dot{x}^{i},{\hat{L}},{\hat{L}}_{,i},{\hat{L}}_{\cdot i},...%
{\hat{L}}_{\cdot i_{1}...\cdot i_{r}}).$

On the bundle $Y_{g},$ the following Lagrangian is a natural (generally
covariant) one:%
\begin{equation*}
\lambda _{g}^{+}=R_{0}d\Sigma ^{+},
\end{equation*}%
where, again, $R_{0}$ is constructed by means of the usual formula, in terms of the the coordinate functions $\hat{L}, \hat{L}_{,i}$ etc.; using (\ref%
{geodesic_spray}), (\ref{def:nlin}) and (\ref{N_curvature_comps}), we find
that $\lambda _{g}^{+}$ contains fourth order derivatives of $L$, i.e., $\lambda
_{g}^{+}\in \Omega _{7}(J^{4}Y_{g}).$
Naturality of this Lagrangian follows taking into account that, along any section, both $R_{0}$ (which is an invariant scalar, since it is constructed using only operations with d-tensors) and $d\Sigma^{+}$ are invariant under coordinate changes on $TM$ induced by arbitrary coordinate changes on $M$, see also the Appendix \ref{app:Lagrangians} for a discussion of natural Lagrangians.

The Euler-Lagrange form of $\lambda _{g}^{+}$ is $\mathcal{E}_{\lambda
^{+}}=E\theta \wedge d\Sigma ^{+}$, where, \cite{Hohmann:2018rpp}:%
\begin{equation}
E=\frac{1}{2}g^{ij}(LR_{0})_{\cdot i\cdot
j}-3R_{0}-g^{Lij}(P_{i|j}-P_{i}P_{j}+(\nabla P_{i})_{\cdot j})\,,
\label{vacuum_eqn}
\end{equation}%
and thus the field equation, which determines the extremal points of the
action is: $E=0$.

\subsubsection{Kinetic gas Lagrangian}
\label{sssec:kingas}

It turned out that there exists a physical field which naturally couples to
Finsler geometry and can act as source for the dynamics of a Finsler
spacetime. This field is the 1-particle distribution function (1PDF) $%
\varphi $ of a kinetic gas, which describes the dynamics of a kinetic gas on
the tangent bundle of spacetime \cite{Sarbach:2013fya,Ehlers2011,Sarbach}.

Usually, the gravitational field of a kinetic gas is described in terms of
the Einstein-Vlasov equations \cite{Andreasson:2011ng}, which, however, only
take the \textit{averaged} kinetic energy of the particles constituting the
gas into account. By coupling the 1PDF of the kinetic gas directly to the
Finslerian geometry of spacetime, this averaging can be omitted; thus, the velocity distribution of the gas particles contributing to the
gravitational field can be fully taken into account \cite%
{kinetic-gas,Hohmann:2020yia}.

A kinetic gas is defined as a collection of a large number of particles, whose
properties are encoded into 1PDF, i.e.\ a function
\begin{equation*}
\mathit{\ }\varphi :\mathcal{O}\rightarrow \mathbb{R},~\varphi =\varphi
\left( x,\dot{x}\right) .
\end{equation*}%
Its interpretation is the following. The number of particles crossing a
given (6-dimensional) hypersurface $\sigma \subset \mathcal{O}$ is%
\begin{equation}
N[\sigma ]=\underset{\sigma }{\int }\varphi ~\textrm{vol},  \label{N}
\end{equation}%
where $\textrm{vol}=\dfrac{1}{3!}d\omega \wedge d\omega \wedge d\omega $ is the canonical invariant volume form on $\sigma $, determined
by the Lorentzian (or pseudo-Finslerian) structure on spacetime. This volume form induces a coupling between the geometry of spacetime and
the 1PDF. Prolonging $\varphi $ to $TM$ by 0-homogeneity as
discussed in Section \ref{sec:Homog_objects_sections}, we can
equivalently regard $\varphi $ as a function defined on $\mathcal{O}^{+}=\pi
^{+}\left( \mathcal{O}\right) \subset PTM^{+}$. The partial functions $%
\varphi _{x}=\varphi \left( x,\cdot \right) $ are all assumed to be
compactly supported (which is physically interpreted as the fact that the
speeds of the particles composing the gas have an upper bound lower than the
speed of light).

The Lagrangian defining the dynamics of the kinetic gas on a Finsler
spacetime is, see \cite{kinetic-gas},
\begin{equation}
\lambda _{m}^{+} = m\varphi d\Sigma ^{+} =m \varphi \frac{|\det(g)|}{\hat{L}^{2}}  \mathrm{Vol}_{0} =\mathcal{L}%
_{m}\mathrm{Vol}_{0} \label{eq:kgasLag}
\end{equation}%
where the Lagrange density  $\mathcal{L}_{m}=m \frac{|\det(g)|}{\hat{L}^{2}} \varphi$ depends on $x^{i},\dot{x}^{i},{%
\hat{L}},...,{\hat{L}}_{\cdot i \cdot j}$ and $m$ is the mass parameter of the gas
particles, here assumed all of the same mass for simplicity. The ${\hat{L}}%
,...,{\hat{L}}_{\cdot i\cdot j}$ dependence in the Lagrange density appears
due to the dependence of the volume form on the Finsler Lagrangian and on
the Finsler metric tensor $g$. Yet, for the sake of uniformness (since we
will couple it to $\lambda _{g}^{+},$ which lives on $J^{4}Y_{g}$), we will
regard $\mathcal{L}_{m}$ as a function on $J^{4}Y_{g},$ rather than on $%
J^{2}Y_{g}.$

Consider on $J^{4}Y_{g}$ the Lagrangian%
\begin{equation*}
\lambda ^{+}=\dfrac{1}{2\kappa ^{2}}\lambda _{g}^{+}+\lambda _{m}^{+},
\end{equation*}%
then, calculation of the Euler-Lagrange form by variation with respect to $L$ leads,
\cite{kinetic-gas,Hohmann:2020yia}, to the Finsler gravity equations sourced by a kinetic gas:%
\begin{equation}
\frac{1}{2}g^{ij}(LR_{0})_{\cdot i\cdot
j}-3R_{0}-g^{Lij}(P_{i|j}-P_{i}P_{j}+(\nabla P_{i})_{\cdot j})=-\kappa
^{2}\varphi\,, \label{complete field eqn}
\end{equation}
where $\kappa$ is the gravitational coupling constant.

We note that the above equation determines \textit{nonzero} values of $L$; accordingly, in the construction of the actions corresponding to $\lambda_{g}^{+}$ and $\lambda_{m}^{+}$, one must only consider non-lightlike domains for $L$. This is a difference from actions of metric field theories built directly over the spacetime manifold $M$, which do not distinguish between possible causal properties of vectors.

\section{Energy-momentum distribution tensor}
\label{sec:EMDistri} %Natural Finsler field Lagrangians.

An important concept in physics, which is derived from the action of a field theory, is the the energy-momentum tensor. One way to interpret the energy momentum tensor mathematically is that it measures \cite{Gotay}, ``the response of the matter Lagrangian to compactly supported diffeomorphisms of \textit{spacetime}''. This interpretation will be preserved in Finslerian field theory. In other words, naturality (general covariance, or general invariance \cite{Krupka-Trautman}) of Lagrangians will still be understood as invariance under (lifted) diffeomorphisms of \textit{spacetime} - though, in this case, the base of our configuration bundle is not spacetime, but its positively projectivized tangent bundle. This will require an extension of the technique presented in \cite{Voicu-em-tensors} and will result in a ``weaker'' (averaged) energy-momentum conservation law.

\subsection{Generally covariant Lagrangians}\label{ssec:GCL}

To identify the energy-momentum tensor in our construction of field theories
on Finsler spacetimes, we need some preparations:

\begin{enumerate}
\item \textbf{Lifts of diffeomorphisms $\phi_{0}$ of} $M$ \textbf{into doubly fibered automorphisms of} $Y$, that cover the natural lifts\footnote{Such lifts exist, e.g. when $Y$ is a bundle of $k$-homogeneous d-tensors, which is the $\mathbb{R}_{+}^{\ast}$-orbit space of a bundle $\overset{\circ}{Y}$ of d-tensors on $\overset{\circ}{TM}$. Diffeomorphisms $\phi_{0}$ of $M$ are naturally lifted into fibered automorphisms $d \phi_{0}$ of $TM$ and further, by tensor lifting to $\overset{\circ}{Y}$.} of $\phi_{0}$ to $PTM^{+}$, see the diagram \eqref{eq:phi-phi0}, since a priori diffeomorphisms of $M$ do not act on $Y$.

\item \textbf{A splitting of the total Lagrangian $\lambda^{+}$ of the theory into a background (vacuum) Lagrangian and a matter one} and, accordingly, of the variables of the theory into background and dynamical ones. The background Lagrangian (which we denote by $\lambda^{+}_{g}$) will only depend on the background variables (e.g., metric components, Finsler function etc), whereas the matter Lagrangian $\lambda^{+}_{m}$ will depend on all the variables, see \cite{Giachetta}. Roughly, denoting the background coordinates by $y^{\sigma}_B$ and \textit{non-background} or \textit{dynamical} variables $y^{\sigma}_D$, we have:
\begin{equation*}
\lambda^{+}(y^{\sigma}_B,..., y^{\sigma}_B{}_{,i...\cdot j},y^{\sigma}_D,...y^{\sigma}_D{}_{,i...\cdot j}) = \lambda^{+}_g(y^{\sigma}_B,..., y^{\sigma}_B{}_{,i...\cdot j}) + \lambda^{+}_m(y^{\sigma}_B,..., y^{\sigma}_B{}_{,i...\cdot j},y^{\sigma}_D,...y^{\sigma}_D{}_{,i...\cdot j})\,.
\end{equation*}
For instance, in general relativity, one has $y^{\sigma}_B = g^{ij}$, whereas $y^{\sigma}_D$ can be, e.g., the electromagnetic 4-potential.
The names "background" vs. "dynamical" come from the fact that the Lagrangian can be split into a part  $\lambda^{+}_{g}$, which only contains the background variables, and a part $\lambda^{+}_{m}$ which contains all information about the dynamical variables and their coupling to the background variables. Hence, even if one leaves $\lambda^{+}_g$ aside and fixes a value of the background fields, one can study the dynamics of the dynamical fields coupled to a fixed background.

\noindent Then, under the assumption that the matter Lagrangian $\lambda^{+}
_{m}$ is \textit{generally covariant} (see again the end of Appendix \ref%
{app:Lagrangians}), it will be invariant under any one-parameter group of canonical lifts of
diffeomorphisms of the spacetime manifold $M$, thus giving rise to
conserved Noether currents $\mathcal{J}^{\Xi }$ (where $\Xi =\mathfrak{F}%
(\xi _{0})$, is the canonical lift to $Y$ of a diffeomorphism generating
vector field $\xi _{0}$ from $M$). Roughly speaking, the energy-momentum
tensor will be given by the correspondence $\xi _{0}\mapsto \mathcal{J}^{\Xi
}.$
\end{enumerate}

In the case of Finsler spacetimes, the fundamental background variable is the Finsler
Lagrangian $L$ itself. Yet, the whole construction can be done in a completely similar manner, e.g., for
the Finsler metric tensor components $g_{ij},$ as background variables.

Consider a fibered product
\begin{equation*}
Y:=Y_{g}\times _{PTM^{+}}Y_{m}
\end{equation*}%
over $PTM^{+}$, where $Y_{g}=(\overset{\circ }{TM}\times \mathbb{R})_{/_{\sim }}$ was constructed in Section \ref{sssec:FinslerGrav}
and $Y_{m}$ is both a a fiber bundle over $PTM^{+}$ and a natural fiber bundle over $M$. In
particular, $Y$ has a double fibered manifold structure:%
\begin{equation*}
Y\overset{\Pi }{\longrightarrow }PTM^{+}\overset{\pi _{M}}{\longrightarrow }M
\end{equation*}%
We denote the homogeneous coordinates corresponding to a doubly fibered
chart on $Y$ by $(x^{i},\dot{x}^{i},{\hat{L}},y^{\sigma}_{D})$, where $y_B=\hat{L}$
is the coordinate on the fiber of $Y_{g}$ and $y^{\sigma}_{D}$ are local
coordinates on the fiber of $Y_{m}.$

As both $Y_m$ and $Y_g$ are natural bundles over M it follows that any vector field $%
\xi _{0}\in \mathcal{X}(M)$ admits a canonical lift $\Xi \in \mathcal{X}(Y).$

Consider a generally covariant Lagrangian $\lambda _{m}^{+}\in \Omega
(J^{r}Y):$%
\begin{equation*}
\lambda _{m}^{+}=\mathcal{L}_{m}(x^{i},\dot{x}^{i},{\hat{L}},{\hat{L}}_{,i},{%
\hat{L}}_{\cdot i},...,{\hat{L}}_{\cdot i_{1}...i_{r}},y^{\sigma
}_{D},...,y_{D\cdot i_{1}...i_{r}}^{\sigma })\mathrm{Vol}_{0},
\end{equation*}%
which will be interpreted as the \textit{matter Lagrangian} (as already mentioned above,
the total Lagrangian of the theory will be obtained as $\lambda
^{+}:=\lambda _{g}^{+}+\lambda _{m}^{+}$). Since $\lambda _{m}^{+}$ is
generally covariant, for any compactly supported vector field $\xi _{0}\in \mathcal{X}%
(M),$ $\lambda _{m}^{+}$ is invariant under the flow of the $r$-th jet
prolongation of the canonical lift $\Xi :=\mathfrak{F}(\xi _{0}),$ i.e.,%
\begin{equation}
\mathfrak{L}_{J^{r}\Xi }\lambda _{m}^{+}=0.
\label{invariance cond finslerian}
\end{equation}

In the following, we will explore in detail the consequences of this invariance of $\lambda_m^+$.

\subsection{The energy momentum distribution tensor and the energy momentum
density}

We will first give the technical precise definition of the energy-momentum distribution tensor, and demonstrate the concept on the example of the kinetic gas at the end of this subsection.

Assume $\left\{ \phi _{0,\varepsilon }\right\} $ is a 1-parameter group of
compactly supported diffeomorphisms of $M,$ generated by $\xi _{0}\in
\mathcal{X}(M)$, $\xi _{0}=\xi ^{i}\partial _{i}.$ Then:

\begin{enumerate}
\item Each $\phi _{0,\varepsilon }$ is first naturally lifted to $TM,$ as $%
\phi _{\varepsilon }:=d\phi _{0,\varepsilon }.$ The generator of $\left\{
\phi _{\varepsilon }\right\} $ is the complete lift $\xi \in \mathcal{X}(%
\overset{\circ }{TM})$ of $\xi _{0}:$%
\begin{equation}
\xi =\xi ^{i}\partial _{i}+\dot{\xi}^{i}\dot{\partial}_{i},~\ \ \ \ \dot{\xi}%
^{i}=\xi _{~,j}^{i}\dot{x}^{j}.  \label{csi_lift PTM+}
\end{equation}
Since the canonical lift $\xi$ is 0-homogeneous, we can identify it with a vector field on $PTM^{+}$ (more precisely, with its pushforward by $\pi^{+}$), see Section \ref{sssec:PTM+TM}.

\item Further, taking into account that $Y_{g}=(\overset{\circ }{TM}\times
\mathbb{R})_{/_{\sim }}$ is obtained as a quotient space of the \textit{%
trivial} bundle $\overset{\circ }{TM}\times \mathbb{R},$ the canonical lift, \cite{Giachetta} $%
\Phi _{g,\varepsilon }:Y_{g}\rightarrow Y_{g}$ of $\phi _{\varepsilon }$ is also a trivial one i.e., it acts on the fiber variable $\hat{L}$ as the identity:
\begin{equation*}
\Phi _{g,\varepsilon }[(x,\dot{x},{\hat{L}})]=[(\phi _{\varepsilon }(x,\dot{x%
}),{\hat{L}})];
\end{equation*}
The above mapping is well defined (i.e., independent on the choice of the
representative of the class $[(x,\dot{x},{\hat{L}})]$), due to the linearity
of $\phi _{\varepsilon }$ in $\dot{x}.$ As the lifted diffeomorphisms act trivially on $\hat{L}$, the generator $\xi $ is
canonically lifted into a vector field $\Xi _{g}\in \mathcal{X}(Y_{g}),$ with vanishing $\frac{\partial}{\partial \hat{L}}$ component,
i.e.:
\begin{equation*}
\Xi _{g}=\xi ^{i}\partial _{i}+\dot{\xi}^{i}\dot{\partial}_{i} + 0 \tfrac{\partial}{\partial \hat{L}}.
\end{equation*}

\item
According to our first assumption at the beginning of Section \ref{ssec:GCL}, there exists a canonical lift $\xi $ to $Y_{m}$, into some vector field $\Xi _{m}$ of the form $\Xi =\xi ^{i}\partial _{i}+\dot{\xi}^{i}\dot{\partial} _{i} + \Xi ^{\sigma } \frac{\partial}{\partial y^{\sigma }_{D}}.$ All in all, we
obtain that the canonical lift of $\xi _{0}\in \mathcal{X}(M)$ to the
fibered product $Y=Y_{g}\times _{PTM^{+}}Y_{m}$ is expressed in a fibered
chart by adding to $\xi$ the contributions describing the transformation of each of the fiber variables
\begin{equation}
\Xi =\xi ^{i}\partial _{i}+\dot{\xi}^{i}\dot{\partial} _{i} + 0\tfrac{\partial}{\partial \hat L} +\Xi ^{\sigma }\tfrac{\partial}{\partial y^{\sigma}_{D}},  \label{total lift csi}
\end{equation}
where, see \cite{Gotay}, $\Xi ^{\sigma }$ are functions of the coordinates $%
x^{i},\dot{x}^{i},y^{\sigma }_{D},...,y_{D\cdot i_{1}...i_{r}}^{\sigma }$ and of
a finite number of the derivatives of $\xi ^{i}.$
\end{enumerate}

\noindent \textbf{First variation formula.}

Accordingly, the Euler-Lagrange form $\mathcal{E}(\lambda _{m}^{+})$ will be
split into a $Y_{g}$ and a $Y_{m}$-component as%
\begin{equation*}
\mathcal{E}(\lambda _{m}^{+})=\mathcal{E}_{g}(\lambda _{m}^{+})+\mathcal{E}%
_{m}(\lambda _{m}^{+}),
\end{equation*}%
where:%
\begin{equation}
\mathcal{E}_{g}(\lambda _{m}^{+})=\dfrac{\delta \mathcal{L}_{m}}{\delta {%
\hat{L}}}\theta \wedge \mathrm{Vol}_{0},\ ~\ \ \ \ \ \mathcal{E}_{m}(\lambda
_{m}^{+})=\dfrac{\delta \mathcal{L}_{m}}{\delta y^{\sigma }_{D}}\theta ^{\sigma}_{D}\wedge\mathrm{Vol}_{0},  \label{E_g}
\end{equation}%
and $\theta^{\sigma}_{D}= \mathrm{d} y^{\sigma}_{D}-y^{\sigma}_{D,i}\mathrm{d}x^{i} - y^{\sigma}_{D\cdot i}\mathrm{d}\dot{x}^{i} $. Since $h%
\mathfrak{L}_{J^{r}\Xi }\lambda _{m}^{+}=0$ (which follows from the
invariance condition (\ref{invariance cond finslerian})), this leads to:%
\begin{equation*}
h\mathbf{i}_{J^{s}\Xi }\mathcal{E}_{g}(\lambda _{m}^{+})+h\mathbf{i}%
_{J^{s}\Xi }\mathcal{E}_{m}(\lambda _{m}^{+})-hd\mathcal{J}^{\Xi }=~0.
\end{equation*}%
But, on-shell for the variables $y^{\sigma }_{D},$ i.e., along sections $\gamma
:=(L,\gamma _{m})$ such that the ``matter field'', i.e. the section $\gamma
_{m}:PTM^{+}\rightarrow Y_{m},\left( x^{i},\dot{x}^{i}\right) \mapsto \left(
x^{i},\dot{x}^{i},y^{\sigma }_{D}\left( x^{i},\dot x^i\right) \right) $, is critical for $%
\lambda _{m}^{+},$ the $\mathcal{E}_{m}$-term above vanishes, i.e.:
\begin{equation}
h\mathbf{i}_{J^{s}\Xi }\mathcal{E}_{g}(\lambda _{m}^{+})-hd\mathcal{J}^{\Xi
}~\simeq _{\gamma _{m}}0,  \label{variation lambda_m}
\end{equation}%
where $\simeq _{\gamma _{m}}$ means equality on-shell for the matter
field $\gamma _{m}.$

\bigskip

\noindent \textbf{The energy-momentum distribution tensor.}

The surviving Euler-Lagrange component $h\mathbf{i}_{J^{s}\Xi }\mathcal{E}%
_{g}(\lambda _{m}^{+})$ in (\ref{variation lambda_m}) can again be split
into a linear expression in $\xi ^{i}$ and a divergence expression; the
latter will couple with $hd\mathcal{J}^{\Xi}$ into a boundary
term and will provide the building block of the energy-momentum
distribution tensor $\Theta$. More precisely,

\begin{lemma}
\label{Theta_B_Lemma} For any natural Finsler field Lagrangian $\lambda
_{m}^{+}\in \Omega _{7}(J^{r}Y),$ there exist unique $\mathcal{F}(M)$%
-linear mappings $\Theta :\mathcal{X}(M)\rightarrow \Omega (J^{s}Y),$ $\mathcal{B}:%
\mathcal{X}(M)\rightarrow \Omega (J^{s+1}Y),$ with $\Pi ^{s}$
(respectively, $\Pi ^{s+1}$)-horizontal values (where $s\leq 2r)$ such that,
for any $\xi _{0}\in \mathcal{X}(M)$:%
\begin{equation}
h\mathbf{i}_{J^{s}\Xi }\mathcal{E}_{g}(\lambda _{m}^{+})=\mathcal{B}(\xi
_{0})+hd\Theta (\xi _{0}).  \label{Theta_B_splitting}
\end{equation}
\end{lemma}

\begin{proof}
We will first construct $\Theta$ and $\mathcal{B}$ in a fibered chart and then, show that the obtained expressions are independent on the choice of this chart. In any fibered chart, $\mathcal{E}_{g}$ is expressed as:%
\begin{equation}
\mathcal{E}_{g}(\lambda _{m}^{+})=\dfrac{\delta \mathcal{L}_{m}}{\delta {%
\hat{L}}}~\theta \wedge \mathrm{Vol}_{0}=: - \frac{1}{2} \mathfrak{T~}\hat{L}^{-1} \theta \wedge d\Sigma
^{+},  \label{E_g_1}
\end{equation}%
where $\mathfrak{T}$ is a $0$-homogeneous scalar which acts as source term for Finsler gravity equations~\eqref{vacuum_eqn}, and the factor $\hat L^{-1}$ is introduced to ensure this degree of  homogeneity (as both $\hat{L}^{-1} \theta$ and $d\Sigma^{+}$ are 0-homogeneous.) The precise expression of $\mathfrak{T}$ depends on the chosen volume form. For instance, if $d \Sigma^{+}$ is the canonical volume form (\ref{canonical volume form}), then:
\begin{equation}
\mathfrak{T}=-2\dfrac{{\hat{L}}^{3}}{\left\vert \det g\right\vert }\dfrac{%
\delta \mathcal{L}_{m}}{\delta {\hat{L}}}.  \label{tau}
\end{equation}%
Since $\lambda ^{+}$ is a natural Lagrangian, $\dfrac{\delta \mathcal{L}_{m}%
}{\delta {\hat{L}}}$ is a scalar density and, accordingly, $\mathfrak{T}$ is
a scalar invariant. Then, $\mathbf{i}_{J^{s}\Xi }\mathcal{E}_{g}(\lambda
_{m}^{+})= (-\frac{1}{2}\mathfrak{T}\hat{L}^{-1}\mathbf{i}_{J^{s}\Xi }\theta )d\Sigma ^{+}+ \frac{1}{2}(\mathfrak{T%
}\hat{L}^{-1}) \theta \wedge \mathbf{i}_{J^{s}\Xi } d\Sigma ^{+};$ since the last term is a
multiple of $\theta$, it is a contact form; the remaining component is thus
the horizontal component $h\mathbf{i}_{J^{s}\Xi }\mathcal{E}_{g}(\lambda
_{m}^{+})$ and can be expressed (up to a pullback by $\Pi ^{s+1,s}$ of the
right hand side) as:%
\begin{equation}
h\mathbf{i}_{J^{s}\Xi }\mathcal{E}_{g}(\lambda _{m}^{+})=- \frac{1}{2}(\mathfrak{T~}%
\hat{L}^{-1} \mathbf{i}_{J^{s}\Xi }\theta )d\Sigma ^{+}.  \label{hE_g_rough}
\end{equation}%
Further, using $\theta =d{\hat{L}}-\hat{L}_{,i}dx^{i}-{\hat{L}}_{\cdot i}d%
\dot{x}^{i}=d{\hat{L}}-\boldsymbol{\delta }_{i}\hat{L}dx^{i}-{\hat{L}}_{\cdot i}%
\mathbf{\delta }\dot{x}^{i}$ and $\boldsymbol{\delta }_{i}\hat{L}=0,$ we find:
\begin{equation*}
\theta =d{\hat{L}}-{\hat{L}}_{\cdot i}\delta \dot{x}^{i}=d{\hat{L}}-2\dot{x}_{i}(d\dot{x}^{i}+G_{~j}^{i}dx^{j});
\end{equation*}%
inserting into $\theta $ the lift (\ref{total lift csi}) of $\xi _{0},$ this
becomes:%
\begin{equation*}
\mathbf{i}_{J^{r}\Xi }\theta =-2\dot{x}_{i}(\dot{\xi}^{i}+G_{~j}^{i}\xi
^{j})=-2\dot{x}_{i}(\xi _{~,j}^{i}\dot{x}^{j}+G_{~j}^{i}\xi ^{j})=-2\dot{x}%
_{i}\nabla \xi ^{i},
\end{equation*}%
where in the second equality we used: $\dot{\xi}^{i}=\xi _{~,j}^{i}\dot{x}%
^{j}.$ We can thus rewrite (\ref{hE_g_rough}) as:%
\begin{equation*}
h\mathbf{i}_{J^{s}\Xi }\mathcal{E}_{g}(\lambda _{m}^{+})= \mathfrak{T~}\hat{L}^{-1}\dot{%
x}_{i}\nabla \xi ^{i}d\Sigma ^{+}.
\end{equation*}%
Taking into account that $\nabla \dot{x}_{i}=0$ and $\nabla \hat{L}=0$, this can be uniquely split
into a linear term in $\xi ^{i}$ and the divergence of a linear term in $\xi
^{i}$:
\begin{equation}
h\mathbf{i}_{J^{s}\Xi }\mathcal{E}_{g}(\lambda _{m}^{+})= [\nabla (\mathfrak{T} \hat{L}^{-1}\dot{x}_{i}\xi ^{i}) - \xi ^{i}\dot{x}_{i}\hat{L}^{-1}\nabla \mathfrak{T}]d\Sigma ^{+}.
\label{splitting E_g}
\end{equation}%
Then, using $\dot{x}_{~|j}^{i}=0$ and $\nabla =\dot{x}^{j}D_{\mathbf{\delta }_{j}},$ we can rearrange the divergence term as
\begin{equation}
 \nabla (\mathfrak{T}\hat{L}^{-1}\dot{x}%
_{i}\xi ^{i}) = (\mathfrak{T}\hat{L}^{-1}\dot{x}_{i}\dot{x}^{j}\xi ^{i})_{|j},
\end{equation}
which suggests the notation:
\begin{equation}
\Theta _{~i}^{j}:=\mathfrak{T}\hat{L}^{-1}\dot{x}^{j}\dot{x}_{i}  \label{def Theta local}
\end{equation}%
As $\mathfrak{T}$ is a scalar invariant, the functions $\Theta _{~i}^{j},$
defined on the given fibered chart, transform under induced fibered
coordinate changes as the components of a tensor on $M$ (equivalently, as
d-tensor components on $TM$). Also, noticing that the last term in (\ref{splitting E_g}) can be written as:
$ \xi ^{i}\dot{x}_{i}\hat{L}^{-1}\nabla \mathfrak{T} = \xi ^{i} \dot{x}_{i}\hat{L}^{-1}\dot{x}^{j}\mathfrak{T}_{|j} = \Theta_{~i|j}^{j}\xi ^{i} $, this suggests to introduce the mappings $\Theta :\mathcal{X}(M)\rightarrow \Omega _{6}(J^{s}Y)$, $\mathcal{B}:\mathcal{X}%
(M)\rightarrow \Omega _{7}(J^{s+1}Y)$ given by
\begin{eqnarray}
\Theta (\xi _{0}) &=&(\Theta _{~i}^{j}\xi ^{i})\mathbf{i}_{\delta
_{j}}d\Sigma ^{+},  \label{def_theta} \\
\mathcal{B}(\xi _{0}) &=&-\Theta _{~i|j}^{j}\xi ^{i}d\Sigma ^{+},
\label{def_B}
\end{eqnarray}%
(where $\xi _{0}=\xi ^{i}\partial _{i}$). These mappings are well defined,
i.e., independent on the chosen coordinate charts; moreover,
they have $\Pi ^{s}$ (respectively, $\Pi ^{s+1}$)-horizontal values, they
are both linear in $\xi $ and obey (\ref{Theta_B_splitting}), which
completes the proof of the existence. Uniqueness of $\mathcal{B}$ and $%
\Theta $ follows from the uniqueness of the splitting (\ref{splitting E_g})
and the arbitrariness of $\xi ^{i}$.
\end{proof}

\bigskip

\textbf{Note.} The proof of the above result is based on a similar idea to the one of Lemma 2 in \cite{Voicu-em-tensors}. The essential difference, in the Finslerian case, is that naturality of Lagrangians is based on the group of diffeomorphisms of $M$ (and not of $PTM^{+}$ as one would have expected following \cite{Voicu-em-tensors}), i.e., naturality comes from a manifold of lower dimension than the one of the base space of our configuration manifold $Y$. This will result, as we will see below, in a ``weaker'' (averaged) form of the energy-momentum balance law.

Actually, taking into account (\ref{def Theta local}), in homogeneous
fibered coordinates, $\Theta $ is expressed as:
\begin{equation}
\Theta =\Theta _{~j}^{i}dx^{j}\otimes \mathbf{i}_{\mathbf{\delta }%
_{i}}d\Sigma ^{+}=\mathfrak{T}(\hat{F}_{\cdot j}dx^{j})\otimes \mathbf{i}_{l^{i}\mathbf{%
\delta }_{i}}d\Sigma ^{+},  \label{local expression Theta}
\end{equation}%
where $\hat{F}= \sqrt{|\hat{L}|}$; a quick computation shows that $\hat{L}^{-1}\dot{x}^{i}\dot{x}_{j} =\hat{F}_{\cdot j} l^{i},$ regardless of the sign of $\hat{L}$. Equivalently, in a coordinate-free writing:%
\begin{equation}
\Theta =\mathfrak{T}\omega ^{+}\otimes \mathbf{i}_{\ell }d\Sigma ^{+},
\label{Theta coordinate free}
\end{equation}%
where we have identified, by abuse of notation, the Reeb vector field $\ell
=l^{i}\delta _{i}\in \mathcal{X}(\mathcal{A}_{0}^{+})$ with the vector field
on $J^{s+1}Y$ obtained by replacing $\delta _{i}$ with the formal total
adapted derivative $\boldsymbol{\delta }_{i},$ i.e., with: $l^{i}\boldsymbol{\delta }_{i}\in \mathcal{X}(J^{s+1}Y)$. In the same fashion, the values $\omega _{%
\left[ \left( x,\dot{x}\right) \right] }^{+}$ of the mapping $\omega ^{+}:%
\mathcal{A}_{0}^{+}\rightarrow \Omega _{1}\left( M\right) $ are identified
with their pullbacks to $J^{s+1}Y.$

\begin{definition}[Energy-momentum distribution tensor]
The energy-momentum distribution tensor associated to a natural Lagrangian $%
\lambda _{m}^{+}$ on a bundle $Y=Y_{g}\times _{PTM^{+}}Y_{m}$, which is natural over a Finsler spacetime $M$,
is the $\mathcal{F}(M)$-linear mapping $\Theta :\mathcal{X}(M)\mapsto \Omega
_{6}(J^{s}Y)$ defined by (\ref{Theta coordinate free}).
\end{definition}

\begin{definition}[Energy-momentum scalar]
	We call the function $\mathfrak{T}: \mathcal{A}_{0}^{+} \rightarrow \mathbb{R}$, defined by the relation \eqref{E_g_1},  and explicitly given by \eqref{tau}, the energy-momentum scalar.
\end{definition}

We will call the $\mathcal{F}(M)$-linear mapping $\mathcal{B}:\mathcal{X}%
(M)\mapsto \Omega _{7}(J^{s+1}Y)$ defined by (\ref{def_B}), the \textit{%
balance function}, as energy-momentum conservation (or energy-momentum
balance) law is naturally characterized in terms of $\mathcal{B},$ as we
will see below.

\bigskip

\noindent \textbf{Averaged energy-momentum conservation law.}

Consider, in the following, local sections $\gamma\in \Gamma (Y)$ such
that  $\mathrm{supp}(J^s\gamma^*\lambda_m^+)\subset \mathcal{T}^+$.
This way, it makes sense to integrate the form $J^{s}\gamma
^{\ast }\mathbf{i}_{J^{s}\Xi }\mathcal{E}_{g}(\lambda _{m}^{+})$ on the
entire set $\mathcal{T}_{x}^{+}=\mathcal{O}_{x}^{+}$ of timelike directions
at $x.$

Consider a piece $D_{0}\subset M$ and denote by
\begin{equation*}
\mathcal{T}^{+}(D_{0}):=\underset{x\in D_{0}}{\cup }\mathcal{T}_{x}^{+}=%
\underset{x\in D_{0}}{\cup }\mathcal{O}_{x}^{+},
\end{equation*}%
the set of all timelike (equivalently, of observer) directions corresponding
to $D_{0}.$ Then, (\ref{Theta_B_splitting})\ becomes, with $\gamma :=\left(
L,\gamma _{m}\right) $:%
\begin{equation}
\underset{\mathcal{T}^{+}(D_{0})}{\int }J^{s}\gamma ^{\ast }\mathbf{i}%
_{J^{s}\Xi }\mathcal{E}_{g}(\lambda _{m}^{+})=\underset{\mathcal{T}%
^{+}(D_{0})}{\int }J^{s+1}\gamma ^{\ast }\mathcal{B}(\xi _{0})+\underset{%
\partial \mathcal{T}^{+}(D_{0})}{\int }J^{s}\gamma ^{\ast }\Theta (\xi _{0}).
\label{theta_B_splitting integral}
\end{equation}%
But, on-shell for $\gamma _{m},$ we have, according to (\ref{variation
lambda_m}): $J^{s}\gamma ^{\ast }\mathbf{i}_{J^{s}\Xi }\mathcal{E}%
_{g}(\lambda _{m}^{+})-J^{s}\gamma ^{\ast }d\mathcal{J}^{\Xi }~\simeq
_{\gamma _{m}}0;$ substituting into the above relation, this gives:
\begin{equation}
\underset{\mathcal{T}^{+}(D_{0})}{\int }J^{s+1}\gamma ^{\ast }\mathcal{B}%
(\xi _{0})+\underset{\partial \mathcal{T}^{+}(D_{0})}{\int }J^{s}\gamma
^{\ast }(\Theta (\xi _{0})-\mathcal{J}^{\Xi })~\simeq _{\gamma _{m}}0.
\label{integral variation lambda_m}
\end{equation}%
We are now able to prove the following result.

\begin{theorem}
Consider a bundle $Y_{m}$ over $PTM^{+},$ which is natural over $M$, and an arbitrary section $%
\gamma =({L},\gamma _{m})\in \Gamma (Y_{g}\times _{PTM^{+}}Y_{m})$ such that $\mathrm{supp}(J^s\gamma^*\lambda_m^+)\subset \mathcal{T}^+ $, then:
\begin{enumerate}
	\item \textbf{Averaged energy-momentum conservation law: }At any $x\in M$ and in any corresponding fibered chart:
	\begin{equation}
	\underset{\mathcal{T}_{x}^{+}}{\int }(\Theta _{~i|j}^{j}\circ J^{s+1}\gamma
	)d\Sigma _{x}^{+}=0,  \label{avg conservation law}
	\end{equation}%
	where $d\Sigma ^{+}=:d^{4}x\wedge d\Sigma _{x}^{+}.$

	\item $\Theta \left( \xi _{0}\right) $ \textbf{is a ``corrected Noether
		current''}, i.e., for any $\xi _{0}\in \mathcal{X}(M)$%
	\begin{equation}
	\underset{\partial \mathcal{T}^{+}(D_{0})}{\int }J^{s}\gamma ^{\ast }\Theta
	(\xi _{0})=\underset{\partial \mathcal{T}^{+}(D_{0})}{\int }J^{s}\gamma
	^{\ast }\mathcal{J}^{\Xi },  \label{Theta_Noether current}
	\end{equation}%
	where $\Xi $ denotes the canonical lift of $\xi _{0}$ to $Y.$
\end{enumerate}
\end{theorem}

\begin{proof}
\begin{enumerate}
	\item Fix $x_{0}\in M.$ Consider an arbitrary piece $D_{0}\subset M$
	containing $x_{0}$ as an interior point and an arbitrary $\xi _{0}\in
	\mathcal{X}(M)$ with support contained in $D_{0}.$

	Now, let us have a look at the boundary term in (\ref{integral variation lambda_m}). Since the support of the integrand, at every $x\in M$, is strictly contained in $\mathcal{T}^+_{x}$,
	 the only possible nonzero values are obtained at points $\left[\left( x,\dot{x}\right) \right] $ with $x\in \partial D_{0}.$ But, at these
	points, $\xi _{0}$ identically vanishes (hence also $\Xi =0$, as $\Xi$ is built from $\xi$ and its derivatives), which means
	that this boundary term is actually zero. It follows:
	\begin{equation}
	\underset{\mathcal{T}^{+}(D_{0})}{\int }J^{s+1}\gamma ^{\ast }%
	\mathcal{B}(\xi _{0})~\simeq _{\gamma _{m}}0.
	\end{equation}

	In coordinates, this is:
	\begin{equation*}
	\underset{\mathcal{T}^{+}(D_{0})}{\int }(\Theta _{~i|j}^{j} \circ J^{s+1}\gamma) \xi ^{i}d\Sigma
	^{+}\simeq _{\gamma _{m}}0.
	\end{equation*}

	Squeezing $D_{0}$ around $x_{0}$ such that $D_{0}$ is contained into a
	single chart domain, the above integral can be written as an iterated
	integral $\underset{D_{0}}{\int }\xi ^{i}(\underset{\mathcal{T}_{x}^{+}}{%
		\int }(\Theta _{~i|j}^{j}\circ J^{s+1}\gamma) d\Sigma _{x}^{+})d^{4}x,$ which, taking into
	account the arbitrariness of $\xi ^{i}$, leads to the result.

	\item follows then immediately from (\ref{integral variation
		lambda_m}) and 1.
\end{enumerate}
\end{proof}

Relation (\ref{Theta_Noether current}) says that, the energy-momentum tensor
$\Theta (\xi _{0})$ is, at least up to a term which does not contribute to
the integral (\ref{Theta_Noether current})), the conserved Noether current $%
\mathcal{J}^{\Xi }$ - i.e. (see also \cite{Gotay}), it gives the correct
notions of energy and momentum of the system under discussion.

\begin{remark}
	Taking into account that $\mathcal{O}_{x}^{+}=\mathcal{T}%
	_{x}^{+}$, the averaged conservation law can be rewritten as:%
	\begin{equation}
		\underset{\mathcal{O}_{x}^{+}}{\int }(\Theta _{~i|j}^{j}\circ J^{s+1}\gamma
		)d\Sigma _{x}^{+}=0.  \label{conservation law O}
	\end{equation}
\end{remark}

It is worth noting that, due to the fact that naturality of Lagrangians comes from $M$, which is a space of lower dimension than the one of the space $PTM^{+}$ on which the action integral is considered, in the above relation, integration over $\mathcal{O}_{x}^{+}$ (or, equivalently, $\mathcal{T}_{x}^{+}$) cannot be removed, i.e., we can typically only establish an \textit{averaged} conservation law. This is a distinctive feature of Finslerian field theory.

\bigskip

\noindent\textbf{Energy-momentum density on $M$}

The mapping $\Theta :\mathcal{X}(M)\rightarrow \Omega (J^{s}Y)$ gives rise
to an energy-momentum tensor\textit{\ density} on $M$, by averaging over
observer (or timelike) directions $\mathcal{O}_{x}^{+}=\pi ^{+}(\mathcal{O}%
_{x}).$ Consider an arbitrary fibered chart on $Y$;  $\Theta (\xi )=\Theta
_{~j}^{i}\xi ^{j}\otimes \mathbf{i}_{\delta _{i}}d\Sigma ^{+}$. Then, for
any section $\gamma \in \Gamma(Y)$ such that $\mathrm{supp}(J^s\gamma^*\lambda_m^+)\subset \mathcal{T}^+ $, set
\begin{equation}
\mathcal{T}_{~j}^{i}(x):=\underset{\mathcal{O}_{x}^{+}}{\int }
(\Theta _{~j}^{i} \circ J^{s}\gamma)_{|(x,\dot{x})}d\Sigma _{x}^{+},~\ \ \ \forall x\in M\,.
\label{em_tensor_density}
\end{equation}%
Under the above assumption this integral is
finite, so the result is well defined. Moreover, given the expression of $%
d\Sigma _{x}^{+},$ the functions $\mathcal{T}_{~j}^{i}(x)$ represent the
components of a tensor density on $M.$

\bigskip

\noindent \textbf{Example: the energy momentum distribution tensor of a kinetic gas.}

The kinetic gas example, which motivated the whole above construction, has been previously presented from a somewhat pedestrian perspective. In \cite[Eqs. (42)-(43) ]{kinetic-gas}, the maps $\Theta$ and $B$, can be read off. We briefly identify these maps here from the more abstract and mathematically precise construction we presented.

In the case of a kinetic gas discussed in Section \ref{sssec:kingas}, the kinetic gas Lagrangian \eqref{eq:kgasLag} is given by $\lambda _{m}^{+}=m\varphi d\Sigma ^{+},$ where $\varphi$ is the 1-particle
distribution function, reinterpreted as a function of $x,\dot{x},L$ and its
derivatives,
\begin{equation*}
\varphi (x,\dot{x})=f(x,\dot{x},L(x,\dot{x}),...,L_{\cdot i\cdot j}(x,\dot{x}))
\end{equation*}
and $d \Sigma^{+}$ is chosen as the canonical volume form \eqref{canonical volume form}.
Varying $\lambda _{m}^{+}$, we use \eqref{tau} to obtain $\mathfrak{T}:=\dfrac{1}{2}m\varphi;$
accordingly, the energy-momentum tensor distribution $\Theta $ has the local
components, compare to \eqref{def Theta local},
\begin{equation*}
\Theta _{~j}^{i}=\dfrac{1}{2L}m \varphi~\dot{x}^{i}\dot{x}_{j}.
\end{equation*}%
For any kinetic gas, the averaged conservation law (\ref{conservation law O}%
) holds.

In particular, for collisionless gases, it is known that $\varphi $ is
subject to the \textit{Liouville equation} $\nabla \varphi =0,$ equivalently:%
\begin{equation*}
D_{\ell}\varphi =0.
\end{equation*}
Taking into account that $l_{~|j}^{i}=0$, we notice that the Liouville
equation is nothing else than a \textit{pointwise} \textit{covariant
conservation law }of $\Theta $:%
\begin{equation*}
D_{\delta _{i}}\Theta _{~j}^{i}=0.
\end{equation*}

\noindent \textbf{Particular case: Lorentzian spaces.}

On a Lorentzian manifold $(M,a)$, the quantities%
\begin{equation}
T_{~j}^{i}=\dfrac{1}{\sqrt{\left\vert \det a\right\vert }}\mathcal{T}%
_{~j}^{i}  \label{riemannian_em_tensor}
\end{equation}%
represent the components of a tensor of type (1,1) on $M,$ and their
Levi-Civita covariant derivatives are, \cite{Crampin}, just the integrals of
the Chern covariant derivatives of $\Theta :$ $T_{j;i}^{i}=(\sqrt{\left\vert
\det a\right\vert })^{-1}\underset{\mathcal{O}_{x}^{+}}{\int }J^{s}\gamma
^{\ast }\Theta _{~j|i}^{i}(x,\dot{x})d\Sigma _{x}.$ Hence, the
energy-momentum conservation law (\ref{avg conservation law}) reads%
\begin{equation*}
T_{~j;i}^{i}=0.
\end{equation*}

In the particular case of kinetic gases on a Lorentzian spacetime, our
expression (\ref{em_tensor_density}) of the energy-momentum density agrees
to the known one, see \cite{Sarbach}.

\bigskip

It is important to note that, in general Finsler spacetimes, we have no
metric tensor on $M,$ hence (\ref{riemannian_em_tensor}) makes no sense. All
we can get is an energy-momentum tensor\textit{\ density }on $M,$ by
averaging over observer directions as in (\ref{em_tensor_density}) and,
accordingly, the conservation law (\ref{conservation law O}) of the
energy-momentum distribution $\Theta $.

\section{Summary and Outlook}
In this article we have proposed a general framework for action based field theories on Finsler spacetimes. The starting point of our construction is the assumption that physical fields are homogeneous sections of suitable bundles defined over (conic subbundles of) the tangent bundle of a Finsler spacetime. Using the assumption of homogeneity, we have constructed an equivalent description of fields as sections of bundles over the positive projective tangent bundle \(PTM^+\) instead. This step is crucial for a well-defined application of the variational principle, as it allows for variations with compact support within \(PTM^+\), which is not possible in the aforementioned approach using homogeneous sections over the tangent bundle. Within this framework, we studied the implications of general covariance, and derived the corresponding conserved energy-momentum distribution. As a particular example, we studied the kinetic gas.

Since the framework we propose is kept very general, it can be applied to a wide range of conceivable theories. The most natural class of fields to study, besides the kinetic gas, would be d-tensor fields. The latter provide a simple generalization of tensor fields on the spacetime manifold, which attain a dependence on directions in the tangent space, in addition to their dependence on spacetime. This additional dependence could be employed to model a velocity-dependent interaction between such fields with observers or particles. Such a dependence would be expected in an effective description of the quantum nature of spacetime, and leads to a modified dispersion relation for highly energetic particles, which could possibly be detected in observations. An ongoing effort is to extend our construction to a well defined notion of  spinors and spinor field theories on general Finsler spacetimes.

Another potential application of our proposed framework is to address the so far unexplained observations in cosmology. The well-known standard model of cosmology, coined $\Lambda$CDM model as it models 95\% of the matter content of the universe as dark energy \(\Lambda\) and cold dark matter (CDM), both of which have so far eluded direct detection, is under growing tension due to discrepancies between the measured values of the Hubble parameter in different observations. The correct interpretation of these observations depends crucially on understanding the propagation of electromagnetic radiation (and, with the advent of multi-messenger astronomy, also of gravitational waves, neutrinos and high-energetic cosmic particles). A modified propagation law, as it could arise for a field propagating on a Finsler spacetime background, could therefore provide alternative explanations that might resolve the observed tension.

%%%%%%%%%%%%%%%%%%%%%%%%%%%%%%%%%%%%%%%%%%%%%%%%%%%%%%%%%%%%%%%%%%%%%%%%%%%%%%%%%%%%%%%%%%%%%%%%%%
%%%%%%%%%%%%%%%%%%%%%%%%%%%%%%%%%%%%%%%%%%%%%%%%%%%%%%%%%%%%%%%%%%%%%%%%%%%%%%%%%%%%%%%%%%%%%%%%%%

%%%%%%%%%%%%%%%%%%%%%%%%%%%%%%%%%%%%%%%%%%%%%%%%%%%%%%%

\begin{acknowledgments}
C.P. was funded by the Deutsche Forschungsgemeinschaft (DFG, German Research Foundation) - Project Number 420243324. M.H. was supported by the Estonian Research Council grant PRG356 “Gauge Gravity” and by the European Regional Development Fund through the Center of Excellence TK133 “The Dark Side of the Universe”. The authors would like to acknowledge networking support by the COST Action QGMM (CA18108), supported by COST (European Cooperation in Science and Technology). Also, they would like to express their thanks to the anonymous JMP referee, for his/her useful comments and questions.

This article may be downloaded for personal use only. Any other use requires prior permission of the authors and AIP Publishing. This article appeared in the Journal of Mathematical Physics and may be found at https://aip.scitation.org/doi/10.1063/5.0065944.
\end{acknowledgments}

\section*{Data Availability Statement}
The study presented in this article is of purely theoretical and mathematical nature. All results and all sources on which these results are based are cited. Data sharing is not applicable to this article as no new data were created or analyzed in this study

%%%%%%%%%%%%%%%%%%%%%%%%%%%%%%%%%%%%%%%%%%%%%%%%%%%%%%%%%%%%%%%%%%%%%%%%%%%%%%%%%%%%%%%%%%%%%%%%%%
%%%%%%%%%%%%%%%%%%%%%%%%%%%%%%%%%%%%%%%%%%%%%%%%%%%%%%%%%%%%%%%%%%%%%%%%%%%%%%%%%%%%%%%%%%%%%%%%%%

\appendix

%%%%%%%%%%%%%%%%%%%%%%%%%%%%%%%%%%%%%%%%%%%%%%%%%%%%%%%%%%%%%%%%%%%%%%%%%%%%%%%%%%%%%%%%%%%%%%%%%%
%%%%%%%%%%%%%%%%%%%%%%%%%%%%%%%%%%%%%%%%%%%%%%%%%%%%%%%%%%%%%%%%%%%%%%%%%%%%%%%%%%%%%%%%%%%%%%%%%%

\section{Jet bundles and the coordinate-free calculus of variations}

\label{app:A}

In this appendix we briefly present the jet bundle formalism, which allows
for a coordinate-free description of calculus of variations, in terms of
differential forms; for more details, we mainly refer to the monograph \cite%
{Krupka-book}.

\subsection{Fibered manifolds and their jet prolongation}

A \emph{fibered manifold} is a triple $\left( Y,\pi ,X\right) ,$ where $X,Y $
are smooth manifolds with $\dim X=n,$ $\dim Y=n+m$ and $\pi :Y\rightarrow X $
is a surjective submersion. The level sets $Y_{x}=\pi ^{-1}(x)$ are called
the fibers of~$Y.$

Any fibered manifold admits an atlas consisting of fibered charts. These are
local charts $(V,\psi ),$ $\psi =(x^{A},y^{\sigma })$ such that there exists
a local chart $(U,\phi ),$ $\phi =(x^{A})$ on $X$, with $\pi (V)=U,$
in which $\pi $ is represented as $\pi :(x^{A},y^{\sigma })\mapsto (x^{A}).$

In particular, \textit{fiber bundles, }as understood in \cite%
{Palais}, are fibered manifolds that are locally trivial, i.e., in the
above, each $V$ is homeomorphic to a Cartesian product $U\times Z,$
where $Z$ is a manifold, called the typical fiber.

\bigskip

Assume, in the following, that $\left( Y,\pi ,X\right) $ is a fibered
manifold. \textit{Local sections} $\gamma :U\rightarrow Y$ (with $U\subset X$
open) are smooth maps such that $\pi \circ \gamma =id_{X};$ in a fibered
chart, they are represented as:%
\begin{equation*}
\gamma :(x^{A})\mapsto (x^{A},y^{\sigma }(x^{A})).
\end{equation*}%
We denote by $\Gamma (Y)$ the set of sections of $(Y,\pi ,X).$ In the
following, capital Latin indices $A,B,C,...$ will run from 0 to $n-1$ and
Greek indices $\sigma ,\mu ,\nu ,\rho ,...$ will run from $1$ to $m.$

\bigskip

\noindent \textbf{Physical interpretation.} In field theory, these manifolds and quantities are interpreted as follows:
\begin{itemize}
	\item The manifold $Y$ is called the \textit{configuration space}.
	\item The base manifold $X$ is typically (but not always) interpreted as spacetime; a notable exception to this rule is Finslerian field theory, where $X = PTM^{+}$ is the positively projectivized tangent bundle of the spacetime manifold $M$ (and the naturality of Lagrangians will be discussed with respect to $M$). In the following, we will reserve the notation $M$ for manifolds to be interpreted as spacetimes and denote by $X$ generic base manifolds.
	\item Sections $\gamma \in \Gamma(Y)$ are interpreted as \textit{fields}.
\end{itemize}
\bigskip

The \emph{jet bundle} $J^{r}Y=\left\{ J_{x}^{r}\gamma ~|~\gamma \in \Gamma
(Y),~x\in X\right\} $ is naturally equipped with an atlas consisting of
fibered charts $(V^{r},\psi ^{r}),$ $\psi ^{r}=(x^{A},y^{\sigma
},y_{~C_{1}}^{\sigma },...,y_{~C_{1}C_{2}...C_{r}}^{\sigma })$ on $J^{r}Y,$
induced by fibered charts $(V,\psi ),$ via
\begin{equation}
y_{~C_{1}...C_{k}}^{\sigma }(J_{x}^{r}\gamma )=\dfrac{\partial ^{k}y^{\sigma
}}{\partial x^{C_{1}}...\partial x^{C_{k}}}(x^{A}).
\end{equation}%
Any section of $Y$ is naturally prolonged into a section $J^{r}\gamma $ of $J^{r}Y$; in a chart $(V^{r},\psi ^{r}):$
\begin{equation*}
J^{r}\gamma :(x^{A})\mapsto \left( x^{A},y^{\sigma }(x^{A}),\dfrac{\partial
y^{\sigma }}{\partial x^{A}}(x^{B}),....,\dfrac{\partial ^{r}y^{\sigma }}{%
\partial x^{A_{1}}...\partial x_{r}^{A}}(x^{B})\right) .
\end{equation*}%
When referring to local expressions of geometric objects on $J^{r}Y$, we
always understand their expressions in fibered charts $(V^{r},\psi ^{r})$ as
above.

$J^{r}Y$ is a fibered manifold over all lower order jet bundles $J^{s}Y,$ $%
0\leq s<r$ (where $J^{0}Y:=Y$), with canonical projections
\begin{equation*}
\pi ^{r,s}:J^{r}Y\rightarrow J^{s}Y,~\ \ (x^{A},y^{\sigma
},y_{~C_{1}}^{\sigma },...,y_{~C_{1}C_{2}...C_{r}}^{\sigma })\mapsto
(x^{A},y^{\sigma },y_{~C_{1}}^{\sigma },...,y_{~C_{1}C_{2}...C_{s}}^{\sigma
}).
\end{equation*}
$J^{r}Y$ is also a fibered manifold over $X,$ with projection
\begin{equation*}
\pi ^{r}:J^{r}Y\rightarrow X,~~\ \ \ \ (x^{A},y^{\sigma },y_{~C_{1}}^{\sigma
},...,y_{~C_{1}C_{2}...C_{r}}^{\sigma })\mapsto (x^{A}).
\end{equation*}

\subsection{Horizontal and contact forms}

\label{app:horcont}

Let us introduce the following sets on $J^{r}Y$:
\begin{enumerate}
\item $\Omega _{k}(J^{r}Y)$, the set of differential $k$-forms defined over
open subsets $W^{r}\subset J^{r}Y$

\item $\Omega (J^{r}Y):=\underset{k\in \mathbb{N}}{\bigoplus }\Omega _{k}(J^{r}Y)$
the set of all differential forms over open subsets $W^{r}\subset J^{r}Y$;

\item $\mathcal{X}(J^{r}Y):=\Gamma (TJ^{r}Y)$ the module of vector fields on
$W^{r}\subset J^{r}Y$;

\item $\mathcal{F}(J^{r}Y)$, the set of all smooth functions $f:W\rightarrow
\mathbb{R}$ defined on open subsets $W\subset J^{r}Y$.
\end{enumerate}
A differential form $\rho \in \Omega _{k}(J^{r}Y)$ is $\pi ^{r}$\textit{\
-horizonta}l, if $\mathbf{i}_{\Xi }\rho =0$ whenever $\Xi \in \mathcal{X}%
(J^{r}Y)$ is $\pi ^{r}$-vertical (i.e., whenever $d\pi ^{r}(\Xi )=0$). In a
fibered chart, any $\pi ^{r}$-horizontal form is expressed as:%
\begin{equation}
\rho =\dfrac{1}{k!}\rho _{A_{1}A_{2}...A_{k}}dx^{A_{1}}\wedge dx^{A_{2}}\wedge ...\wedge dx^{A_{k}},
\label{horizontal_form}
\end{equation}%
where $\rho _{A_{1}A_{2}...A_{k}}$ are smooth functions of the coordinates $x^{A},y^{\sigma},y_{~C_{1}}^{\sigma },...,y_{~C_{1}C_{2}...C_{r}}^{\sigma}$ on $J^{r}Y$. Similarly, $\pi
^{r,s}$\textit{-horizontal} forms, $0\leq s\leq r$ are locally generated by $%
dx^{A},dy^{\sigma },...,dy_{~C_{1}...C_{s}}^{\sigma }$. A particular example
of horizontal forms are Lagrangians, which we define in the next subsection.

The \textit{horizontalization }operator is the unique\ morphism of exterior
algebras $h:\Omega ^{r}(Y)\rightarrow \Omega ^{r+1}(Y)$ such that, for any $%
f\in \mathcal{F}(J^{r}Y)$ and any fibered chart: $hf=f\circ \pi ^{r+1,r}~$and%
\begin{equation}
hdf=d_{A}fdx^{A},  \label{hdf}
\end{equation}%
where $d_{A}f:=\partial _{A}f+\dfrac{\partial f}{\partial y^{\sigma }}%
y_{~A}^{\sigma }+...\dfrac{\partial f}{\partial y_{~C_{1}...C_{r}}^{\sigma }}%
y_{~C_{1}...C_{r}A}^{\sigma }$ is the total derivative (of order $r+1$) with
respect to $x^{A}.$ On the natural basis 1-forms, it acts as:
\begin{equation}
hdx^{A}:=dx^{A},~\ hdy^{\sigma }=y_{~A}^{\sigma
}dx^{A},...,hdy_{~C_{1}...C_{k}}^{\sigma }=y_{~C_{1}...C_{k}A}^{\sigma
}dx^{A},\ ~\ \ k=\overline{1,r}.  \label{horizontalization_basis}
\end{equation}%
A useful property is the following. For any $f\in \mathcal{F}(J^{r}Y)$, $%
\gamma \in \Gamma (Y):$
\begin{equation}
\partial _{A}(f\circ J^{r}\gamma )=J^{r+1}\gamma ^{\ast }d_{A}f.
\label{property_h}
\end{equation}

\bigskip

A differential form $\rho \in \Omega (J^{r}Y)$ is a \textit{contact form }if
$J^{r}\gamma ^{\ast }\rho =0,$ $\forall \gamma \in \Gamma (Y)$. For
instance,
\begin{equation}
\theta ^{\sigma }=dy^{\sigma }-y_{~C}^{\sigma }dx^{C},~\ \ \theta
_{~A_{1}}^{\sigma }=dy_{~A_{1}}^{\sigma }-y_{~A_{1}C}^{\sigma
}dx^{C},...,\theta _{~A_{1}A_{2}...A_{r-1}}^{\sigma
}=dy_{~A_{1}A_{2}...A_{r-1}}^{\sigma }-y_{~A_{1}A_{2}...A_{r-1}C}^{\sigma
}dx^{C},  \label{contact basis}
\end{equation}%
are contact forms on a given chart domain $V^{r}\subset J^{r}Y,$ providing a
local basis $\{dx^{A},\theta ^{\sigma },....,\theta
_{~A_{1}...A_{r-1}}^{\sigma },dy_{~A_{1}...A_{r}}^{\sigma }\}$ of the module
$\Omega _{1}(J^{r}Y),$ called the \textit{contact basis.}

Raising to the next ``floor'' $J^{r+1}Y,$ any differential form can be
uniquely split as%
\begin{equation*}
\left( \pi ^{r+1,r}\right) ^{\ast }\rho =h\rho +p\rho ,
\end{equation*}%
where $p\rho $ is contact. Intuitively, $h\rho $ is what will survive of $%
\rho $ when pulled back to $X$ by prolonged sections $J^{r+1}\gamma ,$ where
$\gamma \in \Gamma (Y),$ while $p\rho $ becomes invisible: $J^{r+1}\gamma
^{\ast }(p\rho )=0.$

In particular, a $k$-form $\rho \in \Omega (J^{r}Y)$ is $1$\textit{-contact }%
if $\mathbf{i}_{\Xi }\rho $ is a $\pi ^{r}$-horizontal form whenever $\Xi
\in \mathcal{X}(J^{r}Y)$ is $\pi ^{r}$-vertical; in coordinates, 1-contact forms $\rho$ can be recognized by the fact that, in their expression in the contact basis, each term contains exactly one of the contact basis 1-forms $\theta^{\sigma},...,\theta^{\sigma}_{~A_{1} ...A_{r}}$ defined in (\ref{contact basis})).

A $\pi ^{r,0}$-horizontal, 1-contact $(n+1)$-form $\eta \in \Omega
_{n+1}^{r}Y$ is called a \textit{source form}. Locally, a source form is
expressed as:
\begin{equation}
\eta =\eta _{\sigma }\theta ^{\sigma }\wedge d^{n}x,  \label{source_form}
\end{equation}%
where $\eta _{\sigma }=\eta _{\sigma }(x^{A},y^{\mu
},....y_{~A_{1}...A_{r}}^{\mu }).$

\bigskip

\noindent \textbf{Fibered morphisms.}

An \textit{automorphism }of a fibered manifold $\left( Y,\pi ,X\right) $ is,
\cite{Krupka-book}, a diffeomorphism $\Phi :Y\rightarrow Y$ such that exists
a mapping $\phi \in \textrm{Diff}(X)$ with $\pi \circ \Phi =\phi \circ \pi $, i.e.,
the following diagram is commutative:%
\begin{equation}
	\xymatrix{Y \ar^{\Phi}[r] \ar_{\pi}[d] & Y \ar^{\pi}[d]\\
		X \ar^{\phi}[r] & X}
\end{equation}

%\begin{equation*}
%\begin{array}{ccc}
%Y & \overset{\Phi }{\longrightarrow } & Y \\
%_{\pi }\downarrow &  & \downarrow _{\pi } \\
%X & \overset{\phi }{\longrightarrow } & X%
%\end{array}%
%\end{equation*}%
In this case, $\Phi $ is said to \textit{cover} $\phi .$ In coordinates,
these must be of the form:%
\begin{eqnarray}
\ \phi &:&(x^{A})\mapsto \tilde{x}^{A}(x^{B})  \label{fibered_morphism_base}
\\
\Phi &:&(x^{A},y^{\sigma })\mapsto (\tilde{x}^{A}(x^{B}),\tilde{y}^{\sigma
}(x^{B},y^{\mu })).  \label{fibered_morphism}
\end{eqnarray}

The automorphism $\Phi $ is called \textit{strict} if $\phi =id_{X}.$

Any generator $\Xi $ of a 1-parameter group $\left\{ \Phi _{\varepsilon
}\right\} $ of automorphisms of $Y$ is a $\pi $-projectable\textit{\ }vector
field, i.e, $\pi _{\ast }\Xi $ is a well defined vector field on $X$; in a
fibered chart, projectable vector fields are represented as:%
\begin{equation}
\Xi =\xi ^{A}(x^{B})\partial _{A}+\Xi ^{\sigma }(x^{B},y^{\mu }) \tfrac{\partial}{\partial y^\sigma}.
\end{equation}%
In particular, 1-parameter groups of strict automorphisms are generated by $%
\pi $-\textit{vertical} vector fields $\Xi =\Xi ^{\sigma }(x^{B},y^{\mu
})\tfrac{\partial}{\partial y^\sigma}.$

Automorphisms $\Phi :Y\rightarrow Y$ are prolonged into automorphisms of $%
J^{r}Y$ as: $J^{r}\Phi (J_{x}^{r}\gamma ):=J_{\phi (x)}^{r}(\Phi \circ
\gamma \circ \phi ^{-1}).$ The generator of the 1-parameter group $%
\left\{ J^{r}\Phi _{\varepsilon }\right\} ,$ with $\Phi _{\varepsilon }$ as
above, is called the $r$-th prolongation of the vector field $\Xi $ and
denoted by $J^{r}\Xi .$ In particular, for $r=1,$ this is given by:%
\begin{equation*}
J^{1}\Xi =\xi ^{A}\partial _{A}+\Xi ^{\sigma }\partial _{\sigma }+\Xi
_{~A}^{\sigma }\tfrac{\partial }{\partial y_{~A}^{\sigma }},~\ \ \ \ \ \Xi
_{~A}^{\sigma }=d_{A}\Xi ^{\sigma }-y_{~A}^{\sigma }\xi ^{A}.\
\end{equation*}

\subsection{Lagrangians and first variation formula}

\label{app:Lagrangians}

A \textit{Lagrangian} is defined as a $\pi^{r}$-horizontal form $\lambda \in
\Omega _{n}^{r}Y$ of degree $n=\dim X$\textit{; }locally,
\begin{equation}
\lambda =\mathcal{L}d^{n}x,~\ \ \ \ \ \ \mathcal{L=L}(x^{A},y^{\sigma
},...,y_{A_{1}...A_{r}}^{\sigma }),  \label{general Lagrangian}
\end{equation}%
where $d^{n}x:=dx^{1}\wedge ...\wedge dx^{n}.$

By a \textit{piece} $D\subset X,$ we understand, \cite{Krupka-book}, a
compact $n$-dimensional submanifold with boundary of $X$. The \textit{%
action\ }attached to the Lagrangian (\ref{general Lagrangian}) and to a
piece $D\subset X$ is the function $S_{D}:\Gamma (Y)\rightarrow \mathbb{R},$
given by:
\begin{equation*}
S_{D}(\gamma )=\underset{D}{\int }J^{r}\gamma ^{\ast }\lambda .
\end{equation*}%
Consider an arbitrary 1-parameter group $\left\{ \Phi _{\varepsilon
}\right\} $ of automorphisms of $Y,$ with ($\pi$-projectable) generator $\Xi \in \mathcal{X}(Y).$
This will induce a deformation $\gamma \mapsto \gamma _{\varepsilon }:=\Phi
_{\varepsilon }\circ \gamma \circ \phi _{\varepsilon }^{-1}$ of sections $%
\gamma \in \Gamma (Y):$%

\begin{equation}
	\xymatrix{Y \ar^{\Phi_{\epsilon}}[r] & Y\\
		X \ar^{\gamma}[u] & X \ar_{\phi_{\epsilon}^{-1}}[l] \ar_{\gamma_{\epsilon}}[u]}
\end{equation}

%\begin{equation*}
%\begin{array}{ccc}
%Y & \overset{\Phi _{\varepsilon }}{\longrightarrow } & Y \\
%_{\gamma }\uparrow  &  & \uparrow _{\gamma _{\varepsilon }} \\
%X & \overset{\phi _{\varepsilon }^{-1}}{\longleftarrow } & X%
%\end{array}%
%\end{equation*}%

The \textit{variation} $\delta _{\Xi }S_{D}(\gamma ):=\dfrac{d}{d\varepsilon
}|_{\varepsilon =0}S_{\phi _{\varepsilon }(D)}(\gamma _{\varepsilon })$ is
then expressed as the Lie derivative:%
\begin{equation}
\delta _{\Xi }S_{D}(\gamma )=\underset{D}{\int }J^{r}\gamma ^{\ast }%
\mathfrak{L}_{J^{r}\Xi }\lambda .  \label{variation}
\end{equation}%
A section $\gamma \in \Gamma (Y)$ is a \textit{critical section\ }for $S,$
if for any compact $D\subset X$ and for any $\pi$-projectable $\Xi \in \mathcal{X}(Y)$ such
that $supp(\Xi \circ \gamma )\subset D$, there holds: $\delta_{\Xi} S_{D}(\gamma
)=0.$

\bigskip

For any Lagrangian $\lambda \in \Omega _{n}(J^{r}Y)$ and any $\Xi \in
\mathcal{X}(Y),$ there holds the \textit{first variation formula:}%
\begin{equation}
J^{r}\gamma ^{\ast }(\mathfrak{L}_{J^{r}\Xi }\lambda )=J^{s}\gamma ^{\ast }
\mathbf{i}_{J^{s}\Xi }\mathcal{E}_{\lambda }-J^{s}\gamma ^{\ast }d\mathcal{%
\ J}^{\Xi },  \label{first_variation_general}
\end{equation}
where:

\begin{itemize}
\item $\mathcal{E}_{\lambda }\in \Omega _{n+1}(J^{s}Y)$ is a source form of
order $s\leq 2r,$ called the \textit{Euler-Lagrange form}\footnote{The coordinate-free definition of the Euler-Lagrange form associated to a
Lagrangian $\lambda $ employs the notion of Lepage equivalent of $\lambda ,$
see \cite{Krupka-book}, p. 122 and 125.; yet, for our purposes, the precise expressions of Lepage equivalents of our Lagrangians will not be necessary.}; locally, if $\lambda =\mathcal{L}d^{n}x,$
then:
\begin{equation*}
\mathcal{E}_{\lambda }=E_{\sigma }\theta ^{\sigma }\wedge d^{n}x,
\end{equation*}
with:\textit{\ }
\begin{equation}
E_{\sigma }=\dfrac{\delta \mathcal{L}}{\delta y^{\sigma }}=\dfrac{\partial
\mathcal{L}}{\partial y^{\sigma }}-d_{A}\dfrac{\partial \mathcal{L}}{
\partial y_{~A}^{\sigma }}+...+(-1)^{r}d_{A_{1}}...d_{A_{r}}\dfrac{\partial
\mathcal{L}}{\partial y_{~A_{1}...A_{r}}^{\sigma }}.  \label{EL expressions}
\end{equation}
The section $\gamma \in \Gamma (Y)$ is critical for $\lambda $ if and only
if $E_{\sigma }\circ J^{s}\gamma =0.$

\item The $\left( n-1\right) $-form $\mathcal{J}^{\Xi }\in \Omega
_{n-1}(J^{s}Y)$ is called the \textit{Noether current }associated with $%
\lambda $ and to the vector field $\Xi .$ If $\Xi $ is a symmetry generator
for $\lambda ,$ i.e., if $\mathfrak{L}_{J^{r}\Xi }\lambda =0,$ then, \textit{%
\ Noether's first theorem} states that the Noether current is conserved
along critical sections:
\begin{equation}
J^{s}\gamma ^{\ast }d\mathcal{J}^{\Xi }\approx 0,  \label{J_conservation}
\end{equation}
where $\approx $ denotes equality on-shell, i.e., for critical sections $%
\gamma .$
\end{itemize}

The Euler-Lagrange form of $\lambda $ is unique, while the Noether current $%
\mathcal{J}^{\Xi }$ is only unique up to an exact form $d\rho $.

In integral form, the first variation formula reads:%
\begin{equation}
\underset{D}{\int }J^{r}\gamma ^{\ast }(\mathfrak{L}_{J^{r}\Xi }\lambda )=
\underset{D}{\int }J^{s}\gamma ^{\ast }\mathbf{i}_{J^{s}\Xi }\mathcal{E}
(\lambda )-\underset{\partial D}{\int }J^{s}\gamma ^{\ast }\mathcal{J}^{\Xi
}.  \label{first_variation_integral}
\end{equation}

\begin{remark}
\begin{enumerate}
\item The fact that $\mathcal{E}_{\lambda }=E_{\sigma }\theta ^{\sigma
}\wedge d^{n}x$ is a source form implies that locally, only the $\partial
_{A}$ and $\tfrac{\partial}{\partial y^\sigma}$-components of $J^{r}\Xi $ will contribute to
$\textbf{i}_{J^{s}\Xi }\mathcal{E}_{\lambda }$ (i.e., higher order
components of $J^{r}\Xi $ will not contribute to it):
\begin{equation}
\mathbf{i}_{J^{s}\Xi }\mathcal{E}_{\lambda }=(\tilde{\Xi}^{\sigma }E_{\sigma
})d^{n}x,~\ \ ~\ \ \tilde{\Xi}^{\sigma }=\Xi ^{\sigma }-y_{~C}^{\sigma }\xi
^{C}.  \label{Csi_tilde}
\end{equation}%
The functions $(\tilde{\Xi}^{\sigma }\circ J^{s}\gamma ):X\rightarrow
\mathbb{R}$ are commonly denoted in the literature by $\delta y^{\sigma }.$

\item In order to identify the Euler-Lagrange form, it is sufficient to use $%
\pi $\textit{-vertica}l variation vector fields $\Xi \in \mathcal{X}(Y).$
Yet, general vector fields are needed in discussing general covariance and
its consequence, energy-momentum conservation.
\end{enumerate}
\end{remark}

\noindent \textbf{Natural bundles and natural Lagrangians.}\label{app:natbun}

Let $\mathcal{M}_{n}$ denote the category of smooth $n$-dimensional
manifolds, with smooth embeddings as morphisms and $\mathcal{FB},$ the
category of smooth fiber bundles, whose morphisms are smooth fibered
morphisms.

A \textit{natural bundle functor over} $n$\textit{-manifolds} is, \cite%
{Palais}, a functor $\mathfrak{F}:\mathcal{M}_{n}\rightarrow \mathcal{FB},$
such that:

\begin{enumerate}
\item For each $M\in Ob(\mathcal{M}_{n}),$ $\mathfrak{F}(M)$ is a
 fiber bundle over $M$;

\item For each embedding $\alpha _{0}:M\rightarrow M^{\prime }\in Morf(
\mathcal{M}_{n}),$ the fibered manifold morphism $\mathfrak{F}(\alpha _{0}):
\mathfrak{F}(M)\rightarrow \mathfrak{F}(M^{\prime })$ covers $\alpha _{0}.$
\end{enumerate}

If $Y=\mathfrak{F}(M),$ then any automorphism $\phi $ of $M\in Ob(\mathcal{M}%
_{n})$ admits a canonical (or natural) lift $\Phi :=\mathfrak{F}(\phi )$ to $%
Y$. These natural lifts encode the transformations of fields - more
precisely, their local expressions are identical to transition functions on $%
Y$ (see, e.g., \cite{Giachetta}, \cite{Fatibene}). For instance, if $Y$ is a
bundle of tensors of over $M,$ then the canonical lift $\Phi =\mathfrak{F}%
(\phi )$ of $\phi \in Diff(M)$ is given by pullback/pushforward.

Passing to infinitesimal generators, any vector field $\xi \in \mathcal{X}%
(M) $ admits a canonical lift $\Xi :=\mathfrak{F}(\xi )\in \mathcal{X}(Y)$;
in a fibered chart, the components $\Xi ^{\sigma }$ can always be expressed
in terms of the components $\xi ^{i}$ of $\xi $ and a finite number of
partial derivatives thereof, \cite{Gotay}.

For example, in the case of the bundle of tensors $Y=T_{q}^{p}M$ of type $%
(p,q)$ over $M,$ one obtains $\Xi =\xi ^{i}\partial _{i}+\Xi
_{j_{1}...j_{q}}^{i_{1}...i_{p}}\dfrac{\partial }{\partial
	y_{j_{1}...j_{q}}^{i_{1}...i_{p}}},$ where:%
\[
\Xi _{j_{1}...j_{q}}^{i_{1}...i_{p}}=\xi
_{~,h}^{i_{1}}y_{j_{1}...j_{q}}^{hi_{2}...i_{p}}+...\xi
_{~,h}^{i_{p}}y_{j_{1}...j_{q}}^{i_{1}...i_{p-1}h}-\xi
_{~,j_{1}}^{h}y_{hj_{2}...j_{q}}^{i_{1}...i_{p}}-\xi
_{~,j_{1}}^{h}y_{j_{1}...j_{q-1}h}^{i_{1}...i_{p}}.
\]

\bigskip

A globally defined Lagrangian $\lambda \in \Omega _{n}(J^{r}\mathfrak{F}(M))$
is called \textit{natural}, or \textit{generally covariant}, if it is
invariant under canonical lifts of arbitrary diffeomorphisms of spacetime,
i.e., $J^{r}\mathfrak{F}(\phi )^{\ast }\lambda =\lambda $ for all $\phi \in
Diff(M),$ \cite{Fatibene}. Using the formal
similarity between lifts of active diffeomorphisms $\phi \in Diff(M)$ and
(manifold-induced) fibered coordinate changes on $\mathfrak{F}(M),$
naturality amounts to the fact that $\lambda $ must be invariant to \textit{any} such coordinate changes (defined on \textit{any} manifold $\mathfrak{F}(M),$
where $M\in Ob(\mathcal{M}_{n})$). In terms of infinitesimal generators, this
reads:
\begin{equation}
\mathfrak{L}_{J^{r}\mathfrak{F}(\xi )}\lambda =0,
\label{general covariance condition}
\end{equation}
for all $\xi \in \mathcal{X}(M).$
General covariance gives rise to a notion of energy-momentum tensor, \cite{Gotay}.

\bibliographystyle{unsrturl}
\bibliography{MathFound}

\end{document}